\documentclass{CSML}

\pdfoutput=1

\usepackage{lastpage}
\def\dOi{13(4:27)2017}
\lmcsheading%
{\dOi}
{1--\pageref{LastPage}}
{}
{}
{Jan.~13, 2016}
{Dec.~20, 2017}
{}

\usepackage{amsmath}			% \text
\usepackage{amssymb}			% \Vdash
\usepackage{mathpartir}			% for inductive rules
\usepackage{xspace}				% \xspace
\usepackage{comment}
\usepackage{longtable}
\usepackage{color}				% \color and colors
\usepackage{ifthen}				% \ifthenelse, \equals
\usepackage[hidelinks]{hyperref}

\newcommand{\es}{\ensuremath{\emptyset}}
\newcommand{\set}[1]{\ensuremath{\{#1\}}}
\newcommand{\setbar}{\ \mid\ }
\newcommand{\bnfis}{\ensuremath{\ \ ::=\ \ }}
\newcommand{\bnfbar}{\ensuremath{\ \ |\ \ }}

\newcommand{\pol}[1]{\ensuremath{\tilde{#1}}\xspace}

\newcommand{\domain}[1]{\ensuremath{\mathsf{domain}(#1)}}

\newcommand{\eqdef}{\ensuremath{\ \ \stackrel{\mathtt{def}}{=}\ \ }}

\newcommand{\tree}[3][]{
		\ifthenelse{\equal{#1}{}}{
		}{
			#1\
		}
		\displaystyle\frac{
			\begin{array}{c}
				#2
			\end{array}
		}{
			#3
		}
}

\newcommand{\suchthat}{\ensuremath{:}}
\usepackage{multirow}

\newcommand{\figref}[1]{Figure~\ref{fig:#1}}
\newcommand{\secref}[1]{Section~\ref{sec:#1}}
\newcommand{\exref}[1]{Example~\ref{ex:#1}}
\newcommand{\defref}[1]{Definition~\ref{def:#1}}
\newcommand{\thmref}[1]{Theorem~\ref{thm:#1}}
\newcommand{\lemref}[1]{Lemma~\ref{lem:#1}}
\newcommand{\corref}[1]{Corollary~\ref{cor:#1}}
\newcommand{\propref}[1]{Proposition~\ref{cor:#1}}
\newcommand{\appref}[1]{Appendix~\ref{app:#1}}

\newcommand{\Pcalc}{Privacy calculus\xspace}

\newcommand{\dual}[1]{\ensuremath{\overline{#1}}}

\newcommand{\id}{\ensuremath{\mathsf{id}}\xspace}
\newcommand{\hid}{\ensuremath{\_}\xspace}
\newcommand{\ii}{\ensuremath{\iota}\xspace}
\newcommand{\con}{\ensuremath{\delta}\xspace}

\newcommand{\pdata}[2]{\ensuremath{#1\, {\text{\scriptsize $\otimes$}}\, #2}}

\newcommand{\G}{\ensuremath{\mathsf{G}}\xspace}
\newcommand{\fg}[1]{\mathtt{fg}(#1)}

\newcommand{\subst}[2]{\ensuremath{\{^{#1}/_{#2}\}}\xspace}
\newcommand{\fn}[1]{\ensuremath{\mathsf{fn}(#1)}\xspace}
\newcommand{\bn}[1]{\ensuremath{\mathsf{bn}(#1)}\xspace}
\newcommand{\fv}[1]{\ensuremath{\mathsf{fv}(#1)}\xspace}
\newcommand{\bv}[1]{\ensuremath{\mathsf{bv}(#1)}\xspace}

\newcommand{\psep}{\ensuremath{.\,}}
\newcommand{\inputvar}[1]{\ensuremath{(#1)}}
\newcommand{\outputvar}[1]{\ensuremath{\langle#1\rangle}}
\newcommand{\inpprefix}[2]{\ensuremath{#1?\inputvar{#2}}}
\newcommand{\outprefix}[2]{\ensuremath{#1!\outputvar{#2}}}

\newcommand{\inact}{\ensuremath{\mathbf{0}}\xspace}

\newcommand{\out}[2]{\outprefix{#1}{#2} \psep}
\newcommand{\inp}[2]{\inpprefix{#1}{#2} \psep}
\newcommand{\eout}[1]{\dual{#1} \psep}
\newcommand{\einp}[1]{#1 \psep}

\newcommand{\sel}[2]{\ensuremath{#1 \triangleleft #2 \psep}}
\newcommand{\bra}[2]{\ensuremath{#1 \triangleright \set{#2}}}

\newcommand{\Par}{\ensuremath{\,|\,}}
\newcommand{\repl}{\ensuremath{*}}

\newcommand{\newn}[1]{\ensuremath{(\nu\ #1)}}
\newcommand{\newnp}[2]{\newn{#1}\ensuremath{(#2)}}

\newcommand{\If}{\ensuremath{\mathtt{if}}\xspace}
\newcommand{\Then}{\ensuremath{\mathtt{then}}\xspace}
\newcommand{\Else}{\ensuremath{\mathtt{else}}\xspace}

\newcommand{\match}[2]{\ensuremath{#1 = #2}}
\newcommand{\ifelse}[4]{\ensuremath{\If\ \match{#1}{#2}\ \Then\ #3\ \Else\ #4}}
\newcommand{\IfElse}[3]{\ensuremath{\If\ #1\ \Then\ #2\ \Else\ #3}}

\newcommand{\store}[3]{\ensuremath{\dual{#1} \triangleright[\pdata{#2}{#3}]}}

\newcommand{\group}[2]{\ensuremath{#1[#2]}\xspace}
\newcommand{\Spar}{\ensuremath{|\!|}}

\newcommand{\scong}{\ensuremath{\equiv}\xspace}
\newcommand{\map}[1]{\ensuremath{(\!|#1|\!)}}

\newcommand{\rdl}{\ensuremath{\mathsf{rd}}\xspace}
\newcommand{\wrl}{\ensuremath{\mathsf{wr}}\xspace}
\newcommand{\ok}{\ensuremath{\mathsf{ok}}\xspace}
\newcommand{\fail}{\ensuremath{\mathsf{fail}}\xspace}

\newcommand{\actout}[2]{\ensuremath{\outprefix{#1}{#2}}\xspace}
\newcommand{\actinp}[2]{\ensuremath{\inpprefix{#1}{#2}}\xspace}

\newcommand{\trans}[1]{\ensuremath{\stackrel{#1}{\longrightarrow}}}
\newcommand{\Trans}[1]{\ensuremath{\stackrel{#1}{\Longrightarrow}}}

\newcommand{\actdual}{\ensuremath{\ \mathsf{dual}\ }}

\newcommand{\lts}[1]{\ensuremath{[\text{\footnotesize$\mathsf{#1}$}]}\xspace}

\newcommand{\LOut}{\lts{Out}}
\newcommand{\LInp}{\lts{Inp}}
\newcommand{\LSOut}{\lts{SOut}}
\newcommand{\LSInp}{\lts{SInp}}
\newcommand{\LTr}{\lts{True}}
\newcommand{\LFl}{\lts{Fls}}
\newcommand{\LTau}{\lts{Tau}}
\newcommand{\LLPar}{\lts{LPar}}
\newcommand{\LRPar}{\lts{RPar}}
\newcommand{\LScp}{\lts{Scp}}
\newcommand{\LRes}{\lts{Res}}
\newcommand{\LRepl}{\lts{Repl}}
\newcommand{\Alpha}{\lts{Alpha}}

\newcommand{\srule}[1]{\ensuremath{[\text{\footnotesize$\mathsf{S#1}$}]}\xspace}

\newcommand{\STau}{\srule{Tau}}
\newcommand{\SLPar}{\srule{LPar}}
\newcommand{\SRPar}{\srule{RPar}}
\newcommand{\SScp}{\srule{Scp}}
\newcommand{\SRes}{\srule{Res}}
\newcommand{\SPGr}{\srule{PGr}}
\newcommand{\SGr}{\srule{Gr}}

\newcommand{\gtype}[2]{\ensuremath{#1[#2]}\xspace}

\newcommand{\privatet}[2]{\ensuremath{\mathsf{#1}[#2]}\xspace}

\newcommand{\purposet}[2]{\ensuremath{\mathsf{#1}[#2]}\xspace}

\newcommand{\ptuple}[2]{\ensuremath{\langle#1, #2\rangle}\xspace}

\newcommand{\prv}{\ensuremath{\mathsf{t}}\xspace}
\newcommand{\prp}{\ensuremath{\mathsf{p}}\xspace}
\newcommand{\g}{\ensuremath{\mathsf{g}}\xspace}

\newcommand{\cat}{,}

\newcommand{\proves}{\vdash}
\newcommand{\hastype}{\triangleright}

\newcommand{\Ga}{\ensuremath{\Gamma}\xspace}
\newcommand{\De}{\ensuremath{\Delta}\xspace}
\newcommand{\La}{\ensuremath{\Lambda}\xspace}
\newcommand{\Int}{\ensuremath{\Theta}\xspace}

\newcommand{\Inthier}{\ensuremath{\theta}\xspace}
\newcommand{\hiert}[2]{\ensuremath{#1[#2]}\xspace}

\newcommand{\UDe}[1]{\ensuremath{\De^{U}_{#1}}\xspace}
\newcommand{\RDe}[1]{\ensuremath{\De^{R}_{#1}}\xspace}
\newcommand{\ReplDe}{\ensuremath{\De^{\repl}}\xspace}

\newcommand{\trulefont}[1]{\ensuremath{[\text{\footnotesize$\mathsf{T#1}$}]}\xspace}

\newcommand{\TName}{\trulefont{Name}}
\newcommand{\TCons}{\trulefont{Cons}}
\newcommand{\TPd}{\trulefont{Pdata}}

\newcommand{\TId}{\trulefont{Id}}
\newcommand{\TUse}{\trulefont{Use}}
\newcommand{\TEqP}{\trulefont{EqP}}
\newcommand{\TEqA}{\trulefont{EqA}}

\newcommand{\TInp}{\trulefont{Inp}}
\newcommand{\TOut}{\trulefont{Out}}
\newcommand{\TSt}{\trulefont{St}}

\newcommand{\TInact}{\trulefont{Inact}}
\newcommand{\TPar}{\trulefont{Par}}
\newcommand{\TRes}{\trulefont{Res}}
\newcommand{\TRepl}{\trulefont{Repl}}
\newcommand{\TIf}{\trulefont{If}}

\newcommand{\TSPar}{\trulefont{SPar}}
\newcommand{\TSRes}{\trulefont{SRes}}
\newcommand{\TGr}{\trulefont{Gr}}
\newcommand{\TSGr}{\trulefont{SGr}}

\newcommand{\permfont}[1]{\ensuremath{\mathsf{#1}}\xspace}

\newcommand{\perm}{\permfont{prm}}

\newcommand{\rd}{\permfont{read}}
\newcommand{\up}{\permfont{update}}
\newcommand{\rf}{\permfont{reference}}

\newcommand{\diss}[2]{\ensuremath{\permfont{disseminate}\ #1\ #2}\xspace}
\newcommand{\ndiss}[1]{\ensuremath{\permfont{no\ dissemination}\langle#1\rangle}\xspace}

\newcommand{\str}{\permfont{store}}
\newcommand{\idp}{\permfont{readId}}

\newcommand{\idf}[1]{\ensuremath{\permfont{identify}\set{#1}}\xspace}
\newcommand{\use}[1]{\ensuremath{\permfont{usage}\set{#1}}\xspace}

\newcommand{\aggr}{\permfont{aggregate}}

\newcommand{\kind}{\permfont{kind}}
\newcommand{\discl}{\permfont{disclosure}}
\newcommand{\conf}{\permfont{confidential}}
\newcommand{\sens}{\permfont{sensitive}}

\newcommand{\infinite}{\ensuremath{\omega}\xspace}

\newcommand{\identify}[1]{\ensuremath{\mathsf{idperm}(#1)}\xspace}

\newcommand{\haspolicy}{\ensuremath{\gg}}

\newcommand{\role}[3]{#1\set{#2}[#3]}
\newcommand{\lrole}[2]{#1\set{#2}}

\newcommand{\Policy}{\ensuremath{\mathcal{P}}\xspace}
\newcommand{\Hpol}{\ensuremath{\mathsf{H}}\xspace}

\newcommand{\groups}[1]{\ensuremath{\mathsf{groups}(#1)}}
\newcommand{\permis}[1]{\ensuremath{\mathsf{perms}(#1)}}
\newcommand{\Pol}{\Policy}

\newcommand{\lenv}{\ensuremath{\preceq}\xspace}
\newcommand{\sat}{\ensuremath{\Vdash}\xspace}

\newcommand{\countLink}[1]{\ensuremath{\mathsf{countRef}(#1)}}

\newcommand{\groupfont}[1]{\ensuremath{\mathsf{#1}}\xspace}

\newcommand{\Hospital}{\groupfont{Hospital}}
\newcommand{\Nurse}{\groupfont{Nurse}}
\newcommand{\Doctor}{\groupfont{Doctor}}
\newcommand{\DBase}{\groupfont{DBase}}

\newcommand{\treatment}{\groupfont{treatment}}
\newcommand{\medication}{\pdatafont{medication}}

\newcommand{\dna}{\pdatafont{dna}}
\newcommand{\disease}{\pdatafont{disease}}

\newcommand{\forensics}{\groupfont{forensics\ lab}}
\newcommand{\john}{\groupfont{john}}

\newcommand{\pdatafont}[1]{\ensuremath{\mathsf{#1}}\xspace}

\newcommand{\patientd}{\pdatafont{patient\_data}}
\newcommand{\vehicled}{\pdatafont{vehicle\_data}}
\newcommand{\crime}{\pdatafont{crime}}
\newcommand{\diagnosis}{\pdatafont{diagnosis}}

\newcommand{\SpeedCheck}{\groupfont{SpeedCheck}}
\newcommand{\Speedometer}{\groupfont{Speedometer}}
\newcommand{\SpeedAuthority}{\groupfont{SpeedAuthority}}
\newcommand{\Network}{\groupfont{Network}}

\newcommand{\trcamera}{\pdatafont{trafficCam}}
\newcommand{\slimit}{\mathit{overLim}}
\newcommand{\speedlimit}{\pdatafont{speedLim}}
\newcommand{\limit}{\pdatafont{Limit}}
\newcommand{\speed}{\pdatafont{speed}}
\newcommand{\Car}{\groupfont{Car}}
\newcommand{\reg}{\pdatafont{reg}}
\newcommand{\SpeedControl}{\groupfont{SpeedControl}}
\newcommand{\SCSystem}{\groupfont{SCSystem}}
\newcommand{\System}{\groupfont{System}}
\newcommand{\Auth}{\groupfont{Auth}}

\newcommand{\CarReg}{\pdatafont{CarReg}}
\newcommand{\CarSpeed}{\pdatafont{CarSpeed}}
\newcommand{\DriverReg}{\pdatafont{DriverReg}}

\newcommand{\RegNum}{\pdatafont{RegNum}}
\newcommand{\Speed}{\pdatafont{Speed}}

\newcommand{\OBE}{\groupfont{OBE}}
\newcommand{\ETP}{\groupfont{ETP}}
\newcommand{\SC}{\groupfont{SC}}
\newcommand{\GPS}{\groupfont{GPS}}
\newcommand{\PA}{\groupfont{PA}}

\newcommand{\Loc}{\pdatafont{loc}}
\newcommand{\Fee}{\pdatafont{fee}}
\newcommand{\spotcheck}{\pdatafont{spotCheck}}
\newcommand{\computefee}{\pdatafont{computeFee}}

\newcommand{\dk}[1]{{\color{blue} #1}}

\title{Privacy by typing in the $\pi$-calculus}

\author[D.~Kouzapas]{Dimitrios Kouzapas\rsuper a}
\address{{\lsuper a}Department of Computing Science, University of Glasgow}
\email{Dimitrios.Kouzapas@gla.ac.uk}
\author[A.~Philippou]{Anna Philippou\rsuper b}
\address{{\lsuper b}Department of Computer Science, University of Cyprus}
\email{annap@cs.ucy.ac.cy}

\begin{document}

	\maketitle

\begin{abstract}
In this paper we propose a formal framework for studying privacy
in information systems.
The proposal follows a two-axes schema where the first
axis considers privacy as a taxonomy of rights
and the second axis involves the ways an information system stores
and manipulates information.
We develop a correspondence between the schema above and
an associated model of computation.
In particular, we propose the \Pcalc, a calculus
based on the $\pi$-calculus with groups
extended with constructs for reasoning
about private data.
The privacy requirements of an information system
are captured via a privacy policy language.
The correspondence between the privacy model and the \Pcalc semantics
is established using a type system for the calculus and a satisfiability
definition between types and privacy policies.
We deploy a type preservation theorem to show that a system respects a policy and it is safe
if the typing of the system satisfies the policy.
We illustrate our methodology
via analysis of two use cases: a privacy-aware scheme for electronic traffic pricing
and a privacy-preserving technique for speed-limit enforcement.
\end{abstract}

\section{Introduction}

The notion of privacy is a fundamental notion
for society and, as such, it has been an object of
study within various %scientific
disciplines such as philosophy, politics, law, and culture.
In general terms, an analysis of privacy reveals a dynamic concept
strongly dependent on cultural norms and evolution.
%
%Recently, its importance is becoming increasingly
%pronounced as the technological advances
%and the associated widespread accessibility of personal information
%is redefining the very essence of the term \emph{privacy}.
%
Society today is evolving steadily into an information era
where computer systems are required to aggregate
and handle huge volumes of private data.
Interaction of individuals with such systems
is expected to reveal new limits
of tolerance towards what is considered private which
in turn will reveal new threats to individual privacy.
But, fortunately, along with the introduction of new
challenges for privacy, technology can also offer
solutions to these new challenges.

\subsection{A Formal methods approach to privacy}
While the technological advances and the associated widespread
aggregation of private data are rendering the need for developing
systems that preserve individual's privacy %and forbidding
%{\em mishandling} of personal information
increasingly important,
%
%The specification of a system that aggregates private
%data should forbid {\em mishandling} in order to
%subsequently achieve the protection of the individual.
%Unfortunately today,
the established techniques for
providing guarantees that a system handles information
in a privacy-respecting manner are reported as partial and unsatisfactory.
%In simple words, today you should always doubt whether a program
%respects your privacy or not.

To this effect, the motivation of this paper is to address this challenge
by developing
formal frameworks for reasoning about privacy-related concepts.
Such frameworks may provide solid foundations
for understanding the notion of privacy and
allow to formally model and study privacy-related
situations. Rigorous analysis of privacy is expected to give rise to numerous practical
frameworks and techniques for developing sound systems
with respect to privacy properties.
More specifically, a formal approach to privacy may lead
to the development of programming semantics and analysis
tools for developing privacy-respecting code. Tools may
include programming languages, compilers and interpreters,
monitors, and model checkers. Simultaneously,  we envision that existing techniques,
like privacy by design~\cite{PrivbyDes}, and proof carrying code~\cite{Necula97},
may benefit from a formal description of privacy for information systems
and offer new and powerful  results.

The main objective of this paper is to develop a type-system
method for ensuring that a privacy specification, described
as a privacy policy, is satisfied by a computational process.
At the same time, the paper aims to set the foundations of a
methodology that will lead to further research on privacy
in information systems. For this reason we define the notion
of a privacy model as an abstract
model that describes the requirements of privacy in different scenarios
and we establish a correspondence between the
privacy model and the different abstraction levels inside
a computational framework: we show how this privacy model
is expressed in terms of programming semantics and in terms of
specifying and checking a privacy specification against a program.
%We believe that a static

\subsection{Contribution}
More concretely the contributions of this paper are:

\begin{description}
	\item[Privacy Model]
			In \secref{methodology} we identify a privacy model based
			on literature by legal scholars~\cite{solove06:privacy,Tschantz:privacy}.
			The model is based on a taxonomy of privacy violations that
			occur when handling private information.

	\item[Syntax and Semantics]
			In \secref{calculus} we propose a variant of the $\pi$-calculus with groups~\cite{Cardelli:secrecy}
			to develop a set of semantics that can capture the basic notions of the model we set.

	\item[Privacy Policy]
			We formally describe a privacy requirement as an object of a privacy policy language.
			The language is expressed in terms of a formal syntax defined in \secref{policy}.

	\item[Policy Satisfiability]
			The main technical contribution of the
			paper is a correspondence between the privacy policy and
			$\pi$-calculus systems using techniques from the type-system literature
			for the $\pi$-calculus.
			The correspondence is expressed via a satisfiability definition (\defref{policy_sat}),
			where a well-typed system satisfies a privacy policy if the type of the
			system satisfies the policy.

	\item[Soundness and Safety]
			The main results of our approach declare that:
			\begin{itemize}
				\item	Satisfiability is sound; it is preserved by the semantics of our underlying computation model
						(\thmref{sr}).

				\item	A system that satisfies a policy is safe: it will never violate the policy (\thmref{safety}),
						where violations are expressed as a class of error systems defined with respect to privacy policies
						(\defref{error}).
			\end{itemize}

	\item[Use cases]
			We use the results above to develop
			two real use cases of systems that
			handle private information and at the same time
			protect the privacy of the information against external adversaries.
%			The first usecase
			\secref{etp} describes the case where an
			electronic system computes the toll fees of a car based on the car's
			GPS locations.
			\secref{speed_control} describes the case where
			an authority identifies a speeding car
			based on the registration number of the car.

%			In both cases the location of the car
%			may reveal private information about the driver
%			of the car and laws exist to regulate
%			access to such data.
\end{description}

\section{Overview of the Approach - Methodology}

\label{sec:methodology}

In this section we give an overview of our approach.
%using examples when appropriate.
We begin
by discussing a model for privacy which is based on a taxonomy
of privacy violations proposed in~\cite{solove06:privacy}.
Based on this model we then propose a policy language and a model
of computation based on the $\pi$-calculus with groups so
that privacy specifications can be described as a privacy
policy in our policy language and specified as terms in
our calculus.
% A privacy specification can be described
%as a privacy policy and implemented as a $\pi$-calculus with groups
%term.

\subsection{A Model for Privacy}
As already discussed, there is no absolute definition of the notion of
privacy. Nevertheless, and in order to proceed with a formal approach
to privacy, we need to identify an appropriate
model that can describe the basics of privacy in the context
of information systems.
In general terms, privacy can be considered as a set of
dynamic relations between individuals that involve privacy
rights and permissions, and privacy violations.

\subsubsection{Privacy as a Taxonomy of Rights.}
A study of the diverse types of privacy, their interplay with
technology, and the need for formal methodologies for understanding
and protecting privacy, is discussed in \cite{Tschantz:privacy}
where the authors follow in their arguments the
analysis of David Solove, a legal scholar who has provided a
discussion of privacy as a taxonomy of possible privacy
violations~\cite{solove06:privacy}.
The privacy model proposed by Solove requires that a {\em data holder}
is responsible for the privacy of a {\em data subject}
against violations from external {\em adversaries}.

\begin{figure}[h]
\begin{center}
	\begin{tabular}{|c||c|c|c|}
		\hline
		\multicolumn{1}{|c||}{\multirow{2}{*}{Invasion}}
		& \multicolumn{3}{c|}{Information}
%		& \multicolumn{3}{|c|}{}
		\\
		\cline{2-4}
		& Collection & Processing & Dissemination
		\\

		\hline
		\hline
		\multirow{2}{*}{Intrusion} &
		\multirow{2}{*}{Surveillance} &
		\multirow{2}{*}{Aggregation} &
		Breach of
		\\
		&&& Confidentiality
		\\

		\hline
		Decisional &
		\multirow{2}{*}{Interrogation} &
		\multirow{2}{*}{Identification} &
		\multirow{2}{*}{Disclosure}
		\\
		Interference &&&
		\\

		\hline
%		\multicolumn{2}{|c|}{}&
		&
		&
		Insecurity &
		Exposure
		\\

		\cline{3-4}
%		\multicolumn{2}{|c|}{
%			\multirow{2}{*}{}
%		}&
		&
		&
		\multirow{2}{*}{Secondary Use} &
		Increased
		\\
%		\multicolumn{2}{|c|}{}
		&
		&
		&
		Accessibility
		\\

		\cline{3-4}
%		\multicolumn{2}{|c|}{}&
		&
		&
		Exclusion &
		Blackmail
		\\

		\cline{3-4}
		&
		\multicolumn{2}{|c|}{}
		&
		Appropriation
		\\

		\cline{4-4}
		&
		\multicolumn{2}{|c|}{}
		&
		Distortion
		\\
%		\cline{4-4}
		\hline
	\end{tabular}
\end{center}
	\caption{A taxonomy on privacy violations}
	\label{fig:taxonomy}
\end{figure}

According to Solove, and based on an in-depth study
of privacy within the legal field,
privacy violations can be distinguished in
four categories as seen in \figref{taxonomy}:
%
%\begin{enumerate}[i.]
%	\item
			i) {\em invasions};
%	\item
			ii) {\em information collection};
%	\item
			iii) {\em information processing}; and
%	\item
			iv) {\em information dissemination}.
%\end{enumerate}
%
Invasion-related privacy violations
%, that include the violations
%of intrusion and decisional interference,
are violations that occur on the physical sphere of an individual and
are beyond the context of computer systems.
Information-related
privacy violations, on the other hand,
are concerned with the manipulation of an individual's personal
information in ways that may cause the individual to be exposed, threatened,
or feel uncomfortable as to how this personal information is being used.

In \figref{taxonomy} we can see the taxonomy developed by Solove.
We concentrate the discussion on the latter three information-related
categories. In
the category of \emph{information collection} we have the privacy violations
of \emph{surveillance} and \emph{interrogation}. Surveillance has to do with the
collection of information about individuals without their consent,
%An example of surveillance in information system is
e.g., by~eavesdropping
%and the collection of information
on communication channels in systems.
Interrogation
%on the other hand is a violation because it
puts the data subject in the awkward position of
denying to answer a question.

The second %privacy violation
category has to do with \emph{information-processing} violations.
The first violation is the one of \emph{aggregation} which describes
the situation where a data holder aggregates information about a data subject:
while a single piece of information about an individual may not pose
a threat to the individual's privacy, a collection of many pieces of information
may reveal a deeper understanding about the individual's character
and habits.
%As an example in an information system consider
%e.g.~large databases
%that hold records of an individual's activity.
The violation of \emph{identification} occurs
when pieces of  anonymous information are matched
against information coming from a known data subject.
%. This procedure may reveal
%larger information pieces about an individual,
This may take place for example by matching the information of
data columns between different tables inside a database and may
lead to giving access to private information about an individual
to unauthorized entities.
% may create larger
%tables of information associate with a data subject id.
\emph{Insecurity} has to do with the responsibility of the data holder
against any sensitive data of  an individual.
Insecurity may lead to identity theft, identification or, more generally, bring
individuals to harm due to their data being not sufficiently secure
against outside adversaries.
%with serious consequences about an individual
For example, in the context of information systems passwords should be kept secret.
\emph{Secondary usage} arises when a data holder uses the data of an individual
for purposes other than the ones the individual has given their consent to.
%, in giving its private data.
For example, a system that records economic activity for logistic reasons
%falls into a secondary usage violation if it
uses the stored data for marketing.
The next violation relates to the right of data subjects to access their private data
and be informed regarding the reasons and purposes the data are being held. In
the opposite case we identify the violation
of \emph{exclusion}.
%. In a different case we identify the
%violation of exclusion.

The last %information violation
category is %the one of
\emph{information dissemination}.
Private information should be disseminated under conditions. If not,
we might have the violations of \emph{breach of confidentiality}, \emph{disclosure} and
\emph{exposure}. In the first case we are concerned with
confidential information. In the second case we assume a non-disclosure
agreement between the data holder and the data
subject, whereas the third case is concerned with the
%has to do with the dissemination of
exposure of embarrassing information about an individual.
%The distinction between the three categories is a subtle subject.
Following to the next violation, \emph{increased accessibility} may occur when
an adversary has access to a collection of non-private pieces of
information but chooses to make the collection of these data more widely available.
For example, while an email of an employee in a business department may
be a piece of public information, publishing a list of such emails
constitutes an increased-accessibility violation
%For example publishing online a list of emails in a business department.
%although emails are public,
%is publicly known,
as it  may be used for other reasons such as advertise spam.
The violation of \emph{blackmail} occurs when an adversary threatens to
reveal private information.
%The violation of blackmail comes from the fact that an adversary may threaten
%to reveal/disseminate private information.
\emph{Appropriation} occurs when an adversary associates private information
with a product and without the consent of the data subject, and, % of the individual.
%For example, online advertisement
%of sports drink associated with an athlete.
%The last category is the category of
finally, \emph{distortion} involves the dissemination
of false information about an individual that may harm %its reputation and
the way the individual is being judged by others.

\subsubsection{A Privacy Model for Information Systems}
%
%I FEEL WE NEED SOME CONCLUSION. I HAVE ADDED THE FOLLOWING, BUT PERHAPS IT SHOULD BE REFINED/EXTENDED

Based on the discussion above we propose the following model for privacy
for information systems:
An information system, the {\em data holder}, is responsible for the privacy of data of
a {\em data subject} against violations from various {\em adversaries} which
can be users of the modules of the system or external entities. These adversaries may perform
operations on private data, e.g.,~store, send, receive, or process the data and,
depending on these operations, privacy violations may occur. For example, sending
private data from a data holder module to an unauthorised module may result
in the violation of disclosure.

At the centre of this schema we have the notion of {\em private data}.
In our model, we take the view that private data are data structures
representing pieces of information related to individuals along with
the identities of the associated individuals. For example, an
on-line company may store in its databases the address, wish list,
and purchase history for each of its customers.
%e.g.~their height, address, social security number, political views, or an illness suffered.

Viewing in our model private data as associations between individuals and pieces of
information is useful for a number of reasons. To begin
with this approach allows us to distinguish between private data and other
pieces of data. For instance, a certain address on a map has no private
dimension until it is defined to be the address of a certain individual.
Furthermore, considering private data to be associations between information and individuals
enables us to reason about a variety of privacy features. One such feature is
anonymisation of private data which occurs when the identity of the individual associated
with the data is stripped away from the data. Similarly, identification occurs
when one succeeds in filling in the identity associated with a piece of
anonymised data. Moreover, aggregation of data is considered to take place
when a system collects many pieces of information relating to the same
individual. Given the above, in our model we consider private
data as a first-class entity, and we make a distinction between
private data and other pieces of data.
% by making an explicit mention of identities within private data.

Taking a step further, in our model we distinguish between different types of private data.
This is the case in practice since, by nature, some data are more sensitive than others
and compromising it may lead to different types of violations. For example, mishandling
the number of a credit card may lead to an identity-theft violation while
compromising a social security number may lead to an identification violation.
Similarly, %in our distinction
we make a distinction between the different
entities of an information system, based on the observation that different entities
may be associated with different permissions with respect to
different types of private data.

Finally, a model of privacy should include the notion of a \emph{purpose}.
As indicated in legal studies, it is often the case that private data may be used
by a data holder for certain purposes but not for others. As such, the notion
of privacy purpose has been studied in the literature, on one hand, in terms of its
semantics~\cite{Tsch12}, and, on the other hand, in terms of policy languages and
their enforcement~\cite{ByunBL05,ColomboF14,KokkinoftaP15}. In our model, we encompass the notion of a purpose
in terms of associating data manipulation with purposes.

\subsection{Privacy Policies}

After identifying a model for privacy, the next step is to use a proper
{\em description language} able to describe the privacy model and
serve as a bridge between the privacy terminology and a formal framework.

To achieve this objective, we
%
%\paragraph{Privacy Policies and the $\pi$-calculus.}
%To achieve this objective, we develop a meta-theory for the $\pi$-calculus that
%captures privacy as policy.
%Following the model of~\cite{Tschantz:privacy}, we
create a policy language that enables us to describe
privacy requirements for private data over data entities.
%Entities in the language represent the data
%holders, data subjects and/or adversaries.
For each type of private data we expect
entities to follow different policy requirements.
%These policy requirements
%are described through a set of permissions.
Thus, we define policies as objects that describe
%, for each data type,
a hierarchical nesting of entities
where each node/entity of the hierarchy is associated with a set of privacy
permissions. %where permissions are inherited between nested entities.
The choice of permissions encapsulated within a policy language
is an important issue
%, not only due to its practical relevance but also
because identification of these permissions constitutes, in a sense, a
characterization of the notion of privacy. In this work, we make an attempt of
identifying some such permissions, our choice emanating from a class of
privacy violations of Solove's taxonomy which we refine by considering some common
applications where privacy plays a central role.

\newcommand{\Research}{\groupfont{Research}}
\newcommand{\Police}{\groupfont{Police}}
\newcommand{\Lab}{\groupfont{Lab}}
\newcommand{\research}{\pdatafont{research}}

\begin{exa}
	\label{ex:hospital_policy}
	As an example consider the medical system of a hospital (\Hospital)
	obligated to protect patients' data (\patientd).
	Inside the hospital there are five types of adversaries - the database
    administrator, nurses,
    doctors, a research department and a laboratory -  each implementing
	a different behaviour with respect to the privacy policy in place.
	Without any formal introduction,
	we give the privacy policy for the patient private data as an entity $\patientd \gg H$ where
    $H$ is a nested structure that assigns
    different permissions to each level of the hierarchy, as shown below:
	\[
		\begin{array}{l}
			\patientd \gg \Hospital\set{}[\\
			\quad \quad \lrole{\DBase} {\str, \aggr}{},\\
			\quad \quad \lrole{\Nurse} {\rf, \diss{\Hospital}{\infty}}{},\\
			\quad \quad \lrole{\Doctor}{\rf, \rd, \up, \idp, \use{\diagnosis}, \ndiss{\conf}}{}\\
			\quad \quad \lrole{\Research}{\rf, \rd, \use{\research}, \ndiss{\discl}}{}\\
			\quad \quad \lrole{\Lab}{\rf, \rd, \idp, \idf{\crime}, \diss{\Police}{1}}{}\\
		]
		\end{array}
	\]
	According to the policy the various adversaries are assigned permissions as follows:
    A database administrator (\DBase) has the right
	to store (\str) and aggregate (\aggr) patients' data in order to
	compile patients' files.
    %Patients file are abstracted as a reference (permission \rf).
	A nurse (\Nurse) in the hospital is able to access (\rf) a patient's file
	but not read or write data on the file. A nurse may also disseminate
	copies of the file inside the hospital (\diss{\Hospital}{\infty}), e.g.,~to a Doctor.
	A doctor (\Doctor) in the \Hospital may gain access (\rf) to a patient's file, read it (\rd)
    with access to the identity of the patient (\idp),
	and update (\up) the patient's information. A doctor may also use the patient's data to
	perform a diagnosis (\use{\diagnosis}) but cannot disseminate it
	since this would constitute a breach of confidentiality (\ndiss{\conf}).
	
    In turn, a research department (\Research) can access a patient's file (\rf),
    read  the information (\rd), and use it for performing research (\use{\research}).
	However, it is not entitled to access the patient's identity (lack of \idp permission)
    which implies that all information available to it should be anonymised.
    The research department has no right of disclosing
	information (\ndiss{\discl}).
	Finally, a laboratory within the hospital system (\Lab) is allowed to gain access and
	read private data, including the associated identities (\rf, \rd, \idp) and it
    may perform identification using patient data against
    evidence collected on a crime scene
	(permission \idf{\crime}). If an identification succeeds then the \Lab may disseminate
	the patient's identity to the police (permissions \diss{\Police}{1}).
\end{exa}

\subsection{The $\pi$-calculus with groups}
\label{sec:intro_groups}

Moving on to the framework underlying our study, we employ a variant of the
$\pi$-calculus with groups~\cite{Cardelli:secrecy} which we refer to as the \Pcalc.
The $\pi$-calculus with groups extends the $\pi$-calculus with the notion of
\emph{groups} and an associated type system in a way that controls
how data are being processed and disseminated inside a system.
It turns out that groups give a natural abstraction for the representation
of entities in a system. Thus, we build on the notion of a group of the
calculus of~\cite{Cardelli:secrecy}, and we use the group memberships of
processes to distinguish their roles within systems.

To better capture our distilled privacy model, we extend the $\pi$-calculus
with groups as follows. To begin with, we extend the calculus with the notion
of private data: as already discussed, we take the view that private data
are structures of the form $\pdata{\id}{a}$ where $\id$ is the identity of
an individual and $a$ is an information associated with the individual.
To indicate private data whose identify is unknown we write $\pdata{\hid}{a}$.
Furthermore,
we define the notion of a \emph{store} which is a process that stores
and provides access to private data. Specifically, we write $\store{r}{\id}{a}$
for a store containing the private data $\pdata{\id}{a}$ with $r$ being
a link/reference via which this information may be read or updated.
As we show, these constructs are in fact high-level constructs that
can be encoded in the core calculus, but are useful for our subsequent
study of capturing and analyzing privacy requirements.

Given this framework, information collection, processing and dissemination
issues can be analysed through the use of
names/references of the calculus in input, output and object position to
identify when a channel is reading, writing or otherwise manipulating private data, or
when links to private data are being communicated between groups.

\begin{exa}
An implementation of the hospital scenario in \exref{hospital_policy}
in the $\pi$-calculus with groups could be as follows:
\[
	\begin{array}{l}
		\group{\Hospital}{\\
			\begin{array}{rcl}
				&& \group{\DBase}{\;\store{r_1}{\id}{\dna} \Par \store{r_2}{\id}{\medication}\;}\\
				&\Spar&	\group{\Nurse}{\;\out{a}{r_1, r_2} \inact\;}\\
				&\Spar&	\group{\Nurse}{\;\out{b}{r_1} \out{b}{r_1} \inact\;}\\
				&\Spar&	\group{\Doctor}{\;\inp{a}{w, z} \inp{w}{\pdata{x}{y}} \ifelse{y}{\disease_1}{\out{z}{\pdata{x}{\medication'}} \inact}{\inact}\;}\\
				&\Spar&	\group{\Research}{\;\inp{b}{w} \inp{w}{\pdata{\hid}{y}} \ifelse{y}{\disease_2}{P}{Q}\;}\\
				&\Spar&	\group{\Lab}{\;\inp{b}{w} \inp{w}{\pdata{x}{y}} \inp{r}{\pdata{\hid}{z}} \ifelse{y}{z}{ \out{c}{w} \inact }{\inact}\;}
			\end{array}
		\\
		}
	\end{array}
\]
In this system, \Hospital constitutes a group that is known
to the six processes of the subsequent parallel composition.
The six processes are defined on groups \DBase, \Nurse
\Doctor, \Research, and \Lab nested within the \Hospital group.
%and available to processes
%$\store{r_1}{\id}{\dna} \Par \store{r_2}{\id}{\medication}$,
%$\out{a}{r_1, r_2} \inact$,
%$\inp{a}{w, z} \inp{w}{\pdata{x}{y}} \ifelse{y}{\disease_1}{\out{z}{\pdata{x}{\medication'}} \inact}{\inact}$,
%$\inp{a}{w} \inp{w}{\pdata{\hid}{y}} \ifelse{y}{\disease_2}{P}{Q}$, and
%$\inp{a}{w} \inp{w}{\pdata{x}{y}} \inp{r}{\pdata{\hid}{z}} \ifelse{y}{z}{ \out{b}{w} \inact }{\inact}$,
%respectively.
%%Thus, the system is comprised of two processes, one that possesses
%%group memberships $\{\mathsf{Hospital},\mathsf{Nurse}\}$, the nurse,
%%and another that possesses
%%group memberships $\{\mathsf{Hospital},\mathsf{Doctor}\}$, the doctor.
The group
memberships of the processes above characterise their nature  and allow us to
endow them with permissions
while reflecting the entity hierarchy
expressed in the  privacy policy defined above.

The system above describes the cases where:
\begin{itemize}
	\item	A \DBase process defines {\em store} processes that store the patient
			data $\dna$ and $\medication$ associated with a patient's identity, with $\dna$ corresponding
			to the dna of the patient and \medication to the current treatment of the patient.

	\item	A \Nurse process may hold a patient's
			files, represented with names $r_1$ and $r_2$, and can disseminate them inside
			the \Hospital. In this system there are two \Nurse processes that, respectively,
			disseminate patient's files to the doctor via channel $a$, or to the labs via
			channel $b$ (the patient file is represented by reference $r_1$).
            %to other components that have the permission to get the file of a patient.

	\item	A \Doctor may receive a patient's file,
			read it, and then, through the conditional statement, use it to perform diagnosis.
			As we will see below, we identify diagnosis through the type of the matching
			name $\disease_1$ which in this scenario represents the genetic fingerprint of a disease.
			After the diagnosis the \Doctor may update the treatment of the patient in the corresponding
			store.

	\item	A \Research lab may receive a patient's file and perform statistical analysis
			based on the matching of the genetic material of the patient. The
			patient's data made available to the research lab
            is anonymised, as indicated by the private-data value $\pdata{\hid}{z}$.

	\item	A \Lab may receive a patient's file, and perform
			forensic analysis on the patient's \dna ($\pdata{x}{y}$)
            to identify
			dna evidence connected to a crime ($\pdata{\hid}{z}$).
			After a successful identification
			the \Lab may inform the \Police via channel $c$.
\end{itemize}
\end{exa}

\noindent The implementation above can be associated to the privacy policy defined earlier using
static typing techniques for the $\pi$-calculus. We clarify this intuition
by an informal analysis of the \Doctor process. Without any formal
introduction to types, assume that structure
$\pdata{x}{y}$ has type $\privatet{\patientd}{\dna}$ that concludes that
variable $w$, being a channel that disseminates such information within the
$\Hospital$ system, has type $\gtype{\Hospital}{\privatet{\patientd}{\dna}}$.
Structure $\disease_1$ has type $\purposet{\diagnosis}{\dna}$  indicating
that it is a constant piece of information that can be used for the purposes
of performing a diagnosis. Furthermore,
structure $\pdata{x}{\medication'}$ has type $\privatet{\patientd}{\treatment}$.
The fact that the \Doctor receives the patient's file on name $w$
signifies the \rf permission, whereas inputting value $\pdata{x}{y}$ signifies the \rd
permission. The matching operator between \patientd and \diagnosis data identifies
the usage of private data for the purpose of a diagnosis. Finally, the \Doctor
has the right to update private data on the patient's file on name $z$.
We can see that the \Doctor does not disseminate the file (names $w$ and $z$) outside its
group, thus, there is respect of the \ndiss{\conf} permission.

%To clarify the above description, consider the following type mappings.
%The types of the names in the above process are defined as
%
%\[
%	\begin{array}{l}
%		\pdata{\id}{\dna}: \privatet{\patientd}{\dna}\\
%		\pdata{\id}{\medication}: \privatet{\patientd}{\treatment}\\
%		\pdata{x}{y}: \privatet{\patientd}{\dna}\\
%		\pdata{x}{\medication'}: \privatet{\patientd}{\treatment}\\
%%		d_1, d_2: \treatment\\
%		\disease_1: \dna
%	\end{array}
%\]
%%
%\dk{fix - add more}
%that is $\pdata{\id}{b, d_1}, \pdata{x}{y, z}$ are values of sensitive data and $d_1, d_2, c$ are values data within the system.
%The type
%\[
%	r: \Hospital[\privatet{\patientd}{\bpressure, \treatment}]
%\]
%signifies that $r$ is a channel that can be used only by processes which belong to group $\Hospital$
%to carry data of type $\privatet{\patientd}{\bpressure, \treatment}$.
%Similarly,
%\[
%	a: \Hospital[\Hospital[\privatet{\patientd}{\bpressure, \treatment}]]
%\]
%states that $a$ is a channel that can be used by members of group
%$\Hospital$,  to carry objects of type $\Hospital[\Hospital[\privatet{\patientd}{\bpressure, \treatment}]]$.
%Here, we note the distinction between a link on
%patient's data (name $r$) and patient's data themselves (value $\pdata{\id}{b, d_1}$).

Intuitively,
we may see that this system conforms to the defined policy, both in terms of
the group structure as well as the permissions exercised by the processes.
Instead, if the nurse were able to engage in a $r_1?\pdata{x}{y}$ action
then the defined policy would be violated because the nurse would have read the
patient's private data without having the permission to do so.
%Similarly, a violation would occur in the case where the
%type of $a$ was defined as
%$a: \Other[\Hospital[\privatet{\patientd}{\dna, \treatment}]]$ for some
%distinct group $\Other$ because a disclosure outside the \Hospital might occur.
Thus, the encompassing group is essential for capturing requirements of non-disclosure.

Using these building blocks, our methodology is
applied as follows:
Given a system and a typing we perform type checking to
confirm that the system is well-typed while we infer a permission interface.
This interface captures the permissions exercised by the system.
To check that the system complies  with a privacy
policy we provide a correspondence between policies and permission interfaces
the intention being that: a permission interface
satisfies a policy if and only if the system exercises a subset of the allowed
permissions of the policy.
With this machinery at hand, we state and prove a safety theorem according to which,
if a system $\mathsf{Sys}$ type-checks against a typing $\Ga; \La$
and produces an interface $\Int$, and $\Int$ satisfies a privacy policy $\Policy$, then $\mathsf{Sys}$ respects
policy $\Policy$.

\section{Calculus}
\label{sec:calculus}

In this section we define a model of concurrent computation
whose semantics captures the privacy model we intend to investigate.
The calculus we propose is called the \Pcalc and it is based on the $\pi$-calculus
with groups originally presented by Cardelli et al.~\cite{Cardelli:secrecy}.
The $\pi$-calculus with groups is an extension of the $\pi$-calculus
with the notion of a group and an operation of group creation, where
a group is a type for channels.
In~\cite{Cardelli:secrecy} the authors establish a close connection
between group creation and secrecy as they show that a secret belonging
to a certain group cannot be communicated outside the initial scope of the
group.
%\dk{For the sake of simplicity we present the monadic version of \Pcalc.
%The calculus together with the type system can be extended in the polyadic
%case in a straightforward manner. For the examples in this paper we assume
%the polyadic version of \Pcalc.}

%The calculus we propose
The \Pcalc extends the $\pi$-calculus with
groups with some additional constructs that are useful for our privacy investigation.
In fact, as we show, these additional constructs are high-level
operations that can be encoded in the core calculus.
We begin our exposition by describing the basic intuition behind the design choices
of the calculus:
\begin{enumerate}
	\item	Private data are typically data associated to individuals.
            For instance, an election vote is not a private piece of information.
            It becomes, however, private when associated with an identifying piece
            of information such as a name, a social security number, or an address.
            So, in the \Pcalc we distinguish between $\pi$-calculus names and
            private data. In particular, we assume that {\em private data} are structures
            that associate constants, that is, pieces of information, with identifying
            pieces of information which we simply refer to as \emph{identities}.

	\item	Privacy enforcement in an information system is about
			controlling the usage of private data which is typically stored
            within a database of the system. In the \Pcalc we capture such
            databases of private data as a collection of {\em stores}, where a store is
            encoded as a high-level process term and includes operations for manipulation
			of private data. This manipulation takes place with the use of a special kind
            of names called {\em references}.
\end{enumerate}

\begin{figure}[t]
	\[
	\begin{array}{lrcl}
		\text{(identity values)}			& \ii &\bnfis&	\id \bnfbar \hid \bnfbar x \\
		\text{(data values)}				& \con &\bnfis& c \bnfbar x
		\\
		\\
		\text{(private data)}
		&
		\pdata{\ii}{\con}		&&  \text{where }\ii \not= x \implies {\con} = {c} \text{ and } \ii = x \implies {\con} = {y}
		\\
		\\
		\text{(identifiers)}	& u		&\bnfis&		a \bnfbar r \bnfbar x
		\\
		\text{(terms)}			& t		&\bnfis&		a \bnfbar r \bnfbar \pdata{\ii}{\con} \bnfbar c \bnfbar x
		\\
	    \text{(constant terms)}	& v		&\bnfis&		a \bnfbar r \bnfbar \pdata{\id}{c} \bnfbar c
		\\
		\text{(placeholders)}	& k		&\bnfis&		x \bnfbar \pdata{x}{y} \bnfbar \pdata{\hid}{x}

		\\
		\\
		\text{(processes)}
		&
		P	&\bnfis&	\inact \bnfbar \out{u}{t} P \bnfbar \inp{u}{k} P \bnfbar \newn{n} P \bnfbar P \Par P \bnfbar \repl P\\
			&&\bnfbar&	\ifelse{t}{t}{P}{P} \bnfbar \store{r}{\ii}{\con}
			%\store{u}{\ii}{\con}
			%\store{r}{\id}{c} %{c}_1, \ldots, {c}_n\}}
			%\store{\dual{r}}{\id}{c}
		\\
		\\
		\text{(systems)}
		&
		S	&\bnfis&	\group{\G}{P} \bnfbar \group{\G}{S} \bnfbar S \Spar S \bnfbar \newn{n} S
	\end{array}
	\]
	\caption{Syntax of the \Pcalc}
	\label{fig:syntax}
\end{figure}

\noindent To make the intuition above concrete, let us assume a set
of names $\mathcal{N}$, ranged over by $n, m, \ldots,$ and partitioned
into a set of channel names $\mathsf{Ch}$,
ranged over by $a, b, \ldots$, a set of reference names $\mathcal{R}$, ranged
over by $r, r',\ldots$,
and a set of dual endpoints, $\dual{\mathcal{R}}$,  where for each reference $r \in \mathcal{R}$ we assume the existence of
a dual endpoint $\dual{r} \in \dual{\mathcal{R}}$.
%\dk{Note that $\mathcal{N} \cap \dual{\mathcal{R}} = \es$.}
The
endpoint $\dual{r}$ belongs solely to a store-term and is used
for providing access to the associated private data whereas the endpoint $r$
is employed by processes that wish to access the data. Note that the main
distinguishing feature between the two endpoints of a reference is that
an endpoint $\dual{r}$ cannot be passed as an object of a communication, whereas
an endpoint $r$ can: while it is possible to acquire a reference for accessing
an object during computation it is not possible to acquire a reference for
providing access to a store.
Finally, we point out  that channels do not require dual endpoints so we assume that $\dual{a} = a$.

In addition to the set of names our calculus makes use of the following entities.
We assume a set of groups~\cite{Cardelli:secrecy}, $\mathcal{G}$,
that ranges over
$\G, \G_1,\ldots $. Groups are used to control name creation and to provide the property
of secrecy for the information that is being exchanged while characterising
processes within a system.
Furthermore, we assume a set of variables $\mathcal{V}$ that ranges over
$w, x, y, z, \ldots$. % and a set of labels $\mathcal{L}$
%that ranges over $l, l', \ldots$.
Data are represented using the constants set $\mathcal{C}$ ranged over by $c$,
while identities $\id$ range over a set of identities $\mathsf{Id}$. We assume
the existence of a special identity value `$\hid$' called the
hidden identity that is used to
signify that the identity of private data is hidden.
A hidden identity is used by the calculus to enforce private
data anonymity.

The syntax of the \Pcalc is defined in \figref{syntax}.
We first assume that an {\em identity value} \ii can be an identity \id,
a hidden identity \hid or an identity variable $x$.
We also define a {\em data value} \con to be a constant $c$ or a variable $x$.

As already discussed,   private data %$d \in \mathcal{D}$
are considered to be structures that associate an identity value with a
data value, written \pdata{\ii}{\con}. %\pdata{\ii}{\{\delta_1,\ldots,\delta_n\}}.
Structure \pdata{\id}{c} associates information $c$ with
identity \id, while structure \pdata{\hid}{c} contains
private information $c$ about an individual
without revealing the individual's identity.
Finally,  private data can take one of the forms
\pdata{x}{y} and \pdata{\hid}{x}. % where a substitution of {\pdata{x}{y}}
%by a constant structure yields
%$\pdata{x}{y} \subst{\pdata{id}{c}}{\pdata{x}{y}}=\pdata{\id}{c}$
%and $\pdata{\hid}{y} \subst{\pdata{id}{c}}{\pdata{\hid}{y}}=\pdata{\hid}{c}$.
Note that private data are restricted by definition  only to the four
forms defined above. Any other form of \pdata{\ii}{\con},
e.g.,~\pdata{\id}{x}, is disallowed by definition.

Based on the above, the meta-variables of the calculus include the following:
\begin{itemize}
	\item	{\em identifiers} $u$  denote channels, references or variables,
	\item	{\em terms} $t$  are channels, references, private data, constants or variables,
	\item	{\em constant terms} $v$ include all terms except variables, and
	\item	{\em placeholders} $k$ describe
			either variables or variable private data.
\end{itemize}

The syntax of the calculus is defined at two levels, the process level,
$P$, and the system level, $S$.
At the {\em process} level we have the $\pi$-calculus syntax extended
with the syntax for stores.
Process \inact is the inactive process.
The output prefix $\out{u}{t} P$ denotes a process
that sends a term $t$ over identifier $u$ and proceeds as $P$.
As already mentioned, term $t$ may not be the
dual endpoint of a reference.
Dually, the input prefix $\inp{u}{k} P$ denotes a process
that receives a value over identifier $u$, substitutes it on
variable placeholder $k$, and proceeds as $P$.
Process $\newn{n} P$ restricts name $n$ inside the scope of $P$.
Process $P \Par Q$ executes processes $P$ and $Q$ in parallel.
Process $\repl P$ can be read as an infinite number of
$P$'s executing in parallel.
The conditional process $\ifelse{t_1}{t_2}{P}{Q}$ performs matching
on terms $t_1$ and $t_2$ and continues to $P$ if the match
succeeds, and to $Q$ otherwise.
Finally, process \store{r}{\ii}{\con} is a process that is
used to store data of the form $\pdata{\id}{c}$.
We point out that store references cannot be specified via variable placeholders. This is to
ensure that each store has a  unique store reference. Nonetheless, stores can be
created dynamically, using the construct of name restriction
in combination with process replication.

Free and bound variables of a process $P$, denoted \fv{P} and \bv{P}, respectively,
follow the standard $\pi$-calculus
definition for all $\pi$-calculus constructs, whereas for
store processes we define
$\bv{\store{r}{\ii}{\con}} = \es$ and
$\fv{\store{r}{x}{y}} = \set{x, y}$.
Additionally, it is convenient to define
$\fv{x} = \set{x}$ and $\fv{\pdata{x}{y}} = \set{x, y}$.
The substitution function, denoted \subst{v}{k},
is then defined as follows: In the case that $k$ is $x$ then, the substitution follows
the standard $\pi$-calculus definition and results in substituting all free occurrences
of $k$ by $v$ with renaming of bound names, if necessary, to prevent any occurrences
of $v$ from becoming bounded. In the case that $k = \pdata{x}{y}$ and $v = \pdata{\id}{c}$,
the substitution results in substituting all free occurrences of $x$ and $y$ by $\id$ and $c$,
respectively. Finally, in the case that $k = \pdata{\hid}{y}$ and $v = \pdata{\hid}{c}$,
we define  \subst{v}{k} as the substitution of all free occurrences of $y$ by $c$. %Thus, for
%example,  $\pdata{\hid}{y} \subst{\pdata{id}{c}}{\pdata{\hid}{y}}=\pdata{\hid}{c}$
%which implements our intention that a process should not be able to acquire the identity
%of anonymized data.
Note that the substitution function is undefined for other combinations
of $v$ and $k$.

%the substitution results in substituting all free occurrences of $x$ and $y$ by $\id$ and $c$,for $y$for all the
%$\pi$-calculus terms.
%For stores
%we define
%$\store{r}{x}{y} \subst{\pdata{\id}{c}}{\pdata{x}{y}} = \store{r}{\id}{c}$ and
%$\store{r}{\id}{c} \subst{\pdata{\id}{c'}}{\pdata{\id}{c}} = \store{r}{\id}{c'}$.

%Note, that we choose to define a store as a data structure that stores multiple
%private data structures in order to simulate the privacy of aggregation.

In turn, systems are essentially processes that are extended to include the group construct.
Groups are used to arrange processes into multiple levels of naming abstractions.
The group construct is
applied both at the level of processes \group{\G}{P} and at the level of systems
\group{\G}{S}. Finally, we have the name restriction construct as well as parallel
composition for systems.

Note that, for the sake of simplicity, we have presented the monadic version of the \Pcalc.
The calculus together with the type system can be extended in the polyadic
case in a straightforward manner. Furthermore, we point out that, unlike the $\pi$-calculus with groups, the
\Pcalc does not include a group creation construct. This is due to the fact that groups
in  our framework are associated with classes of entities and their associated privacy
permissions. Thus, the groups of a system need to be fixed.
%For the examples in this paper we assume the polyadic version of \Pcalc.

%\begin{exa}
%\end{exa}

%%%%%%%%%%%%%%%%%%
%   LTS
%%%%%%%%%%%%%%%%%%

\subsection{Labelled Transition Semantics}

We present the semantics of the \Pcalc through
a labelled transition system (LTS).
%In our system there is no notion of group extrusion
%as it is the case in~\cite{Cardelli:secrecy}.
We define a labelled transition semantics instead
of a reduction semantics due to a characteristic of
the intended structural congruence in the \Pcalc.
In particular, the definition of such a congruence would
omit the axiom
$\group{\G}{S_1 \Par S_2} \equiv \group{\G}{S_1} \Par S_2$ if $\G\not\in\fg{S_2}$
as it was used in~\cite{Cardelli:secrecy}. This is due to our intended semantics of the
group concept which is considered to assign capabilities to processes.
Thus, nesting of a process $P$ within some group $\G$, as in $\group{\G} P$,
cannot be lost even if $\G\not\in\fg{P}$, since the $\group{\G}{\cdot}$ construct
has the additional meaning of group membership in the \Pcalc and it
instils $P$ with privacy-related permissions as we will discuss in
the sequel. The absence of this law renders the reduction semantics rule
of parallel composition rather complex.
%This
%complicates the definition of the semantics of
%the calculus as operational semantics. To offer a
%straightforward presentation we choose the definition
%of a LTS.

To define a labelled transition semantics we first introduce the set of labels:
\[
	\ell	\bnfis	\newn{\pol{m}}\actout{n}{v} \bnfbar \actinp{n}{v} \bnfbar \tau
\]
We distinguish three action labels.
The output label $\newn{\pol{m}}\actout{n}{v}$ denotes
an output action on name $n$ that carries object $v$. Names $\pol{m}$
is a (possibly empty) set of names in the output object that are restricted.
The input label $\actinp{n}{v}$ denotes the action of inputting
value $v$ on name $n$.
Action $\tau$ is the standard internal action.
To clarify internal interaction we define a symmetric duality relation $\mathsf{dual}$
over labels $\ell$:
\[
	\begin{array}{c}

		%\newn{\pol{m}} \actout{n}{v} \actdual \actinp{n}{v}
		{\newn{\pol{m}} \actout{a}{v} \actdual \actinp{a}{v}}
		\\[4mm]
		\tree {
			(v_1 = v_2) \vee
			(v_1 = \pdata{\id}{c} \land  v_2 = \pdata{\hid}{c})
		}{
			\newn{\pol{m}} \actout{\dual{r}}{v_1} \actdual \actinp{r}{v_2}
		}
		\qquad
		{\tree {
			(v_1 = v_2) \vee
			(v_2 = \pdata{\id}{c} \land  v_1 = \pdata{\hid}{c})
			%\\
			%\textrm{is this condition necessary?}
			%\\
			%\textrm{can you update when there is anonymity}
		}{
			\newn{\pol{m}} \actout{r}{v_1} \actdual \actinp{\dual{r}}{v_2}
		}}
	\end{array}
\]
The first pair of the relation defines the standard input/output duality between label actions,
where an output on name {$a$} matches an input on name {$a$} that carries the
same object. The {last two duality pairs relate} actions on dual {reference} endpoints:
we distinguish between the case where dual {reference} endpoints carry the same object
and the case where an input {(resp., output)} without an identity
is matched against an output {(resp., input)} with an identity. Thus, we allow
communicating private data without revealing their identity.

\begin{figure}[t]
	\begin{mathpar}
		\LOut\;	\out{n}{v} P \trans{\actout{n}{v}} P
		\and
		\LInp\;	\inp{n}{k} P \trans{\actinp{n}{v}} P \subst{v}{k}
		\and
		\LSOut\;	\store{r}{\id}{c} \trans{\actout{\dual{r}}{\pdata{\id}{c}}} \store{r}{\id}{c}		%}
		\and
        \inferrule*[left=\LSInp] {
			\ii=x \lor \ii = \id
		}{
			\store{r}{\ii}{\con} \trans{\actinp{\dual{r}}{\pdata{\id}{c}}} \store{r}{\id}{c}
		}
		\and
		\inferrule*[left=\LTr] {
			P \trans{\ell} P'
		}{
			\ifelse{a}{a}{P}{Q} \trans{\ell} P'
		}
		\and
		\inferrule*[left=\LFl] {
			Q \trans{\ell} Q'
			\and
			a \not= b
		}{
			\ifelse{a}{b}{P}{Q} \trans{\ell} Q'
		}
		\and
		\inferrule*[left=\LRes] {
			P \trans{\ell} P'
			\and
			n \notin \fn{\ell}
		}{
			\newn{n} P \trans{\ell} \newn{n} P'
		}
		\and
		\inferrule*[left=\LScp] {
			P \trans{\newn{\pol{m}} \actout{n}{v}} P'
			\and
			m \in \fn{v}
		}{
			\newn{m} P \trans{\newn{m \cat \pol{m}} \actout{n}{v}} P'
		}
		\and
		\inferrule*[left=\LTau] {
			P \trans{\ell_1} P'
			\and
			Q \trans{\ell_2} Q'
			\and
			\ell_1 \actdual \ell_2
		}{
			P \Par Q \trans{\tau} \newnp{\bn{\ell_1} \cup \bn{\ell_2}} {P' \Par Q'}
		}
		\and
		\inferrule*[left=\LLPar] {
			P \trans{\ell} P'
			\and
			\bn{\ell}\cap \fn{Q} = \es
		}{
			P \Par Q \trans{\ell} P' \Par Q
		}
		\and
		\inferrule*[left=\LRPar] {
			Q \trans{\ell} Q' \and \bn{\ell}\cap \fn{P} = \es
		}{
			P \Par Q \trans{\ell} P \Par Q'
		}
		\and
		\inferrule*[left=\LRepl] {
			P \trans{\ell} P'
		}{
			\repl P \trans{\ell} P' \Par \repl P
		}
		\and
		\inferrule*[left=\Alpha] {
			P\equiv_{\alpha} P''\quad P'' \trans{\ell} P'
		}{
			 P \trans{\ell} P'
		}
	\end{mathpar}
	\caption{Labelled Transition System for Processes}
	\label{fig:lts}
\end{figure}

The labelled transition semantics for processes is defined in \figref{lts} and the
labelled transition semantics for systems in \figref{slts}.

According to rule $\LOut$,
output-prefixed process $\out{n}{v} P$
may interact on action $\actout{n}{v}$ and continue as $P$.
Similarly, input-prefixed process $\inp{n}{k} P$ may interact on action $\actinp{n}{v}$
and continue as $P$ with $k$ substituted by $v$ as in rule $\LInp$.
Naturally, the input action can only take place for compatible $v$ and $k$ values which are
the cases when the substitution function is defined.
Rules $\LSOut$ and $\LSInp$ illustrate the two possible actions of a store
process: a store may use the dual endpoint of the store reference $r$
%and in particular any reference instance $r_i$, $1\leq i \leq n$,
to engage in action $\actout{\dual{r}}{\pdata{\id}{{c}}}$ or in action
\actinp{\dual{r}}{\pdata{\id}{{c}}}. In the first case it returns
to its initial state (rule $\LSOut$) and in the second case it updates
 the store accordingly (rule $\LSInp$). Note that, for the
 input action, if the store already contains data, then the received data should
match the identity of the stored data.
In turn, rules $\LTr$ and $\LFl$ define the semantics of the conditional construct
for the two cases where the condition evaluates to true and false, respectively.
Through rule $\LRes$ actions are closed under
the restriction operator provided that the restricted
name does not occur free in the action.
Rule $\LScp$ extends along with the action $\newn{\pol{m}} \actout{n}{v}$
the scope of name $m$ if $m$ is restricted.
Next, rule $\LTau$ captures that parallel processes
may communicate with each other on dual actions and produce
 action $\tau$, whereas rules $\LLPar$ and $\LRPar$ are symmetric rules
that state that actions
are closed under the parallel composition provided that
there is no name conflict between the bounded names of the
action and the free names of the parallel process.
Continuing to the replication operator, $\LRepl$ states that  $\repl P$ may execute
an action of $P$ and produce the parallel composition of
the replication and the continuation of $P$,
and, according to rule $\Alpha$, the labelled transition
system is closed
under alpha-equivalence ($\equiv_{\alpha}$).

\begin{figure}[t]
	\begin{mathpar}
		\inferrule*[left=\SPGr] {
			P \trans{\ell} P'
		}{
			\group{\G}{P} \trans{\ell} \group{\G}{P'}
		}
		\and
		\inferrule*[left=\SGr] {
			S \trans{\ell} S'
		}{
			\group{\G}{S} \trans{\ell} \group{\G}{S'}
		}
		\and
		\inferrule*[left=\SRes] {
			S \trans{\ell} S'
			\and
			n \notin \fn{\ell}
		}{
			\newn{n} S \trans{\ell} \newn{n} S'
		}
		\and
		\inferrule*[left=\SScp] {
			S \trans{\newn{\pol{m}} \actout{n}{v}} S'
			\and
			m \in \fn{v}
		}{
			\newn{m} S \trans{\newn{m \cat \pol{m}} \actout{n}{v}} S'
		}
		\and
		\inferrule*[left=\SLPar] {
			S_1 \trans{\ell} S_1'
			\and
			\bn{\ell}\cap\fn{S_2} = \es
		}{
			S_1 \Spar S_2 \trans{\ell} S_1' \Spar S_2
		}
		\and
		\inferrule*[left=\SRPar] {
			S_2 \trans{\ell} S_2'
			\and
			\bn{\ell}\cap\fn{S_1} = \es
		}{
			S_1 \Spar S_2 \trans{\ell} S_1 \Spar S_2'
		}
		\and
		\inferrule*[left=\STau] {
			S_1 \trans{\ell_1} S_1'
			\and
			S_2 \trans{\ell_2} S_2'
			\and
			\ell_1 \actdual \ell_2
		}{
			S_1 \Spar S_2 \trans{\tau} \newnp{\bn{\ell_1} \cup \bn{\ell_2}} {S_1' \Spar S_2'}
		}
	\end{mathpar}

	\caption{Labelled Transition System for Systems}
	\label{fig:slts}
\end{figure}

Moving on to the semantics of systems we have the following: According
to rules $\SPGr$ and $\SGr$, actions of processes and systems are preserved by
the group restriction operator. The remaining rules are similar to their
associated counter-parts for processes.

We often write $\trans{}$ for $\trans{\tau}$ and
$\Trans{}$ for $\trans{}^*$. Furthermore,
we write $\Trans{\hat{\ell}}$ for $\Trans{} \trans{\ell} \Trans{}$
if $\ell = \tau$, and $\Trans{}$ if $\ell = \tau$. Finally,
given $\pol{\ell} = \ell_1\ldots\ell_n$,
we write $\Trans{\pol{\ell}}$ for $\trans{\ell_1}\ldots\trans{\ell_n}$.

It is useful to define a structural congruence relation
over the terms of the \Pcalc.
\begin{defi}
	Structural congruence, denoted \scong, is defined
	separately on process terms and system terms.

	\begin{itemize}
		\item	Structural congruence for processes is a binary relation
				over processes defined as the least congruence that satisfies the rules:
				\[
				\begin{array}{c}
					P \Par Q \scong Q \Par P
					\qquad
					(P \Par Q) \Par R \scong P \Par (Q \Par R)
					\qquad
					P \Par \inact \scong P
					\qquad
					\newn{a} \inact = \inact
					\\
					n \notin \fn{Q} \implies \newn{n} P \Par Q \scong \newnp{n}{P \Par Q}
					\qquad
					\newn{n} \newn{m} P \scong \newn{m} \newn{n} P
				\end{array}
				\]

		\item	Structural congruence for systems is a binary relation
				over systems defined as the least congruence that satisfies the rules:
				\[
				\begin{array}{c}
					P \scong Q \implies \group{\G}{P} \scong \group{\G}{Q}
					\qquad
					S_1 \Spar S_2 \scong S_2 \Spar S_1
					\qquad
					(S_1 \Spar S_2) \Spar S_3 \scong S_1 \Spar (S_2 \Spar S_3)
%					\qquad
%					S \Par \inact \scong P
%					\qquad
%					\newn{a} \inact = \inact
					\\
					n \notin \fn{S_2} \implies \newn{n} S_1 \Par S_2 \scong \newnp{n}{S_1 \Par S_2}
					\qquad
					\newn{n} \newn{m} S \scong \newn{m} \newn{n} S
				\end{array}
				\]
	\end{itemize}
\end{defi}\medskip

\noindent The structural congruence relation preserves the labelled transition semantics:
\begin{thm}
	Consider processes $P$ and $Q$, and systems $S_1, S_2$.
	\begin{itemize}
		\item	If $P \scong Q$ and  there exist $ \ell$ and $P'$ such that $P \trans{\ell} P'$ then
				there exists $Q'$ such that $Q \trans{\ell} Q'$ and $P' \scong Q'$.
		\item	If $S_1 \scong S_2$ and  there exist $\ell$ and $ S_1'$ such that $S_1 \trans{\ell} S_1'$ then
				there exists $S_2'$ such that $S_2 \trans{\ell} S_2'$ and $S_1' \scong S_2'$.
	\end{itemize}
\end{thm}

\begin{proof}
	The proof is an induction over the definition
	of \scong. For each case of the induction
	the proof considers the derivation of each
	action $\ell$.
\end{proof}

The syntax and the semantics of the \Pcalc %\dk{$\pi$-calculus for privacy}
are encodable in the standard $\pi$-calculus with select and branch
operations~\cite{HondaVK98}. The encoding enjoys a form of operational correspondence.
The details of the translation and operational correspondence
can be found in \appref{enc}.

\begin{exa}

As an example consider the case where a system reads %read and update
private information from a store term.
%
%\begin{itemize}
%	\item
			The derivation for an
			internal transition for the system above is:
	%on a system that reads a store is:
%
			\[\small
				\tree[\LTau] {
					\tree[\LSOut] {
						\store{r}{\id}{c} \trans{\actout{\dual{r}}{\pdata{\id}{c}}} \store{r}{\id}{c}
					}{
						\group{\G_1}{\store{r}{\id}{c}} \trans{\actout{\dual{r}}{\pdata{\id}{c}}} \group{\G_1}{\store{r}{\id}{c}}
					}
					\qquad
					\tree[\LInp] {
						\inp{r}{\pdata{\hid}{x}} P \trans{\actinp{r}{\pdata{\hid}{c}}} P \subst{\pdata{\hid}{c}}{\pdata{\hid}{x}}
					}{
						\group{\G_2}{\inp{r}{\pdata{\hid}{x}} P} \trans{\actinp{r}{\pdata{\hid}{c}}} \group{\G_2}{P \subst{\pdata{\hid}{c}}{\pdata{\hid}{x}}}
					}
					\\[6mm]
					\actout{\dual{r}}{\pdata{\id}{c}} \actdual \actinp{r}{\pdata{\hid}{c}}
				}{
					\group{\G_1}{\store{r}{\id}{c}} \;\Spar \;\group{\G_2}{\inp{r}{\pdata{\hid}{x}} P}
					\trans{\tau}
					\group{\G_1}{\store{r}{\id}{c}} \;\Spar \;\group{\G_2}{P \subst{\pdata{\hid}{c}}{\pdata{\hid}{x}}}
				}
			\]

            We observe an output action \actout{\dual{r}}{\pdata{\id}{c}} on the store
			system which is dual to the input action \actinp{{r}}{\pdata{\hid}{c}}
			executed by the receiving system. The duality of the two actions ensures
			the internal transition of the parallel composition of the two systems.
			The receiving input action has an anonymous identity and thus
			the received data is substituted anonymously.
			
\end{exa}

\section{Privacy Policy}
\label{sec:policy}

In this section we define a simple language for enunciating privacy requirements
of systems defined in our process calculus. Typically, privacy policy languages
express positive and negative norms that are expected to hold in
a system. These norms distinguish what \emph{may} happen, in
the case of a positive norm, and what may not happen, in
the case of a negative norm, on \emph{data attributes} which are types
of sensitive data within a system, and, in particular, how the various agents,
who are referred to by their \emph{roles},
may/may not handle these data.

The notions of an attribute and a role are reflected in our framework via
the notions of private data and groups, respectively. Thus, our  policy language
is defined in such a way as to specify the allowed and disallowed permissions
associated with the various groups for each type of private data. To express this
we make use of a type system to distinguish between the various types of data
within a system and to associate policies to each different type of private
data.

Specifically, we assume a set of ground types, ranged over by $\g$, a set of
private data types, ranged over by $\prv$, and a set of constant types,
ranged over by $\prp$. Based on the above we define the following set of types:
\begin{eqnarray*}
	T \bnfis \privatet{\prv}{\g} \bnfbar \purposet{\prp}{\g} \bnfbar \gtype{\G}{T}
\end{eqnarray*}
Thus, types range over $T$ and include type \privatet{\prv}{\g},
used to type private data and type \purposet{\prp}{\g} used to type constants.
Types are also inductively defined using groups \gtype{\G}{T}.
The intuition behind the notion of
group $\G$ in $\gtype{\G}{T}$ is that a name $x : \gtype{\G}{T}$
may only be used for communicating data of type $T$ between processes
that {\em belong} to group $\G$.
In our privacy setting we map private data of the form
$\pdata{\id}{c}$ to private data types $\privatet{\prv}{\g}$.
This  signifies
the fact that constant $c$ of type $\g$ is associated
with an identity in a piece of private information ($\pdata{\id}{c}$).
The value $\id$ is implicitly assumed to be an
identity without any extra type assignment. As a result, identities
cannot be used as objects in communications or for any other type
of processing within our calculus.

\begin{exa}
	\label{ex:polCar}
	As a running example for this section, consider a system
	that describes some internal modules of a car:
	\[
		\begin{array}{l}
		\group{\Car}{\\
			\begin{array}{rcl}
				&&		\group{\Speedometer}{\;\store{r}{\id}{sp} \Par \repl \newnp{sp}{\out{r}{\pdata{\id}{sp}} \inact} \; }\\
				&\Spar& \group{\SpeedCheck}{\;\inp{r}{\pdata{x}{y}} \IfElse{y > \speedlimit}{P}{Q} \; }
			\end{array}
		\\
		}
		\end{array}
	\]
	The internal module defined by group \Speedometer is a module
	that models the speed of the car. The speed of the car
	is considered to be private data and it is associated with the
	identity, \id, of the driver of the car. Thus,
	$\pdata{\id}{sp}$ has type \privatet{\vehicled}{\speed},
	where \vehicled is the type of private data pertaining to the car and
	\speed is a ground type which corresponds to the speed of the car.	
	
	The module defined by group \SpeedCheck
	is responsible for checking when a car exceeds some speed limit.
	Group \SpeedCheck inputs the current speed of the car (of type
    \privatet{\vehicled}{\speed}) via name $r$ that has type
	\group{\Car}{\privatet{\vehicled}{\speed}} which
	represents a name within group \Car that can receive data of type
	\privatet{\vehicled}{\speed}. Constant \speedlimit has
	type \purposet{\limit}{\speed} and it denotes
	a ground value of type \speed, decorated
	with the constant type \limit, signifying that it is a constant
    that can be used for the purpose of checking satisfaction of
    the speed limit.
\end{exa}

%For example, assuming that $\vehicled$ is the type of
%private data pertaining to a car and $\mathsf{speed}$ is a ground type
%which corresponds to the speed of a car, type $\privatet{\vehicled}{\mathsf{speed}}$
%is the type of private data holding the speed of a car. Further, considering $\Car$
%to be a group describing the environment of a car, type $\gtype{\Car}{\privatet{\vehicled}{\mathsf{speed}}}$
%is the type of a name within group $\Car$ that can be used for communicating
%private data regarding the speed of the car.

%DIMITRI, THIS SEEMS OUT OF PLACE HERE. STILL IT WOULD BE GOOD TO GiVE SOME INTUITION TO G[T], BUT WHERE?

We define privacy policies as an association of permissions
to a hierarchy of groups with respect to types of private data.
The set of permissions we consider for handling private data is
the following:
\begin{eqnarray*}
	\perm	&\bnfis&	\rd \bnfbar \up \bnfbar \rf \bnfbar \diss{G}{\lambda} \bnfbar \str \bnfbar \idp \\
			&\bnfbar&	\ndiss{\kind}\bnfbar \use{\prp} \bnfbar \idf{\prv} %\proc %\idf \bnfbar \use{\pol{\prp}}
						\bnfbar \aggr
	\\
	\lambda	&\bnfis&	1, 2, \dots \bnfbar \infinite
	\\
	\kind	&\bnfis&	\discl \bnfbar \conf \bnfbar \sens
\end{eqnarray*}
Permissions are denoted by \perm and are as follows:
\begin{itemize}
	\item
			Permission \rd when associated with a type of private data and a group indicates
			that the private data may be read by processes belonging to the specific group.

	\item
			Permission \up gives the possibility of updating the contents of a store
			with a new piece of private data.

	\item
			Permission \rf allows to gain access to a
			reference on private data.

	\item
			Permission \diss{\G}{\lambda}, when associated with a type of private data,
			indicates that the private data may be disseminated up to $\lambda$ times to
			processes belonging to group \G, where $\lambda$ ranges over natural numbers and
			\infinite, with \infinite being the unlimited number of times.

	\item
			Permission \str allows to define a store process for the storage of private data.
	\item
			Permission \idp
			allows access to the \id of the private data. Lack of the \idp permission
			corresponds to the process of anonymisation of private data.
	\item
			Permission \ndiss{\kind}, when associated with a certain type
			of data and a group, disallows the dissemination of the specific type of
			data outside the scope of the associated group under the reason described
			by \kind. \kind ranges over \discl that denotes a non-disclosure
			agreement, \conf that denotes a confidentiality relation, and \sens
			that denotes sensitive data.

	\item	Permission \use{\prp} defines the right
			to match private data against constants of type \prp.

	\item	Permission
			\idf{\prv} allows the identification of private data against
			private data with type \prv. Typically, we require that
			\prv characterises data that are unique for each individual,
			such as identity numbers and passwords.

	\item
			Finally, permission \aggr grants the right to aggregate private
			data.
			Data aggregation takes place when a component
			in the system hierarchy is allowed to store
			more than a single piece of private information for the same
            data subject.
\end{itemize}\medskip

%Data usage is denoted through permissions \use{\prp} and \idf{\prv}.

%Permission \idf grants the permission
%of aggregating data for identification purposes. Permission
%\use{\pol{\prp}} states the purposes \pol{\prp} for using the private
%data.

\noindent In turn, a policy is an assignment of permissions
to a hierarchy of groups with
respect to types of sensitive data and
it is defined as follows:
\begin{eqnarray*}
	\Policy &\bnfis& \prv \haspolicy \Hpol \bnfbar \prv \haspolicy \Hpol; \Policy\\
	\Hpol &\bnfis& \role{\G}{\pol{\perm}}{\pol{\Hpol}}
\end{eqnarray*}
where $\pol{\perm}$ is a set of permissions.
A policy has the form $\prv_1 \haspolicy \Hpol_1; \ldots; \prv_n \haspolicy \Hpol_n$
where $\prv_i$ are the base types subject to privacy.
Given  policy $\Policy$ we write $\Policy(\prv)$ for $\Hpol_i$
where $\prv = \prv_i$.
The components $\Hpol_i$, which we refer to as \emph{permission hierarchies},  specify
the group-permission associations for each base type. A permission hierarchy \Hpol has
the form
$\role{\G}{\pol{\perm}}{\Hpol_1,\ldots, \Hpol_m}$
% The component captures a hierarchy of allowed permissions
%on the various groups: the outer group $G$ has allowed permissions $\perm$ and,
%additionally, the hierarchy permissions of the $\Hpol_i$ also hold. Intuitively,
which expresses that an entity possessing a group membership in group \G
has rights $\pol{\perm}$ to the data in question and, if additionally it is a member
of some group $\G_i$, where $\Hpol_i = \role{\G_i}{\pol{\perm}_i}{\ldots}$, then it also has the
rights $\pol{\perm}_i$, and so on.

\begin{exa}
	We define a privacy policy for a network system that utilises
	the internal car modules of our running example (\exref{polCar}) as:
	\[
		\begin{array}{lcl}
			\label{ex:polNet}
			\vehicled &\haspolicy&
				\role{\Network}{}{\\
				&& \qquad \role{\Car}{\diss{\Network}{\infinite}}{\\
				&& \qquad \qquad	\lrole{\Speedometer}{\up, \str},\\
				&& \qquad \qquad	\lrole{\SpeedCheck}{\rd, \idp, \use{\limit}}\\
				&& \qquad },
				\\
				&& \qquad \lrole{\SpeedAuthority}{\rf, \rd}\\
				&& }
		\end{array}
	\]
	The privacy requirements of the policy above are concerned
	with private data of type \vehicled.
	We assume that the modules of a car, expressed here by group \Car,
	are connected to a network (\Network). Part of the network is
	a speed authority.
	The \Speedometer module is responsible for displaying the
	speed of the car. This is modelled by permission \str.
	Permission \up enables the \Speedometer to regularly update the
	speed of the car. The \SpeedCheck module can read the
	speed and use it for the purpose of checking it against the speed limit.
	Any group inside the \Car group can disseminate inside group
	\Network a link to the speed store in the \Speedometer. This is expressed
	by permission \diss{\Network}{\infinite}.
	Finally, the permission \rf of the \SpeedAuthority allows the \SpeedAuthority
	to get a link to the speed's store. The authority also has permission \rd so it
	can read the \Speedometer store.
	
\end{exa}

Following the need to combine permission
environments, we define the following:
Operator $\oplus$ is defined on $\lambda$ values
as the commutative monoid that follows the rules:
\begin{eqnarray*}
	\lambda \oplus \infinite = \infinite
	\qquad \qquad \qquad
	\lambda_1 \oplus \lambda_2 = \lambda_1 + \lambda_2 \quad \lambda_1 \not=\infinite, \lambda_2 \not=\infinite
\end{eqnarray*}
%
%Operator $\oplus$ defines standard integer addition when $\oplus$
%operates on integers with \infinite as the absorbing element. \dk{move somewhere else?}
Operator $\uplus$ is used to combine sets of permissions:
\[
\begin{array}{rcl}
	\pol{\perm}_1 \uplus \pol{\perm}_2 & = &
	\set{\perm \setbar \perm \in \pol{\perm}_1 \cup \pol{\perm}_2, \perm \not= \diss{\G}{\lambda}}\\
%	&\cup&	\set{\diss{\G}{(\lambda_1 \oplus \lambda_2)} \setbar \diss{\G}{\lambda_1} \in \pol{\perm}_1, \diss{\G}{\lambda_2} \in \pol{\perm}_2}\\
	&\cup&	\set{\diss{\G}{(\lambda_1 \oplus \lambda_2)} \setbar \diss{\G}{\lambda_i} \in \pol{\perm}_i, i \in \set{1, 2}}\\
	&\cup&	\set{\diss{\G}{\lambda} \setbar \diss{\G}{\lambda} \in \perm_1, \diss{\G}{\lambda'} \notin \perm_2}\\
	&\cup&	\set{\diss{\G}{\lambda} \setbar \diss{\G}{\lambda} \in \perm_2, \diss{\G}{\lambda'} \notin \perm_1}\\
\end{array}
\]
Operator $\uplus$ takes the union of two permission sets with the
exception of the \diss{\G}{\lambda} permission; whenever we have
two $\diss{\G}{\lambda}$ permissions
the result is the same permission up-to the addition of the $\lambda$
value.

We proceed to define auxiliary functions $\groups{\Hpol}$ and $\permis{\Hpol}$
so as to
gather the sets of groups and the set of permissions, respectively,
inside a hierarchy structure:
%
%\begin{eqnarray*}
%	\groups{\Hpol} & = & \left\{
%	\begin{array}{ll}
%		\set{\G} \cup (\bigcup_{j \in J} \groups{\Hpol_j}) & \textrm{if } \Hpol = \role{\G}{\pol{\perm}}{\Hpol_j}_{j \in J}\\
%       \es & \mbox{if } \Hpol = \epsilon
%	\end{array}
%    \right. \\
%	\permis{\Hpol}& = & \left\{
%	\begin{array}{ll}
%		\pol{\perm} \uplus (\biguplus_{j \in J} \permis{\Hpol_j}) & \textrm{if } \Hpol = \role{\G}{\pol{\perm}}{\Hpol_j}_{j \in J}\\
%       \es & \mbox{if } \Hpol = \epsilon
%	\end{array}
%\right.
%\end{eqnarray*}
%
\begin{eqnarray*}
	\groups{\Hpol}
	& = &
	\set{\G} \cup (\bigcup_{j \in J} \groups{\Hpol_j}) \\
	\permis{\Hpol}
	& = &
	\pol{\perm} \uplus (\biguplus_{j \in J} \permis{\Hpol_j})
\end{eqnarray*}

We are now ready to introduce a well-formedness check on the policy structure.
%We require that policies are well-formed:
%
\begin{defi}[Well-formed Policy]
	We say that a policy
	$\Policy = \prv_1 \haspolicy \Hpol_1; \ldots; \prv_n \haspolicy \Hpol_n$
	is \emph{well formed}, written $\Pol: \diamond$,
	if it satisfies the following:
	\begin{enumerate}
		\item	The $\prv_i$ are distinct.

		\item	If $\Hpol = \role{\G}{\pol{\perm}}{\Hpol_j}_{j \in J}$ occurs within some
				$\Hpol_i$ then $\G \notin \groups{\Hpol_j}$ for all $j\in J$.
%				, that is, the group hierarchy is acyclic.

		\item	\begin{sloppypar} 
				If $\Hpol = \role{\G}{\pol{\perm}}{\Hpol_j}_{j \in J}$ occurs within some
				$\Hpol_i$, $\ndiss{\kind}\in\pol{\perm}$ and
				$\diss{\G'}{\lambda} \in \permis{\Hpol_j}$ for some $j\in J$, then $\G'\in\groups{\Hpol}$.
				\end{sloppypar}
%				In words, no non-disclosure requirement imposed at some level of a
%				hierarchy can be in conflict with a disclosure requirement granted in its sub-hierarchy.
	\end{enumerate}
\end{defi}\medskip
\noindent A well-formed policy is required to have definition for distinct private types,
an acyclic group hierarchy and,
furthermore, no non-disclosure requirement imposed at some level of a
hierarchy can be in conflict with a disclosure requirement granted in its sub-hierarchy.
Hereafter, we assume that policies are well-formed policies.
As a shorthand, we write $\G: \pol{\perm}$ for  $\role{\G}{\pol{\perm}}{\epsilon}$ and
we abbreviate $\G$ for $\G:\emptyset$.

Given a set of groups $\pol{G}$ and hierarchy $\Hpol$, it is often convenient to
extract a flat structure of the policy hierarchy referring solely to the groups $\pol{G}$.
This flat hierarchy captures the nesting of the groups, as defined by the hierarchy,
and accumulates the set of permissions associated to agents belonging to the groups in
question. Thus, a flat hierarchy has the form
\begin{eqnarray*}
	\Inthier \bnfis \hiert{\G}{\Inthier} \bnfbar \hiert{\G}{\pol{\perm}}
\end{eqnarray*}
Precisely, given a set of groups $\pol{G}$ and a hierarchy $\Hpol$, we  define the flat hierarchy $\Hpol_{\pol{G}}$ as follows:
\begin{eqnarray*}
	\Hpol_{\G} = \hiert{\G}{\pol{\perm}} & \textrm{if} & \Hpol = \lrole{\G}{\pol{\perm}}\\
%	\Hpol_{\G\cat\pol{\G}} = \hiert{\G}{\Hpol'_{\pol{\G}}} & \textrm{if} &	\Hpol = \role{\G}{\pol{\perm}}{\pol{\Hpol}'} \land
%																		\role{\G'}{\pol{\perm}'}{\pol{\Hpol}_1} \in \pol{\Hpol}' \land
%																		\Hpol' = \role{\G'}{\pol{\perm} \uplus \pol{\perm}'}{\pol{\Hpol}_1}
	(\role{\G}{\pol{\perm}}{\pol{\Hpol}})_{\G\cat\pol{\G}} = \hiert{\G}{\Inthier}
			& \textrm{if} &
%			\Hpol = \role{\G}{\pol{\perm}}{\pol{\Hpol}'} \land
			\role{\G'}{\pol{\perm}'}{\pol{\Hpol}'} \in \pol{\Hpol} \land
			\Inthier = (\role{\G'}{\pol{\perm} \uplus \pol{\perm}'}{\pol{\Hpol}'})_{\pol{\G}}
%			\Hpol' = \role{\G'}{\pol{\perm} \uplus \pol{\perm}'}{\pol{\Hpol}'}
\end{eqnarray*}
%
%DIMITRI I'M FINDING THIS DEFINITION DIFFICULT TO COMPREHEND. CAN YOU MAKE IT CLEARER?
%
%Flat hierarchy $\Hpol_{\pol{\G}}$ expresses the path in hierarchy
%\Hpol that follows the nodes in \pol{\G} and at the same time
%it gathers the permissions of the path up to the $\uplus$ operator.

\begin{exa}
Returning to our running example, let the policy in \exref{polNet} to be $\Policy = \vehicled \haspolicy \Hpol$.
Furthermore, if $\pol{\G} = \Network\cat\Car\cat\SpeedCheck$ then:
\[
	\Hpol_{\pol{\G}} = \hiert{\Network}{\hiert{\Car}{\hiert{\SpeedCheck}{\diss{\Network}{\infinite} \cat \rd \cat \idp \cat \use{\limit}}}}
\]
\end{exa}

%For example let $\Hpol = \role{\G}{\pol{\perm}}{\lrole{\G_1}{\pol{\perm_1}}, \lrole{\G_2}{\pol{\perm_2}}}$.
%Then $\Hpol_{\G \cat \G_1} = \hiert{\G}{\hiert{\G_1}{\pol{\perm} \uplus \pol{\perm}_1}}$ and
%$\Hpol_{\G \cat \G_2} = \hiert{\G}{\hiert{\G_2}{\pol{\perm} \uplus \pol{\perm}_2}}$.

The operators
\groups{\cdot} and \permis{\cdot} are extended to include \Inthier structures:
\[
\begin{array}{rclcrcl}
	\groups{\hiert{\G}{\pol{\perm}}} &=& \set{\G}
	& \qquad \qquad &
	\groups{\hiert{\G}{\Inthier}} &=& \set{\G} \cup \groups{\Inthier}
	\\[3mm]
	\permis{\hiert{\G}{\pol{\perm}}} &=& \pol{\perm}
	& \qquad \qquad &
	\permis{\hiert{\G}{\Inthier}} &=& \permis{\Inthier}
\end{array}
\]

\section{Type System}

The goal of this paper is to statically ensure that a process
respects a privacy policy. To achieve this we develop a typing system.
The key idea is that the type system performs a static analysis
on the syntax of processes and derives the behaviour of each process
with respect to data collection, data processing, and data
dissemination as a type. Then the derived type is checked for
satisfaction against a privacy policy.

We partition the typing rules into three categories:
i) typing rules for values and expressions;
ii) typing rules for processes; and
iii) typing rules for systems.
We begin by defining the typing environments on which type checking is carried out.

\subsection{Typing environments}
We first define the typing environments.
\begin{eqnarray*}
	\Ga	&\bnfis&	\Ga \cat u: \gtype{\G}{T} \bnfbar
						\Ga \cat \pdata{\ii}{\con}: \privatet{\prv}{\g} \bnfbar
						\Ga \cat \con: \privatet{\prp}{\g} \bnfbar
						\es
	\\
	\De	&\bnfis&	\De \cat \prv : \pol{\perm} \bnfbar
						\es
\end{eqnarray*}
Type environment \Ga maps names and constants to appropriate types.
Specifically, names $u$ are mapped to types of the form $\gtype{\G}{T}$.
%i.e.~names $u$ cannot carry private data or constants.\dk{the last sentence
%is not true!}
References are mapped to private data types of the form \privatet{\prv}{\g} and
constants are mapped to purpose types of the form \privatet{\prp}{\g}.
Permission environment \De assigns a set
of privacy permissions to types of private data. Permission environment \De
is intended to be used for conformance
against a privacy policy. %Intuitively, assignment
%of permissions to a data type \prv can be understood as follows: \ldots.

Following the need to combine $\De$ environments, for extracting the interface
of parallel processes we extend operator $\uplus$, previously defined for
sets of permissions,  to permission environments:
\[
\begin{array}{rcl}
	\De_1 \uplus \De_2 =
	\set{\prv: \pol{\perm}_1 \uplus \pol{\perm}_2 \setbar \prv: \pol{\perm}_1 \in \De_1 ,  \prv: \pol{\perm}_2 \in \De_2}
	{\ \cup\  \De_1\backslash\De_2\ \cup\ \De_2\backslash\De_1}
\end{array}
\]
At the level of permission environments $\De$, operator $\uplus$ is in fact the union of
the two environments with the exception of common private types where
the associated permission sets are combined with the $\uplus$ operator.

%%%%%%%%%%%%%%%%%%%%%%%%
%     Values
%%%%%%%%%%%%%%%%%%%%%%%%

\subsection{A type system for values and expressions}

We present the typing rules for values and expressions.
Depending on what kind of values and expressions processes
use we can derive how a process handles private data.

We use a typing judgement for values $v$ and a typing judgement
for the matching expression \match{v_1}{v_2}:
\[
	\Ga \proves v: T \hastype \De
	\qquad \qquad \qquad
	\Ga \proves \match{v_1}{v_2} \hastype \De
\]
Before we define the typing rules
we present the auxiliary function \identify{\id}:
\[
	\identify{\id} = \set{\idp}
	\qquad \qquad
	\identify{x} = \set{\idp}
	\qquad \qquad
	\identify{\hid} = \es
\]
The \identify{\ii} function is used
to derive the \idp permission out
of an \id; a visible \id maps to
the \idp permission, while a hidden
identity, \hid, maps to no permissions.

\begin{figure}[t]
\begin{mathpar}
		\inferrule*[left=\TName] {
			%u = r \iff T = \gtype{\G}{\privatet{\prv}{\g}}
		}{
			\Ga \cat u: T \proves u: T \hastype \es
		}
		\and
		\inferrule*[left=\TPd] {
		}{
			\Ga \cat \pdata{\ii}{\con}: \privatet{\prv}{\g} \proves \pdata{\ii}{\con}: \privatet{\prv}{\g} \hastype \{\prv: \identify{\ii}\}
		}
		\and
		\inferrule*[left=\TCons] {
		}{
			%\Ga \cat \con : \purposet{\prp}{\g} \proves \con : \purposet{\prp}{\g}\;\hastype\; \es
			\Ga \cat \con : T \proves \con : T\hastype \es
		}
		\and

		\inferrule*[left=\TId] {
			\Ga \proves \pdata{\id}{\con_1} : \privatet{\prv_1}{\g} \hastype \De_1
			\and
			\Ga \proves \pdata{\hid}{\con_2} : \privatet{\prv_2}{\g} \hastype \De_2
		}{
			\Ga \proves \match{\con_1}{\con_2} \hastype \De_1 \uplus \De_2 \uplus \{\prv_1: \set{\idf{\prv_2}}\}
		}
		\and
		\inferrule*[left=\TUse] {
			\ii \not= \hid
			\quad
			\Ga \cat \con_2: \purposet{\prp}{\g} \proves \pdata{\ii}{\con_1}: \privatet{\prv}{\g} \hastype \De
		}{
			\Ga \cat \con_2: \purposet{\prp}{\g} \proves \match{\con_1}{\con_2} \hastype \De \uplus \{\prv: \set{\use{\prp}}\}
		}
		\and
		\inferrule*[left=\TEqP] {
			\Ga \proves \pdata{\id_1}{\con_1} : \privatet{\prv}{\g} \hastype \De_1
			\and
			\Ga \proves \pdata{\id_2}{\con_2} : \privatet{\prv}{\g} \hastype \De_2
		}{
			\Ga \proves \match{\pdata{\id_1}{\con_1}}{\pdata{\id_2}{\con_2}} \hastype \De_1 \uplus \De_2
		}
		\and
		\inferrule*[left=\TEqA] {
			\Ga \proves \pdata{\hid}{\con_1} : \privatet{\prv}{\g} \hastype \De_1
			\and
			\Ga \proves \pdata{\hid}{\con_2} : \privatet{\prv}{\g} \hastype \De_2
		}{
			\Ga \proves \match{\pdata{\hid}{\con_1}}{\pdata{\hid}{\con_2}} \hastype \De_1 \uplus \De_2
		}
\end{mathpar}
\caption{Typing rules for values and expressions}
\label{fig:types_value}
\end{figure}

\figref{types_value} defines the typing rules
for values and expressions.
Rules \TName, \TCons, and \TPd type, respectively, names $u$, constants $\con$,
and private data \pdata{\ii}{\con} with respect to type environment \Ga.
Rule \TPd is also used to check whether a process has access to
the identity of the private data via the definition of $\identify{\ii}$.
%Rule \TPMon types a monadic private data structure given a polyadic
%private data structure.
%Rule \TPCons is used to type constants inside
%private data structures.
Rule for identification \TId types a matching operation: the key
idea is that through matching between data whose identity is known
and data whose identity is unknown, an identification may be performed.
For example, if we let:
\[
	\Ga = \pdata{\john}{\mathsf{dna}_1} : \privatet{\patientd}{\mathsf{DNA}} \cat \pdata{\hid}{\mathsf{dna}_2}: \privatet{\crime}{\mathsf{DNA}}
\]
then a \forensics system defined as $\group{\forensics}{\ifelse{\mathsf{dna}_1}{\mathsf{dna}_2}{P}{Q}}$
by performing the comparison ${\mathsf{dna}_1}={\mathsf{dna}_2}$, may identify that
the DNA obtained at a crime scene is identical to \john's DNA and thus perform an identification
for a crime investigation. Thus, the type system, and rule \TId in particular, will deduce that
\[
	\Ga \proves {\mathsf{dna}_1}={\mathsf{dna}_2} \hastype \patientd : \idf{\crime}
\]
This situation requires that the \forensics process has permission to identify
based on the private data of type \patientd.

%For example let:
%\[
%	\Ga = \pdata{\mathsf{john}}{n_1}: \privatet{\mathsf{password}}{\mathsf{pin}}
%	\cat n_2: \privatet{\mathsf{credit\ card}}{\mathsf{pin}}
%\]
%that defines the $\mathsf{pin}$ password for $\mathsf{john}$'s credit card in \Ga,
%and the $\mathsf{pin}$ for the $\mathsf{credit\ card}$.
%We can then let an $\mathsf{ATM}$ to use the private data of $\mathsf{john}$
%to access its $\mathsf{credit\ card}$:
%\[
%	\Ga \proves \group{\mathsf{ATM}}{\ifelse{n_1}{n_2}{P}{Q}} \hastype \mathsf{password} : \use{\mathsf{credit\ card}}
%\]

Rule \TUse defines private data usage.
We assume that the usage of private data is always reduced to
a name matching operation of private data over constant
data. For example assume:
\[
	\Ga = \pdata{\john}{\mathsf{dna}_1}: \privatet{\patientd}{\mathsf{DNA}}, \mathsf{dna}_2: \purposet{\diagnosis}{\mathsf{DNA}}
\]
Then a doctor defined as $\group{\Doctor}{\ifelse{\mathsf{dna}_1}{\mathsf{dna}_2}{P}{Q}}$, may use \john's \patientd for the purpose
of performing a \diagnosis. In particular, rule \TUse allows us to deduce that
\[
	\Ga \proves  {\mathsf{dna}_1}={\mathsf{dna}_2}\hastype \patientd : \use{\diagnosis} %\use{\diagnosis}
\]

Finally, rules \TEqP and \TEqA perform type matching of private data against each other
and anonymised data against each other. We take the stance that such comparisons
do not exercise any permissions though, we believe, that this is a question for
further investigation.
%%%%%%%%%%%%%%%%%%%%%%%
%   Terms
%%%%%%%%%%%%%%%%%%%%%%%

\subsection{A type system for process terms}

Types for processes rely on a linear environment $\La$ and a store environment $Z$:
\begin{eqnarray*}
	\La	&\bnfis&	\La \cat r \bnfbar \es\\
    Z &  \bnfis&    Z\cat \ptuple{\ii}{\prv} \bnfbar \es
\end{eqnarray*}
Environment $\La$ accumulates the references to stores that are present in a process
and its being linear captures that we cannot use the same reference in more than
one store.
Thus,~$\La_1,\La_2$ is defined
only when $\La_1$ and $\La_2$ are disjoint.
In turn, environment $Z$ contains the identifiers whose private
data is stored in a system along with the respective private type of the
stored data.
Thus, typing judgements have the form
\[
	\Ga; \La;Z \proves P \hastype \De %\qquad \qquad \Ga \proves v : T \hastype \De %\qquad \qquad \Ga \proves v \hastype \De
\]
which essentially states that given a type environment $\Ga$, a reference environment $\La$
and a store environment, $Z$, process $P$ is well typed and
produces a permission environment \De. %We use the following notations:

\begin{comment}\begin{eqnarray*}
	\UDe{T} &=&
	\left\{
		\begin{array}{lcl}
			\prv: \set{\up}			&\text{ if }&	T = \privatet{\prv}{\g}\\
			\prv: \set{\diss{\G}{1}}	&\text{ if }&	T = \gtype{\G}{\privatet{\prv}{\g}}\\
			\prv: \es				&&				\text{otherwise}
		\end{array}
	\right.
	\\
	\RDe{T} &=&
	\left\{
		\begin{array}{lcl}
			\prv: \set{\rd}	&\text{ if }&	T = \privatet{\prv}{\g}\\
			\prv: \set{\rf}	&\text{ if }&	T = \gtype{\G}{\privatet{\prv}{\g}}\\
			\prv: \es		&&				\text{otherwise}
		\end{array}
	\right.
\end{eqnarray*}
Further, we define operator \ReplDe as:
\end{comment}
%
The typing rules for processes are defined in \figref{types_processes}.
Rule \TInact states that the inactive process produces an empty permission
environment and rule \TSt  that a store process produces a permission
environment which contains the $\str$ permission in addition to any
further permissions associated with the stored private data.

\newcommand{\aggrt}[1]{ {\color{red} #1} }

\begin{figure}
	\begin{mathpar}
		\inferrule*[left=\TInact] {
		}{
			\Ga; \es; Z \proves \inact \hastype \es
		}
		\and
		{
		\inferrule*[left=\TSt] {
			\Ga \proves r: \gtype{\G}{\privatet{\prv}{\g}} \hastype \es
			\and
			\Ga \proves \pdata{\ii}{\con}: \privatet{\prv}{\g} \hastype \De'
		}{
			\Ga; r; \ptuple{\ii}{\prv} \proves \store{r}{\ii}{\con} \hastype \prv: \set{\str}\uplus\De'
		}
		}
		\and
		\inferrule*[left=\TOut] {
				\begin{array}{c}
					\Ga; \La; Z \proves P \hastype \De_1
					\\
					\Ga \proves u: \gtype{\G}{T} \hastype \es
					\quad
					\Ga \proves v: T \hastype \De_2
				\end{array}
		}{
			\Ga; \La; Z \proves \out{u}{v} P \hastype \De_1 \uplus \De_2 \uplus \UDe{T}
		}
		\and
		\inferrule*[left=\TInp] {
			\begin{array}{c}
				\Ga; \La; Z \proves P \hastype \De_1
				\\
				\Ga \proves u: \gtype{\G}{T} \hastype \es
				\quad
				\Ga \proves k: T \hastype \De_2
			\end{array}
		}{
			\Ga\backslash k; \La\backslash k; Z \proves \inp{u}{k} P \hastype \De_1 \uplus \De_2 \uplus \RDe{T}
		}
		\and
		\inferrule*[left=\TRes] {
			\Ga; \La; Z \proves P \hastype \De
		}{
			\Ga\backslash\set{n}; \La\backslash\set{n}; Z \proves \newn{n} P \hastype \De
		}
		\and
		{
		\inferrule*[left=\TRepl] {
			\begin{array}{c}
				\Ga; \es; Z \proves P \hastype \De
				\\
				%\exists \Delta' \suchthat (\forall \ii: \prv, \ii: \prv \in Z \iff \prv: \set{\aggr} \in \Delta')
                \Delta' = \{\prv: \set{\aggr} \setbar \ptuple{\ii}{\prv} \in Z\}
			\end{array}
		}{
			\Ga; \es; Z \proves \repl P \hastype \ReplDe \uplus \Delta'
		}
		}
		\and
		\inferrule*[left=\TPar] {
			\begin{array}{c}
				\Ga; \La_i; Z_i \proves P_i \hastype \De_i
				\quad
				i \in \set{1, 2}
				\\
				\begin{array}{rcl}
				\Delta &=& \set{\prv: \set{\aggr} \setbar\\
												& & \ptuple{\ii}{\prv} \in Z_i \land \ptuple{\ii'}{\prv'} \in Z_j \land
												[\ii = \ii'\lor \ii = x\lor \ii'= x ]\land \set{i,j} = \{1,2\}}
				\end{array}
			\end{array}
		}{
			\Ga; \La_1 \cat \La_2; Z_1 \cat Z_2 \proves P_1\Par P_2\hastype\De_1 \uplus \De_2 \uplus \De
		}
		%\Delta = \{\prv: \set{\aggr}\mid \langle \ii,\t\rangle\in Z_i \land [\langle \ii,\t'\rangle \in Z_j \lor
		%                           \langle x,\t'\rangle \in Z_j \lor  \in \domain{Z_j} ]\land \{i,j\} = \{1,2\}\}\\
%					\exists \Delta \suchthat & (\forall i, j, i \not= j, \forall \ii: \prv \in Z_i, \exists x,\\
%%
%								&	(\ii \in \domain{Z_j} \lor x \in \domain{Z_j})) \lor (\ii = x \land Z_j \not= \es))\\
%								&	\iff	\prv: \set{\aggr} \in \Delta
%%					\exists \Delta	\suchthat (\forall \ii: \prv,  & & (\ii: \prv \in Z_1 \cup Z_2 \land \ii \in \domain{Z_1} \cap \domain{Z_2}) \\
%%									 &\lor& (\ii: \prv \in Z_i \land \exists x, x \in \domain{Z_1} \cup \domain{Z_2}))
%%									\iff \prv: \set{\aggr} \in \Delta
%%					\\
%%					\exists Z'	&\suchthat& && (\not\exists x: \prv \in \domain{Z_1} \cup \domain{Z_2})
%%								\iff Z' = Z_1 \cup Z_2\\
%%								&& &\vee& (\exists x: \prv \in Z_1 \cup Z_2)
%%								\iff Z' = \set{x:\prv}\\
%				\end{array}
%			\end{array}
%		}{
%			\Ga; \La_1 \cat \La_2; Z_1 \cat Z_2 \proves P_1\Par P_2\hastype\De_1 \uplus \De_2 \uplus \De	}
%		}
%
		\and
		\inferrule*[left=\TIf] {
			\Ga; \La_i; Z_i \proves P_i \hastype \De_i
			\and
			i \in \set{1, 2}
			\and
			\Ga \proves \match{v_1}{v_2} \hastype \De
		}{
			\Ga; \La_1 \cat \La_2; Z_1 \cat Z_2 \proves \ifelse{v_1}{v_2}{P_1}{P_2} \hastype \De \uplus \De_1 \uplus \De_2 %\uplus \De_3
		}
	\end{mathpar}
\caption{Typing rules for processes}
\label{fig:types_processes}
\end{figure}

Rule \TOut types the output-prefixed
process:  If environment $\Ga;\La;Z$ produces $\Delta_1$ as a permission interface of $P$,
$u$ and $v$ have compatible types according to $\Ga$, and $v$ produces a permission
interface $\Delta_2$,
we conclude that the process $\out{u}{v} P$ produces
an interface where the combination of interfaces $\Delta_1$ and
$\Delta_2$ is extended with the permissions $\UDe{T}$, where
\begin{eqnarray*}
	\UDe{T} &=&
	\left\{
		\begin{array}{lcl}
			\prv: \set{\up}			&\text{ if }&	T = \privatet{\prv}{\g}\\
			\prv: \set{\diss{\G}{1}}	&\text{ if }&	T = \gtype{\G}{\privatet{\prv}{\g}}\\
			\prv: \es				&&				\text{otherwise}
		\end{array}
	\right.
\end{eqnarray*}
That is,  (i) if $T$
is private type \privatet{\prv}{\g} then $\UDe{T}$ is the permission
$\prv:\up$ since the process is writing
an object of type \privatet{\prv}{\g}, (ii) if $T=\gtype{\G}{ \privatet{\prv}{\g}}$ then $\UDe{T}$
is the permission $\diss{\G}{1}$ since the process is  disclosing once a link to private data via a channel of group $\G$,
and (iii) the empty set of permissions otherwise.

Rule \TInp is similar, except that the permission interface generated contains
the permissions exercised by process $P$ by the input value $k$ along
with permissions $\RDe{T}$, where set  $\RDe{T}$ is defined by:
\begin{eqnarray*}
	\RDe{T} &=&
	\left\{
		\begin{array}{lcl}
			\prv: \set{\rd}	&\text{ if }&	T = \privatet{\prv}{\g}\\
			\prv: \set{\rf}	&\text{ if }&	T = \gtype{\G}{\privatet{\prv}{\g}}\\
			\prv: \es		&&				\text{otherwise}
		\end{array}
	\right.
\end{eqnarray*}
That is  $\RDe{T}$ contains (i) permission
$\rd$ if $T$ is private type \privatet{\prv}{\g} since the process is reading
an object of type \privatet{\prv}{\g}, (ii) permission $\rf$ if
$T=\gtype{\G}{ \privatet{\prv}{\g}}$ since in this case the process is reading
a reference to a private object and
(iii) the empty set otherwise.

For name restriction, \TRes specifies that if $P$ type checks
within an environment $\Ga;\La$, then $(\nu n) P$
type checks in an environment $\Ga \backslash\set{n};\La\backslash\set{n};Z$, i.e.~without name $n$.
Moving on to rule \TRepl we have that
if a process $P$ produces an interface $\Delta$ then  $\repl{P}$ produces
an interface $\ReplDe\uplus \Delta'$, where
\begin{eqnarray*}
	\ReplDe	&=&\set{\perm \setbar \perm \in \De , \perm \not= \diss{\G}{\lambda}}\\
			&\cup&		\set{\diss{\G}{\infinite} \setbar \diss{\G}{\lambda} \in \De}
\end{eqnarray*}
In words,  (i) $\ReplDe$ is such that if a type
is disclosed $\lambda >0$ times in $\Delta$ then it is disclosed for an unlimited number of times
in $\ReplDe$ and (ii) $\De'$ contains the permission \aggr for all stores of type $\prv$ in
$P$. This is because the store replication may allow a system component
to store multiple private data for the same individual.
In turn, rules \TPar uses the $\uplus$ operator to compose the process
interfaces of $P_1$ and $P_2$ while additionally including any
aggregation being exercised by the two systems. Such aggregation takes place if
the parallel components have the capacity of storing more than one piece of information for the
same individual. This is the case either if they possess more than one store for the same identifier,
or if they have more than one store and one of the stores has not yet been instantiated
(i.e. the identity field is a variable),
thus, has the capacity of user-based data aggregation.

Finally, rule \TIf extracts the permission interface of a conditional process
as the collection of the permissions exercised by each of the component processes as well as those
exercised by the condition. Note that in this and the previous rule, $\Lambda_1$ and $\Lambda_2$
are assumed to be disjoint, and $Z_1,Z_2$ represent the concatenation of the two environments.

\subsection{A type system for system terms}

In order to produce a permission interface for systems, we employ the following typing
environment which extends  the typing environment $\Delta$ with information regarding
the groups being active while executing permissions for different types of private
data. Specifically, we write:
\begin{eqnarray*}
	\Int	&\bnfis&	\Int \cat \prv: \Inthier \bnfbar \es %\Htheta \bnfbar \es
\end{eqnarray*}
%where
%\[
%	\Inthier	\bnfis	\hiert{\G}{\Inthier} \bnfbar \hiert{\G}{\pol{\perm}}
%\]
Thus, $\Inthier$ is a flat hierarchy  of the form
$\hiert{\G_1}{\hiert{\G_2}{ \ldots \hiert{\G_n}{\pol{\perm}} \ldots}}$
%$\G_1[\G_2[\ldots \G_n[\pol{\perm}]\ldots] ]$
associating a sequence of groups with a set of permissions, and $\Int$ is a
collection of such associations for different private types. Structure \Inthier is called
interface hierarchy.
The typing judgement for systems  becomes:
\[
	\Ga; \La \proves S \hastype \Int
\]
and captures that system $S$ is well-typed in type environment $\Ga$ and
linear store environment $\La$ and produces the interface $\Int$ which
records the group memberships of all components of $S$ as well
as the permissions exercised by each of the components.  \figref{types_systems}
presents the associated typing rules. The first rule employed
at the system level is rule \TGr
according to which, if $P$ produces a typing $\De$, then system $\group{\G}{P}$
produces the $\Int$-interface where group $\G$  is applied to all components of $\De$.
This captures that there exists a component in the system belonging to group $G$ and
exercising permissions as defined by $\Delta$. For rule \TSGr, we have that if
$S$ produces a typing interface $\Int$
then process $\group{\G}{S}$  produces
a new interface where all components of $\Int$ are extended by adding group
$\G$ to the group membership of all components.
Next, for the parallel composition of systems, rule \TSPar,
concatenates the system interfaces of $S_1$ and $S_2$, thus collecting the
permissions exercised by each of components of both of $S_1$ and $S_2$. Finally,
for name restriction, \TSRes specifies that if $S$ type checks
within an environment $\Ga;\La$, then $(\nu n) S$
type checks in environment $\Ga\backslash\set{n};\La\backslash\set{n}$.

\begin{figure}
	\begin{mathpar}
		\inferrule*[left=\TGr] {
			\begin{array}{c}
				\Ga; \La;Z \proves P \hastype \De
				\\
				 \Int = \{\prv: \hiert{\G}{\pol{\perm}} \mid \prv: \pol{\perm} \in \De\}
			\end{array}
		}{
			\Ga; \La \proves \group{\G}{P} \hastype \Int
		}
		\and
		\inferrule*[left=\TSGr] {
			\begin{array}{c}
				\Ga; \La \proves S \hastype \Int
				\\
				\Int' =\{\prv: \hiert{\G}{\Inthier}\mid \prv: \Inthier \in \Int\}
			\end{array}
		}{
			\Ga; \La \proves \group{\G}{S} \hastype \Int'
		}
		\and
		\inferrule*[left=\TSPar] {
			\Ga; \La_i \proves S_i \hastype \Int_i
			\quad
			i \in \set{1, 2}
%			\Ga; \La \proves S_2 \hastype \Int_2
		}{
			\Ga; \La_1, \La_2 \proves S_1 \Spar S_2 \hastype \Int_1 \cat \Int_2
		}
		\and
		\inferrule*[left=\TSRes] {
			\Ga; \La \proves S \hastype \Int
		}{
			\Ga\backslash\set{n}; \La\backslash\set{n} \proves \newn{n} S \hastype \Int
		}
	\end{mathpar}
\caption{Typing rules for systems}
\label{fig:types_systems}
\end{figure}

\begin{exa}
	As an example we type the \Lab system from \secref{intro_groups}:
	\[
		\group{\Lab}{\inp{b}{w} \inp{w}{\pdata{x}{y}} \inp{r}{\pdata{\hid}{z}} \ifelse{y}{z}{ \out{c}{w} \inact }{\inact}}\\
	\]
	Consider also a type environment:
	\begin{eqnarray*}
		\Ga	&=&	r: \gtype{\Police}{\privatet{\crime}{\dna}} \cat
				\pdata{\hid}{z}: \privatet{\crime}{\dna} \cat\
				\\
			&&	w: \gtype{\Hospital}{\privatet{\patientd}{\dna}} \cat
				\pdata{x}{y}: \privatet{\patientd}{\dna}
				\\
			&&	b: \gtype{\Hospital}{\gtype{\Hospital}{\privatet{\crime}{\dna}}} \cat
				c: \gtype{\Police}{\gtype{\Hospital}{\privatet{\crime}{\dna}}}
	\end{eqnarray*}
	In system \Lab,
	name $r$ and variable $w$ have a reference type and are used to access
	private data of type \crime and \patientd, respectively.
	{Name $r$ will be used to receive data from the police database.}
	Placeholders \pdata{\hid}{z} and \pdata{x}{y} %are used to
	instantiate private data of type \crime and \patientd, respectively.
	Name $b$ is used to substitute reference variable $w$ and
	name $c$ is used to send reference $w$ to the group \Police.

	To type the matching expression we use typing rule \TId:
	\[
		\tree[\TId] {
			\Ga \proves \pdata{x}{y}: \privatet{\patientd}{\dna} \hastype \patientd: \set{\idp}
			\\
			\Ga \proves \pdata{\hid}{z}: \privatet{\crime}{\dna} \hastype \crime: \es
		}{
			\Gamma \proves \match{y}{z} \hastype \crime: \set{\idf{\patientd}} \cat \patientd: \set{\idp}
		}
	\]
	In the matching expression above we have private data placeholder $\pdata{\hid}{z}$
	to be of type \privatet{\crime}{\dna} (private data that is anonymised) and
	placeholder \pdata{x}{y} of type \privatet{\patientd}{\dna} with $x$ being
	known (permission \idp).
	Both variables
	$z$ and $y$ represent the \dna type and can be matched against each other.
	The fact that $z$ comes from an anonymous private data and $y$ data has
	a known data subject implies the process of identification of \crime
	data against \patientd, which is expressed as mapping \crime: \set{\idf{\patientd}}
	in the permission typing.

	The conditional term is typed using typing rule \TIf:
	\[
		\tree[\TIf]{
			\Ga \proves \match{y}{z} \hastype \crime: \set{\idf{\patientd}} \cat \patientd: \set{\idp}
			\\
			\Ga; \es; \es \proves \out{c}{w} \inact \hastype \patientd: \set{\diss{\Police}{1}}
		}{
			\Ga; \es; \es \proves \ifelse{y}{z}{ \out{c}{w} \inact }{\inact} \hastype \De
		}
	\]
	with $\De = \crime: \set{\idf{\patientd}} \cat \patientd: \set{\idp \cat \diss{\Police}{1}}$.
	The conditional term types the two branches of the term and
	identifies the kind of data processing in the matching operator.
	The branch $\out{b}{w} \inact$ results in the permission
	environment \patientd: \set{\diss{\Police}{1}} because
	reference $w$ is being disseminated.
	The resulting permission environment is the $\uplus$-union of
	the three former permission environments.

	The whole process of the system is typed as:
	\[
		\tree {
			\Ga; \es; \es \proves \inp{w}{\pdata{x}{y}} \inp{r}{\pdata{\hid}{z}} \ifelse{y}{z}{ \out{c}{w} \inact }{\inact} \hastype \De_1
		}{
			\Ga; \es; \es \proves \inp{b}{w} \inp{w}{\pdata{x}{y}} \inp{r}{\pdata{\hid}{z}} \ifelse{y}{z}{ \out{c}{w} \inact }{\inact} \hastype \De_1 \uplus \De_2
		}
	\]
	with
	\begin{eqnarray*}
	\De_1& \!\!= &\!\!\patientd{:} \set{\rd \cat \idp \cat \diss{\Police}{1}} \cat \crime{:} \set{\rd \cat \idf{\patientd}}\\
	\De_2 &\!\!= &\!\!\patientd{:} \set{\rf}
	\end{eqnarray*}
	Finally, we use typing rule \TSGr to type the whole system:
	\[
		\tree[\TSGr] {
			\Ga; \es; \es \proves \inp{b}{w} \inp{w}{\pdata{x}{y}} \inp{r}{\pdata{\hid}{z}} \ifelse{y}{z}{ \out{c}{w} \inact }{\inact} \hastype \De_1 \uplus \De_2
		}{
			\Ga; \es \proves \group{\Lab}{\inp{b}{w} \inp{w}{\pdata{x}{y}} \inp{r}{\pdata{\hid}{z}} \ifelse{y}{z}{ \out{c}{w} \inact }{\inact}} \hastype \Int
		}
	\]
	with
	\[
			\begin{array}{rcl}
				\Int	&=&	\patientd: \hiert{\Lab}{\rf \cat \rd \cat \idp \cat \diss{\Police}{1}} \cat\\
						&&	\crime: \hiert{\Lab}{\rd \cat \idf{\patientd}}
			\end{array}
	\]
	Rule \TSGr constructs a permission interface \Int
	using the typing environment from the containing process.
	In this case the rule constructs two interface hierarchies
	on group hierarchy \Lab using the permissions for
	\patientd and \crime in the process's typing environment.
\end{exa}

\section{Soundness and Safety}

In this section we establish soundness and safety results for
our framework.
The first result we establish is the result
of type preservation where we show that the type system
is preserved by the labelled transition system.
Type preservation is then used to prove a safety property
of a system with respect to privacy policies.
A basic notion we introduce in this section is the notion
of policy satisfaction which we show to be preserved
by the semantics. Finally, we show that a safe system
would never reduce to an error system through a safety theorem.
%\dk{Missing proofs of results can be found in the appendix.}

\subsection{Type Preservation}

We first establish standard weakening and strengthening lemmas.

\begin{lem}[Weakening]\label{lem:weaken}
	$ $
	\begin{itemize}
		\item	Let $P$ be a process with $\Ga; \La; Z \proves P \hastype \De$
				and $v \notin \fn{P}$. Then $\Ga \cat v: T; \La; Z \proves P \hastype \De$.
				
		\item	Let $P$ be a process with $\Ga; \La; Z \proves P \hastype \De$
		and $\fv{k} \cap \fv{P} = \es$. Then $\Ga \cat k: T; \La; Z \proves P \hastype \De$.

		\item	Let $S$ be a system with $\Ga; \La \proves S \hastype \Int$ and
				$v \notin \fn{S}$. Then $\Ga \cat v: T; \La \proves S  \hastype \Int$.
		\item	Let $S$ be a system with $\Ga; \La \proves S \hastype \Int$ and
				$\fv{k} \cap \fv{S}$ = \es. Then $\Ga \cat k: T; \La \proves S \hastype \Int$.
	\end{itemize}
\end{lem}

\begin{proof}
	The proof for Part 1 and 2 (resp., Part 3 and 4) is a standard induction on the typing derivation of $P$
	(resp., $S$).
\end{proof}

\begin{lem}[Strengthening]\label{lem:strength}
	$ $
	\begin{itemize}
		\item	Let $P$ be a process with $\Ga \cat v: T; \La; Z \proves P \hastype \De$
				and $v \notin \fn{P}$. Then $\Ga; \La; Z \proves P \hastype \De$.

		\item	Let $P$ be a process with $\Ga \cat k: T; \La; Z \proves P \hastype \De$
		and $\fv{k} \cap \fv{P} = \es$. Then $\Ga; \La; Z \proves P \hastype \De$.

		\item	Let $S$ be a system with $\Ga \cat v: T; \La \proves S \hastype \Int$ and
				$v \notin \fn{S}$. Then $\Ga; \La \proves S \hastype \Int$.
				
		\item	Let $S$ be a system with $\Ga \cat k: T; \La \proves S \hastype \Int$ and $\fv{k} \cap \fv{P} = \es$. Then $\Ga; \La \proves S \hastype \Int$.
	\end{itemize}
\end{lem}

\begin{proof}
	The proof for Part 1 and 2 (resp., Part 3 and 4) is a standard induction on the typing derivation of $P$
	(resp., $S$).
\end{proof}

We may now establish that typing is preserved under substitution.
\begin{lem}[Substitution]
	\label{lem:sl}

	Let $v_2 \notin \domain{\Ga}$.
	\begin{itemize}
		\item
			If $\Ga \cat v_1: T; \La; Z \proves P \hastype \De$
	    	then one of the following holds:
			\begin{itemize}
				\item	$\Ga \cat v_2: T; \La; Z \proves P \subst{v_2}{v_1} \hastype \De$
				\item	$\Ga \cat v_2: T; (\La\backslash{v_1})\cat v_2; Z \proves P \subst{v_2}{v_1} \hastype \De$
				\item	$\Ga \cat v_2: T; \La; (Z \backslash\ptuple{\ii}{\prv})\cat\ptuple{\id}{\prv} \proves P \subst{v_2}{v_1} \hastype \De$,
						if $v_1 = \pdata{\ii}{\delta} \land v_2 = \pdata{\id}{c} \land T = \privatet{\prv}{\g}$
			\end{itemize}
		\item
		If $\Ga \cat k: T; \La; Z \proves P \hastype \De$
		then one of the following holds:
		\begin{itemize}
			\item	$\Ga \cat v_2: T; \La; Z \proves P \subst{v_2}{k} \hastype \De$
%			\item	$\Ga \cat v_2: T; (\La\backslash{k}) \cat v_2; Z \proves P \subst{v_2}{k} \hastype \De$
			\item	$\Ga \cat v_2: T; \La; (Z \backslash\ptuple{x}{\prv})\cat\ptuple{\id}{\prv} \proves P \subst{v_2}{v_1} \hastype \De$,
			if $k = \pdata{x}{y} \land v_2 = \pdata{\id}{c} \land T = \privatet{\prv}{\g}$
		\end{itemize}

	\end{itemize}
\end{lem}

\begin{proof}
	The proof is an induction on the structure of the syntax of $P$.
	We give the interesting case of substituting a value inside a store.
	\begin{itemize}
		\item
			Let $P = \store{r}{\id}{c}$ and
				$\Ga \cat \pdata{\id}{c}: \privatet{\prv}{\g}; r; \id: \prv \proves P \hastype \De \cat \prv: \set{\str}$.
			By the premise of typing rule \TSt used to type the process $P$, we get:
			\begin{eqnarray*}
				\Ga \cat \pdata{\id}{c}: \privatet{\prv}{\g} \proves \pdata{\id}{c}: \privatet{\prv}{\g} \hastype \De'
			\end{eqnarray*}
			\sloppy
			By applying strengthening (Lemma~\ref{lem:strength}) to $\Ga \cat \pdata{\id}{c}: \privatet{\prv}{\g}$ and then
			weakening (Lemma~\ref{lem:weaken}) we get:
			\begin{eqnarray}
				\label{proof:subst1}
				\Ga \cat \pdata{\id}{c'}: \privatet{\prv}{\g} \proves \pdata{\id}{c'} : \privatet{\prv}{\g} \hastype \De'
			\end{eqnarray}
			Assume now that:
			\begin{eqnarray*}
				P \subst{\pdata{\id}{c'}}{\pdata{\id}{c}} = \store{r}{\id}{c'}
			\end{eqnarray*}
			If we use result~\ref{proof:subst1} we can apply the
			typing rule \TSt to the latter substitution and obtain:
			\begin{eqnarray*}
				\Ga \cat \pdata{\id}{c'}: \privatet{\prv}{\g}; u; \id: \prv \proves P \subst{\pdata{\id}{c'}}{\pdata{\id}{c}} \hastype \De \cat \prv: \set{\str}
			\end{eqnarray*}
			as required.

		\item
			Let
				$P = \store{r}{x}{y}$ and
				$\Ga \cat \pdata{x}{y}: \privatet{\prv}{\g}; r; x: \prv \proves P \hastype \De \cat \prv: \set{\str}$.
			By the premise of the typing rule \TSt used to type the process $P$, we get:
			\begin{eqnarray*}
				\Ga \cat \pdata{x}{y}: \privatet{\prv}{\g} \proves \pdata{x}{y}: \privatet{\prv}{\g} \hastype \De'
			\end{eqnarray*}
			By applying strengthening (Lemma~\ref{lem:strength}) to $\Ga \cat \pdata{x}{y}: \privatet{\prv}{\g}$ and then
			weakening (Lemma~\ref{lem:weaken}) we obtain:
			\begin{eqnarray}
				\label{proof:subst2}
				\Ga \cat \pdata{\id}{c}: \privatet{\prv}{\g} \proves \pdata{\id}{c} : \privatet{\prv}{\g} \hastype \De'
			\end{eqnarray}
			Assume now that:
			\begin{eqnarray*}
				P \subst{\pdata{\id}{c}}{\pdata{x}{y}} = \store{r}{\id}{c}
			\end{eqnarray*}
			If we use result~\ref{proof:subst2} we can apply the
			typing rule \TSt to the latter substitution and get:
			\begin{eqnarray*}
				\Ga \cat \pdata{\id}{c}: \privatet{\prv}{\g}; r; \id: \prv \proves P \subst{\pdata{\id}{c}}{\pdata{x}{y}} \hastype \De \cat \prv: \set{\str}
			\end{eqnarray*}
			as required.

		\item
			Let
				$P = \store{r}{\ii}{\con}$ and
				$\Ga \cat r: \gtype{\G}{\privatet{\prv}{\g}}; u; \ii: \prv \proves P \hastype \De \cat \prv: \set{\str}$.
			By the premise of the typing rule \TSt used to type the process $P$, we get:
			\begin{eqnarray*}
				\Ga \cat r: \gtype{\G}{\privatet{\prv}{\g}} \proves r: \gtype{\G}{\privatet{\prv}{\g}} \hastype \es
			\end{eqnarray*}
			By applying strengthening (Lemma~\ref{lem:strength}) to $\Ga \cat r: \gtype{\G}{\privatet{\prv}{\g}}$ and then
			weakening (Lemma~\ref{lem:weaken}) we obtain:
			\begin{eqnarray}
				\label{proof:subst3}
				\Ga \cat r': \gtype{\G}{\privatet{\prv}{\g}} \proves u': \gtype{\G}{\privatet{\prv}{\g}} \hastype \es
			\end{eqnarray}
			Assume now that:
			\begin{eqnarray*}
				P \subst{r'}{r} = \store{r'}{\ii}{\con}
			\end{eqnarray*}
			If we use result~\ref{proof:subst3} we can apply the
			typing rule \TSt to the latter substitution and get:
			\begin{eqnarray*}
				\Ga \cat r': \gtype{\G}{\privatet{\prv}{\g}}; r'; \ii: \prv \proves P \subst{r'}{r} \hastype \De \cat \prv: \set{\str}
			\end{eqnarray*}
			as required.\qedhere
	\end{itemize}
\end{proof}

\noindent We next define a relation \lenv
over permissions and the structures of permissions
to capture the changes in the permission environment
and the interface environment during action execution.
\begin{defi}[$\Int_1 \lenv \Int_2$]
	\label{def:lenv}
	We define the relation \lenv over permissions (\perm),
	permissions environments (\De), and interface environments (\Int) as follows:
	\begin{enumerate}
		\item
			$\pol{\perm}_1 \lenv \pol{\perm}_2$, whenever
			\begin{enumerate}[label=(\roman*)]
				\item
				\sloppy
					for all $\perm \in \pol{\perm}_1$ such that $\perm \not= \diss{\G}{\lambda}$ and $\perm \not= \use{\pol{\prp}}$  we have that $\perm \in \pol{\perm}_2$.
				\item
					for all $\perm \in \pol{\perm}_1$ such that $\perm = \diss{\G}{\lambda_1}$ we have that $\diss{\G}{\lambda_2} \in \pol{\perm}_2$ and $ (\lambda_1 \leq \lambda_2 \lor \lambda_2 = *)$.
				\item
					for all $\perm \in \pol{\perm}_1$ such that $ \perm = \use{\pol{\prp}_1}$ we have that $ \use{\pol{\prp}_2} \in \pol{\perm}_2$ and $\pol{\prp}_1 \subseteq \pol{\prp}_2$.
			\end{enumerate}

		\item
			$\De_1 \lenv \De_2$, whenever
			for all $\prv$ such that $\prv: \pol{\perm}_1 \in \De_1$ we have that $\prv: \pol{\perm}_2 \in \De_2$ and $\pol{\perm}_1 \lenv \pol{\perm}_2$.

		\item
			$\hiert{\G}{\pol{\perm}_1} \lenv \hiert{\G}{\pol{\perm}_2}$ whenever $\pol{\perm}_1 \lenv \pol{\perm}_2$.

		\item
			$\hiert{\G}{\Inthier_1} \lenv \hiert{\G}{\Inthier_2}$ whenever $\Inthier_1 \lenv \Inthier_2$.
		\item	$\Int_1 \lenv \Int_2$, whenever %whenever the following conditions hold:
				for all $\prv$ such that $\prv: \Inthier_1 \in \Int_1$ we have that $ \prv: \Inthier_2 \in \Int_2$ and $\Inthier_1 \lenv \Inthier_2$.
	\end{enumerate}
\end{defi}

The following proposition follows directly from the definition above.
\begin{prop} \label{cor:lenv}
	$ $
	\begin{itemize}
		\item	Let  $i \in \set{1, 2}$. Then $\De_i \lenv \De_1 \uplus \De_2$.
		\item	Let $\De_1 \lenv \De_2$. Then $\De_1 \uplus \De \lenv \De_2 \uplus \De$.
	\end{itemize}
\end{prop}

\begin{proof}
	The proof is immediate from the definition of $\uplus$ and \lenv.
\end{proof}

We may now state type preservation by action execution of a system.
Type preservation employs the \lenv operator.
Specifically, when a process or a system executes an action
we expect a name to maintain or lose its interface
capabilities.
%that are expressed through the typing
%of the name.
%
\begin{thm}[Type Preservation]
	\label{thm:sr}
	Consider a process $P$ and a system $S$.
	\begin{itemize}
		\item
			If $\Ga; \La; Z \proves P \hastype \De$ and $P \trans{\ell} P'$ then
			\begin{itemize}
				\item	if $\ell \not= \actinp{n}{v}$ then $\Ga; \La; Z \proves P' \hastype \De'$
						and $\De' \lenv \De$.
				\item	if  $\ell = \actinp{n}{v}$ then for some $\La'$ and $Z'$, we have $\Ga \cat v: T; \La; Z \proves P' \hastype \De'$
						and $\De' \lenv \De$.
			\end{itemize}

		\item
			If $\Ga; \La \proves S \hastype \Int$ and $S \trans{\ell} S'$
			\begin{itemize}
				\item	if $\ell \not= \actinp{n}{v}$
						then $\Ga; \La \proves S' \hastype \Int'$
						and $\Int' \lenv \Int$.

				\item	if  $\ell = \actinp{n}{v}$ then for some $\La'$, we have
						$\Ga \cat v:T; \La' \proves S' \hastype \Int'$
						and $\Int' \lenv \Int$.
			\end{itemize}
	\end{itemize}
\end{thm}

\begin{proof}
%	A similar type preservation proof can be found in~\cite{KP15}. In this proof
%	we consider the cases that are extended in this version.

	The proof is by induction on the inference tree for \trans{\ell}. We begin with the proof of the first part.

	Base Step:
	\begin{itemize}
		\item	Case: $\out{n}{v} P \trans{\actout{n}{v}} P$ and $\Gamma; \Lambda; Z \proves \out{n}{v} P \hastype \Delta$.
				By the premise of typing rule \TOut we get that $\Gamma; \Lambda; Z \proves P \hastype \Delta'$ with
				$\Gamma \proves v \hastype \Delta'$ and $\Delta = \Delta' \uplus \UDe{T}$. The result is then immediate from \propref{lenv}.

		\item	Case: $\inp{n}{k} P \trans{\actinp{n}{v}} P \subst{v}{k}$ and $\Gamma; \Lambda; Z \proves \inp{n}{k} P \hastype \Delta$.
				By the premise of typing rule \TInp we get that $\Gamma; \Lambda; Z \proves P \hastype \Delta_1$
				and $\Gamma \proves v \hastype \Delta_2$ with $\Delta = \Delta_1 \uplus \Delta_2 \uplus \RDe{T}$.
				From weakening (Lemma~\ref{lem:weaken}) and the Substitution \lemref{sl} we get that for some $\La'$  and $Z'$,
				$\Gamma \cat v:T; \Lambda'; Z' \proves P \subst{v}{k} \hastype \Delta_1$.
				The result is then immediate from \propref{lenv}.

		\item	Case: \store{r}{\id}{c} \trans{\actout{\dual{r}}{c}} \store{r}{\id}{c}.
				The case is trivial.

		\item	Case: \store{r}{\id}{c} \trans{\actinp{\dual{r}}{c'}} \store{r}{\id}{c'}
				and $\Ga; r; \id: \prv \proves \store{r}{\id}{c} \hastype \Delta$. The result is immediate by the Substitution \lemref{sl}.

	\end{itemize}

	Induction Step:
	\begin{itemize}
		\item	The interesting case for the induction step follows
				the structure of the $\tau$ interaction (LTS rule \LTau).
				$P \Par Q \trans{\tau} \newnp{\pol{m}}{P' \Par Q'}$ and
				$\Gamma; \Lambda_P \cat \Lambda_Q; Z_P \cat Z_Q \proves P \Par Q \hastype \Delta_P \uplus \Delta_Q \uplus \Delta$.
				By the premise of LTS rule \LTau and typing rule \TPar we get:
				\begin{eqnarray*}
					P &\trans{\ell_1}& P' \text{ and }
					\Gamma; \Lambda_P; Z_P \proves P \hastype \Delta_P\\
					Q &\trans{\ell_2}& Q' \text{ and }
					\Gamma; \Lambda_Q; Z_Q \proves Q \hastype \Delta_Q\\
					\pol{m} &=& \bn{\ell_1} \cup \bn{\ell_2}
				\end{eqnarray*}
				By the induction hypothesis we know that
				\begin{eqnarray*}
					\Gamma; \Lambda_P; Z_P \proves P' \hastype \Delta_P' \text{ and } \Delta_P' \lenv \Delta_P\\
					\Gamma; \Lambda_Q; Z_Q \proves Q' \hastype \Delta_Q' \text{ and } \Delta_Q' \lenv \Delta_Q
				\end{eqnarray*}
				We apply type rule \TPar together with type rule \TRes to get:
				\begin{eqnarray*}
					\Gamma \backslash \pol{m}; (\Lambda_P \cat \Lambda_Q) \backslash \pol{m}; Z_P \cat Z_Q \proves \newnp{\pol{m}}{P' \Par Q'} \hastype \Delta_P' \uplus \Delta_Q' \uplus \Delta
				\end{eqnarray*}
				The result is then immediate from \propref{lenv}.

		\item	The rest of the induction step cases follow easier argumentations.
	\end{itemize}

	We continue with the proof of the second part of the Theorem. The interesting case is the case where $S = \group{\G}{P}$.
	The rest of the cases are congruence cases that follow argumentations similar
	to those in the previous part.
	\begin{itemize}
		\item	Case:	$\group{\G}{P} \trans{\ell} \group{\G}{P'}$ with
						$\Gamma; \Lambda \proves \group{G}{P} \hastype \Int$.
						By the premise of LTS rule \SPGr and typing rule \TGr we get:
						\begin{eqnarray}
							P &\trans{\ell}& P' \nonumber \\
							\Gamma; \Lambda; Z \proves P \hastype \Delta &\text{ with }&
							\forall \prv: \pol{\perm} \in \Delta \iff \prv: \hiert{G}{\pol{\perm}} \in \Int
							\label{thm:sr1}
						\end{eqnarray}
						Part 1 of this theorem ensures that
						for some $\Gamma' = \es$ if $\ell \not= \actinp{n}{v}$ and $\Gamma' = v:T$ if
						$\ell = \actinp{n}{v}$ and some $\Lambda'$ and $Z'$,
						$\Gamma \cat \Gamma'; \Lambda'; Z' \proves P' \hastype \Delta'$ with $\Delta' \lenv \Delta$.
						If we apply typing rule \TGr we get:
						\begin{eqnarray*}
							\Gamma \cat \Gamma'; \Lambda' \proves \group{G}{P'} \hastype \Int' &\text{ with }&
							\forall \prv: \pol{\perm} \in \Delta' \iff \prv: \hiert{G}{\pol{\perm}} \in \Int'
						\end{eqnarray*}
						We can derive that $\Int' \lenv \Int$ from equation~\ref{thm:sr1} and the fact that $\Delta' \lenv \Delta$.

		\item	The rest of the cases are similar.\qedhere
	\end{itemize}
\end{proof}

\subsection{Policy Satisfaction}

We define the notion of {\em policy satisfaction}.
Policy satisfaction uses the type system as the bridge to establish a
relation between a system and a privacy policy. Intuitively a system
satisfies a policy if its interface environment satisfies a policy.

Working towards policy satisfaction, we first define a satisfaction
relation between policies \Policy and permission interfaces \Int.
\begin{defi}
	We define two satisfaction relations, denoted \sat, as:
	\begin{itemize}
		\item
			Consider a policy hierarchy \Hpol %= \G: \pol{\perm} [\Hpol_j]_{j \in J}
			and an interface hierarchy \Inthier.
			We say that \Inthier \emph{satisfies} \Hpol, written $\Hpol \sat \Inthier$, whenever:
			\[
				\tree {
						\exists k \in J \suchthat \Hpol_k = \G': \pol{\perm'}[\Hpol_i]_{i \in I}
						\\
						\G': \pol{\perm'} \uplus \pol{\perm} [\Hpol_i]_{i \in I} \sat \Inthier
				}{
					\G: \pol{\perm} [\Hpol_j]_{j \in J} \sat \hiert{\G}{\Inthier}
				}
				\qquad \qquad
				\tree {
					\pol{\perm}_2 \lenv \pol{\perm}_1
				}{
					\G: \pol{\perm_1} [] \sat \hiert{\G}{\pol{\perm_2}}
				}
			\]

		\item
			Consider a policy \Policy and an interface $\Int$.
			$\Int$ \emph{satisfies} \Policy, written $\Policy \sat \Int$, whenever:
			\[
				\tree {
					\Hpol \sat \Inthier
				}{
					\prv \gg \Hpol; \Policy \sat \prv : \Inthier
				}
				\qquad \qquad \qquad
				\tree {
					\Hpol \Vdash \Inthier \qquad \Policy \sat \Int
				}{
					\prv \gg \Hpol; \Policy \Vdash \prv : \Inthier; \Int
				}
			\]
	\end{itemize}
\end{defi}\medskip

\noindent According to the definition of $\Hpol \sat \Inthier$,
an interface hierarchy $\Inthier$ satisfies a policy hierarchy $\Hpol$
whenever:
i) the \Inthier group hierarchy is included in group hierarchy $\Hpol$;
and ii) the permissions of the interface hierarchy are included in
the union of the permissions of the corresponding group hierarchy in \Hpol.
%
%if its groups can be decomposed into a partition $\set{\G} \cup \bigcup_{j\in J}G_j$,
%such that there
%exist interface hierarchies $\Ihier{H}_j$ referring to groups $G_j$, each
%satisfying hierarchy $H_j$ and where the union of the assigned permissions $\Ihier{H}_j$ with
%permissions $\seq{p}$ is a superset of the permissions of $\Ihier{H}$, that is,
%$\permis{\Ihier{H}} \lenv (\uplus_{j \in J} \permis{\Ihier{H}_j}) \uplus \tilde{p}$.
Similarly, a $\Int$-interface satisfies a policy, $\Policy \sat \Int$,
if for each component $\prv: \Inthier$ of $\Int$, there exists a component $\prv \gg \Hpol$ of
$\Policy$ such that $\Inthier$ satisfies $\Hpol$. At this point, note that we assume
$;$ to be a commutative operator.
A direct corollary of the definition is
the preservation of the $\lenv$ operator over the
satisfiability relation:
%An important property of the $\lenv$ operator is that it
%preserves policy satisfiability. %Intuitively, this result
%will be used to show that action execution maintains
%policy satisfiability by a process.

\begin{prop}
	\label{cor:respect_policy}
	If $\Policy \sat \Int_1$ and $\Int_2 \lenv \Int_1$ then
	$\Policy \sat \Int_2$.
\end{prop}

Policy satisfaction is a main definition for this paper as it specifies
the situation where a system respects a privacy policy.
The formalisation of the intuition above follows in the next definition.
\begin{defi}[Policy Satisfaction]
	\label{def:policy_sat}
	Consider $\Policy: \diamond$, a type environment \Ga and system $S$.
	We say that $S$ satisfies $\Policy$, written
	$\Policy; \Ga \proves S$, whenever
	there exist $ \La$ and $ \Int $ such that $\Ga; \La \proves S \hastype \Int$ and $\Policy \sat \Int$.
\end{defi}

{The main idea is to check whether a system satisfies a given  privacy
policy under an environment that maps channels, constants, and private data
to channel types, constant types, and private data types, respectively.}

\begin{cor}
	\label{cor:sr}
	If $\Policy; \Ga \proves S$
	and $S \trans{\ell} S'$ then $\Policy; \Ga \proves S'$
\end{cor}

\begin{proof}
	The proof is direct from \corref{respect_policy}
	and \thmref{sr}.
\end{proof}

\subsection{System Safety}
We consider a system to be safe with respect to a privacy policy
when the system cannot reduce to an error system, i.e., a state where the
privacy policy is violated. Towards
this direction we need to define a class of error systems
that is parametrised on privacy policies.

We begin with an auxiliary definition of
$\countLink{P, \Ga, \gtype{\G}{\privatet{\prv}{\g}}}$ that
counts the number of output prefixes of the form
$\out{u}{t} Q$, where $t: {\privatet{\prv}{\g}} \in \Ga$, within a process $P$.
The induction definition of function \countLink{P, \Ga, \gtype{\G}{\privatet{\prv}{\g}}} is as follows.
\begin{defi}[Count References]
	We define function \countLink{P, \Ga, \gtype{\G}{\privatet{\prv}{\g}}} as:
	\begin{itemize}
		\item	$\countLink{\inact, \Ga, \gtype{\G}{\privatet{\prv}{\g}}} = 0$
		\item	$\countLink{\out{u}{t} P, \Ga, \gtype{\G}{\privatet{\prv}{\g}}} = 1 + \countLink{P, \Ga, \gtype{\G}{\privatet{\prv}{\g}}}$ if $t: \privatet{\prv}{\g} \in \Ga$
		\item	$\countLink{\out{u}{t} P, \Ga, \gtype{\G}{\privatet{\prv}{\g}}} = \countLink{P, \Ga, \gtype{\G}{\privatet{\prv}{\g}}}$ if $t: \privatet{\prv}{\g} \not\in \Ga$
		\item	$\countLink{\store{u}{\ii}{\con}, \Ga, \gtype{\G}{\privatet{\prv}{\g}}} = 0$

		\item	$\countLink{\inp{u}{k} P, \Ga, \gtype{\G}{\privatet{\prv}{\g}}} = \countLink{P, \Ga, \gtype{\G}{\privatet{\prv}{\g}}}$
		\item	$\countLink{\newn{n} P, \Ga, \gtype{\G}{\privatet{\prv}{\g}}} = \countLink{P, \Ga, \gtype{\G}{\privatet{\prv}{\g}}}$
		\item	$\countLink{\newn{n} P, \Ga, \gtype{\G}{\privatet{\prv}{\g}}} = \countLink{P, \Ga, \gtype{\G}{\privatet{\prv}{\g}}}$
		\item	$\countLink{P \Par Q, \Ga, \gtype{\G}{\privatet{\prv}{\g}}} = \countLink{P, \Ga, \gtype{\G}{\privatet{\prv}{\g}}} + \countLink{Q, \Ga, \gtype{\G}{\privatet{\prv}{\g}}}$
		\item	$\countLink{\ifelse{u}{u'}{P}{Q}, \Ga, \gtype{\G}{\privatet{\prv}{\g}}}$\\
  $= \countLink{P, \Ga, \gtype{\G}{\privatet{\prv}{\g}}} + \countLink{Q, \Ga, \gtype{\G}{\privatet{\prv}{\g}}}$
	\end{itemize}
\end{defi}
%($\countLink{P, \Ga, \gtype{\G}{\privatet{\prv}{\g}}}$ can be easily defined inductively on the syntax of $P$.)

%\dk{define \countLink{\out{u}{t} P, \Ga, \gtype{\G}{\privatet{\prv}{\g}}} etc}

%Moreover, given a policy hierarchy $H$ and a set of groups $\seq{G}$, let us write
%$H_{\seq{G}}$ for the interface hierarchy
%such that (i) $\groups{H_{\seq{G}} }= \seq{G}$,  (ii) $H \Vdash H_{\seq{G}}$ and,
%	(iii) for all $\Ihier{H}$ such that $\groups{\Ihier{H}} = \seq{G}$ and $H \Vdash \Ihier{H}$,
%then $\permis{\Ihier{H}} \lenv \permis{H_{\seq{G}}}$. Intuitively, $H_{\seq{G}}$ captures
%an  interface hierarchy with the maximum possible permissions for groups $\seq{G}$ as determined
%by $H$. We may now define
%%private data types $t$.
%%have
%%an object with base $t$.
%%
%%
\noindent  We may now define the notion of an \emph{error system} with respect
to a privacy policy. Intuitively an error system is a system that can do an action
not permitted by the privacy policy.
The error system clarifies the satisfiability relation
between policies and processes.
\begin{defi}[Error System]
	\label{def:error}
	Assume $\pol{\G} = \G_1 \cat \dots \cat \G_n$, and
	consider a policy $\Policy$, an environment $\Ga$, and a system
	\begin{eqnarray*}
		\System \equiv
		\group{\G_1}{\newnp{\tilde{x}_1}
					{\group{\G_2}{\newnp{\tilde{x}_2}
						{\dots  (\group{\G_n}{\newnp{\tilde{x}_n}{P \Par Q}}  \Spar S_n  }\dots) } }
			} \Spar S_1
	\end{eqnarray*}
	System \System is an \emph{error system} with respect to
	$\Policy$ and $\Ga$ if there exists $\prv$ such that $\Policy = \prv \gg \Hpol; \Policy'$
	and at least one of the following holds:
	\begin{enumerate}[widest=10]
		\item
			$\rd \notin \permis{\Hpol_{\pol{\G}}}$ %\permis{\Ihier{H}}$
			and $\exists u$
			such that $\Ga \proves u: \gtype{\G}{\privatet{\prv}{\g}} \hastype \De$ and
			$P \equiv \inp{u}{k} P'$.

		\item
			$\up \notin \permis{\Hpol_{\pol{\G}}}$
			and
			$\exists u$
			such that $\Ga \proves u: \gtype{\G}{\privatet{\prv}{\g}} \hastype \De $ and
			$P \equiv \out{u}{v} P'$.

		\item
			$\rf \notin \permis{\Hpol_{\pol{\G}}}$ and
			$\exists k$
			such that $\Ga \proves k: \gtype{\G}{\privatet{\prv}{\g}} \hastype \es$ and
			$P \equiv \inp{u}{k} P'$.

		\item
			$\diss{\G'}{\lambda} \notin \permis{\Hpol_{\pol{\G}}}$ and
			$\exists u$
			such that $\Ga \proves u: \gtype{\G}{\privatet{\prv}{\g}} \hastype \es$ and
			$P \equiv \out{u'}{u} P'$.

		\item
			$\idp \notin \permis{\Hpol_{\pol{\G}}}$ and $\exists u$
			such that $\Ga \proves u: \gtype{\G}{\privatet{\prv}{\g}} \hastype \De$ and
			$P \equiv \inp{u}{\pdata{x}{y}} P'$.

		\item
			$\str \notin \permis{\Hpol_{\pol{\G}}}$ and $\exists r$
			such that $\Ga \proves r: \gtype{\G}{\privatet{\prv}{\g}} \hastype \De$ and
			$P \equiv \store{r}{\id}{\con}$.

		\item
			$\aggr \notin \permis{\Hpol_{\pol{\G}}}$ and $\exists u$
			such that $\Ga \proves r: \gtype{\G}{\privatet{\prv}{\g}} \hastype \De$ and
			$P \equiv \store{r}{\id}{\con} \Par \store{r'}{\id}{\con'}$.
			%with $n > 1$

		\item
			$\use{\prp} \notin \permis{\Hpol_{\pol{\G}}}$ and
			$\exists \con, \con'$
			such that $\Ga \proves \con: \purposet{\prp}{\g}\hastype \es$,
			$\Ga \proves \pdata{\ii}{\con'}: \purposet{\prv}{\g}\hastype \De$ and
			$P \equiv \ifelse{\con'}{\con}{P_1}{P_2}$.

		\item
			$\idf{\prv'} \notin \permis{\Hpol_{\pol{\G}}}$ and
			$\exists \con, \con'$
			such that $\Ga \proves \pdata{\ii}{\con}: \privatet{\prv'}{\g}\hastype \es$,
			$\Ga \proves \pdata{\hid}{\con'}: \purposet{\prv}{\g}\hastype \De$ and
			$P \equiv \ifelse{\con'}{\con}{P_1}{P_2}$.

		\item
			$\diss{\G}{\lambda} \in \permis{\Hpol_{\pol{\G}}}$,
			$\lambda \not= \infinite$ and
			$\countLink{P, \Ga, \gtype{\G}{\privatet{\prv}{\g}}} > \lambda$.

		\item
			there exists a sub-hierarchy of $\Hpol$, $\Hpol' = \role{\G_k}{\pol{\perm}}{\Hpol_i}_{i \in I}$ with $1 \leq k \leq n$
			as well as
			$\ndiss{\kind} \in \pol{\perm}$
			and $\exists u$
			such that $\Ga \proves u: \gtype{\G'}{\gtype{\G}{\privatet{\prv}{\g}}} \hastype \es$, and % $\Ga \proves y \hastype G'[G[t]]$ and
			$P \equiv \out{u}{u'} P'$ with $\G' \notin \groups{\Hpol'}$.

			%$\langle G_1 \dots \G_n, p \cup \set{\nds} \rangle, \langle G_k: \nds \rangle
%			\dk{${\cal P}(t) \Vdash \langle G_1 \dots G_n \rangle \hastype p \cup \set{\nds\ G'}$
			%$\langle G_1 \dots \G_k, p \cup \set{\nds} \rangle \in \paths({\cal P}, t)$,
			%$1 \leq k \leq n$,
%			and $\exists, x, y$
%			such that $\Ga \types x \hastype G[t]$, $\Ga \types y\hastype G''[G[t]]$ and
%			with group type $G$ such that
%			$G \notin \set{G_k \dots G_n}$ and
%			$P \equiv \outp{y}{x} P'$ with $G'' \notin \descendants{G'}{{\cal P}(t)}$
%			}
	\end{enumerate}
\end{defi}\medskip

\noindent According to the definition,  
%for a given type environment \Ga and a policy \Policy
%we define a class of error systems in the definition above.
the first five error systems require that the send or the
receive prefixes inside a system do not respect \Policy.
The first system is an error because it allows a input of
private data inside a group hierarchy, where the corresponding
group hierarchy in \Policy does not include the \rd permission.
Similarly, the second error system outputs private data and at the
same time there is no \up permission in the corresponding
group hierarchy of \Policy.
Systems of clauses (3) and (4) deal with input and output of reference
name without the \rf and \diss{\G}{\lambda} permissions, respectively,
in the corresponding group hierarchy in \Policy.
The fifth error system  requires that private data can be input along
with access to its identity despite the lack of the \idp permission in \Policy.
The next error system defines a store process inside a group
hierarchy with the store permission not in the corresponding
policy hierarchy. A system is an error with respect to aggregation
if it aggregates stores of the same identity without the \aggr
permission in the corresponding policy hierarchy (clause (7)).
Private data usage is expressed via a matching construct:
A usage error system occurs when a process inside a group hierarchy
tries to match private data without the \use{\prp} permission
in the policy (clause (8)). Similarly, an identification error process
occurs when there is a matching on private data and there is
no \idf{\prv} permission defined by the policy (clause (9)).
A system is an error with respect to a policy \Policy if
the number of output prefixes to references in its definition
(counted with the $\countLink{P, \Ga, t}$ function)
are more than the $\lambda$ in the $\diss{G}{\lambda}$ permission
of the policy's hierarchy (clause (10)).
Finally, if a policy specifies that no data should be
disseminated outside  some group $\G$, then a process should
not be able to send private data links to groups
that are not contained within the hierarchy of $\G$ (clause (11)).

\begin{comment}

The first two error processes expect that a process with no $\rd$ or $\wrt$
permissions on a certain level of the hierarchy should not have, respectively,
a prefix receiving or sending an object typed with the private data.
Similarly an error process with no $\acc$ permission on a certain level
of the hierarchy should not have an input-prefixed subject with
object a link to private data. An output-prefixed process that
send links through a channel of sort $G'$ is an error process
if it is found in a specific group hierarchy with no $\ds{G'}{\lambda}$ permission.
In the fifth clause, a process is an error if
the number of output prefixes to links in its definition
(counted with the $\countLink{P, \Ga, t}$ function)
are more than the $\lambda$ in the $\ds{G}{\lambda}$ permission
of the process's hierarchy.
Finally, if a policy specifies that no data should be
disclosed outside  some group $G$, then a process should
not be able to send private data links to groups
that are not contained within the hierarchy of $G$.
\end{comment}

As expected, if a system is an error with respect to a
policy \Policy and an environment $\Ga$
then its $\Int$-interface does not satisfy policy \Policy:

\begin{lem}
	\label{lem:non_satisf}
	\label{lem:error_non_satisfiable}
	Let system $S$ be an error system with respect to a well-formed policy \Policy and a type environment $\Ga$.
	If $\Ga \proves S \hastype \Int$ then $\Policy \not\sat \Int$.
\end{lem}

\begin{proof}[Sketch]
	The proof is based on the definition of error systems.
	Each error system $S$ is typed as:
	\[
		\Ga; \La \proves S \hastype \Int
	\]
	We then show that $\Policy \not\sat \Int$.

	The proofs for all error systems are similar. We give a
	typical case. Consider:
	\begin{itemize}
		\item System
\[		S=
		\group{\G_1}{\newnp{\tilde{x}_1}
					{\group{\G_2}{\newnp{\tilde{x}_2}
						{\dots  (\group{\G_n}{\newnp{\tilde{x}_n}{P \Par Q}}  \Spar S_n  }\dots) } }
			} \Spar S_1\]
				with $P \scong \store{r}{\id}{c}$;

		\item	$\pol{\G} = \G_1 \cat \dots \G_n$;

		\item	Policy $\Policy = \prv \haspolicy \Hpol; \Policy'$
				such that $\str \notin \permis{\Hpol_{\pol{G}}}$;

		\item	Type environment $\Ga$ such that $\Ga \proves r: \gtype{\G}{\privatet{\prv}{\g}} \hastype \es$.
	\end{itemize}

	We type process $P$ using typing rule \TSt:
	\[
		\tree[\TSt] {
			\Ga \proves r: \gtype{\G}{\privatet{\prv}{\g}} \hastype \es
			\quad
			\Ga \proves \pdata{\id}{c}: \privatet{\prv}{\g} \hastype \es
		}{
			\Ga; r; \ptuple{\id}{\prv} \proves \store{r}{\id}{c} \hastype \prv: \set{\str}
		}
	\]
	We apply rule \TPar to get:
	\[
		\tree[\TPar] {
			\Ga; r; \ptuple{\id}{\prv} \proves \store{r}{\id}{c} \hastype \prv: \set{\str}
			\\
			\Ga; \La; Z \proves Q \hastype \De
		}{
			\Ga; \Lambda \cat r; Z \cat \ptuple{\id}{\prv} \proves \store{r}{\id}{c} \Par Q \hastype \Delta \uplus \prv: \set{\str}
		}
	\]
	We then apply type rules \TSRes, \TGr,  \TSPar,   \TSRes, and \TSGr,   to
	get:
	\[
		\tree[\TSGr] {
			\tree[]{
				\dots
			}{
				\Ga; \La \proves \newnp{\tilde{x}_1}{\group{\G_2}{\newnp{\tilde{x}_2} {\dots  (\group{\G_n}{\newnp{\tilde{x}_n}{P \Par Q} \Spar S_n  }) \dots } } \Spar S_1} \hastype \Int'
			}
		}{
			\Ga; \La \proves S \proves \Int
		}
	\]
	From the application of rules \TGr, \TSGr we know that
	$\prv: \hiert{\G_1}{ \ldots \hiert{\G_n}{\pol{\perm}}\ldots} \in \Int$ for some $\pol{\perm}$.
	From the definition of $\uplus$ and the application of rule
	\TPar we know that $\str \in \pol{\perm}$.

	From this we can deduce:
	$\pol{\perm} \not \lenv \permis{\Hpol_{\pol{G}}} \implies
	\Hpol \not\sat \hiert{\G_1}{ \ldots \hiert{\G_n}{\pol{\perm}}\ldots} \implies
	\Policy \not\sat \Int$,
	as required.

	The remaining cases for the proof are similar.
%	\dk{complete the proof}
\end{proof}

We may now conclude with a safety theorem which verifies that
the satisfiability of a policy by a typed system is preserved by the semantics.
Below we denote a (possibly empty) sequence of actions $\trans{l_1} \dots \trans{l_n}$
with $\Trans{\pol{\ell}}$

\begin{thm}[Safety]
	\label{thm:safety}
	If $\Policy; \Ga \proves S \hastype \Int$ %, $\Policy \sat \Int$
	and $S \Trans{\pol{\ell}} S'$ then $S'$ is not an error with respect to policy \Policy.
\end{thm}

\begin{proof}
	The proof is immediate by \corref{sr} and \lemref{non_satisf}.
\end{proof}

\section{Use Cases}
\subsection{Electronic Traffic Pricing}

\label{sec:etp}

Electronic Traffic Pricing (\ETP) is an electronic toll collection
scheme in which the fee to be paid by drivers depends on the
road usage of their vehicles where factors such as the type of roads
used and the times of the usage determine the toll. To
achieve this, for each vehicle detailed time and location
information must be collected and processed and the due
amount can be calculated with the help of a digital tariff and
a road map.  A number of  possible implementation methods may be considered
for this scheme~\cite{DJongeJ08}. In the centralised approach all location information
is communicated to the pricing authority which computes
the fee to be paid based on the received information. In the
decentralised approach the fee is computed locally on the car via
the use of a third trusted entity such as a smart card.
In the following subsections we consider these approaches and their
associated privacy characteristics.

\subsection{The centralised approach}
This approach makes use of on-board equipment (\OBE) which computes
regularly the geographical position of the car and forwards it
to the Pricing Authority (\PA).
To avoid drivers tampering with
their \OBE and communicating  false information, the authorities
may perform checks on the spot to confirm that the \OBE is reliable.

We may model
this system with the aid of five groups: $\ETP$ corresponds to
the entirety of the \ETP system,  $\Car$ refers to the car
and is divided into the $\OBE$ and the $\GPS$ subgroups,
and $\PA$ refers to the pricing authority. Note that
in our model we simplify the pricing scheme by omitting timing considerations.
%from the information communicated from the car to the authority.
This, however, can be easily extended by the introduction of additional
processes for storing timing information as well as the associated
actions for storing/communicating the timings of the observations,
similarly to the way locations are handled. Furthermore, according
to this scheme, the car is responsible for frequently (e.g. once a day) sending its collected
data log to the authority. We abstract away this behavior of the car,
by modeling the car sending to the authority $m$ random
locations. Finally, we point out that in our model we simplify the study to a single car. Note that
in order to reason about multiple cars we would need to define a distinct
group $\Car_v$ for each vehicle $v$.

  As far as types are concerned,
we assume the existence of two ground types:
$\lambda$ referring to locations and $\phi$ referring to
fees. Furthermore, we assume the private data types $\Loc$
and $\Fee$ and the constant type $\spotcheck$.
We write $T_l = \gtype{\Loc}{\lambda}$, $T_f = \gtype{\Fee}{\phi}$,
$T_c^l = \gtype{\Car}{T_l}$, $T_{etp}^l = \gtype{\ETP}{T_l}$, $T_{pa}^{l} =\gtype{\PA}{T_l}$,
$T_{pa}^{f} =\gtype{\PA}{T_f}$ and $T_{sc}=\gtype{\ETP}{\gtype{\ETP}{T_l}}$.

We define the system as follows:
\begin{eqnarray*}
	{\mathit System}	& = &
					\group{\ETP}{\;
						\newn{\mathit{spotcheck}}\newnp{\mathit{topa}}{\group{\Car}{C} \;\Spar\; \group{\PA}{A})}\;
					}
	\\
	C		& = &	{
						\newnp{r}{\store{r}{\id}{l}
						\;\;\Spar\;\;
						\group{\GPS}{L}
						\;\;\Spar \;\;
						\group{\OBE}{O}}
					}
	\\
	L		& = &	\repl\; \inp{lc}{newl}\out{r}{\pdata{\id}{newl}} \inact
	\\
	O		&=&		\repl\; \out{\mathit{topa}}{r} \inact \;			
			\Par\;	\repl\; \inp{\mathit{spotcheck}}{z} \out{z}{r} \inact %\inp{r}{\pdata{x}{y}} \out{z}{\pdata{x}{y}} \inact
	\\
    A		& = &	\newn{r_1} \ldots \newn{r_m} (\store{r_1}{\id}{x_1}\Par\ldots\Par\store{r_m}\id{x_m}\Par\store{f}\id{y}\Par R \Par S)
	\\
	R		& = &	\inp{\mathit{topa}}{z_1} \inp{z_1}{\pdata{x}{y}} \out{r_1}{\pdata{x}{y}} \ldots\\
& & \inp{\mathit{topa}}{z_m} \inp{z_m}{\pdata{x}{y}} \out{r_m}{\pdata{x}{y}} \out{f}{\pdata{\id}{\mathit{fee}}}\inact
	\\
	S		& = &	\newn{s}\out{spotcheck}{s} \inp{s}{z} \inp{z}{\pdata{x}{y}}
                    \ifelse{y}{\mathit{l_{sc}}}{\inact}{\out{f}{\pdata{\id}{\mathit{fine}}}}\inact
\end{eqnarray*}
In the model above we have system ${\mathit System}$ where a car process $C$ and
an authority process $A$ are nested within the $\ETP$ group. The two processes
share names $\mathit{topa}$ and $\mathit{spotcheck}$ via which
the car and the authority can communicate for delivering data  from
the car to the authority and for initiating spot checks, respectively.
In turn, a car $C$ is composed
of a store process, where the identifier of the car is associated
with the current location of the car, the component $O$ belonging to group $\OBE$,
and the component $L$ responsible for
computing the current location belonging to group $\GPS$. These three processes
are nested within the
$\Car$ group and behave as follows: The store is initiated with the identifier
of the car and its initial location ($l$) and may be accessed by processes $L$ and $O$.
Process $L$ computes the new locations of the car and saves them
to the store. The OBE $O$ may
read the current location of the car from the store
and forward it to the authority via channel $\mathit{topa}$ or it may
enquire this location
as a response of a spot check initiated by the pricing authority $A$.
In turn, authority $A$ is defined as follows: it contains $m$ stores
for saving $m$ consecutive locations of the car, a store for
saving the fee of the car, and two processes
$R$ and $S$, which are responsible for computing the fee of the car
and performing spot checks, respectively. Specifically, process $R$
receives $m$ consecutive locations of the car and saves them in
its stores before calculating the fee based on the received values.
This value, $\mathit{fee}$, is essentially obtained as a function
of the values of the locations stored in the $m$ stores of the process.
In turn, process $S$ performs a spot check on the car.
During a spot check, component $S$ creates a new channel
$s$ and sends it to the OBE
via which the OBE is expected to return the current
location for a verification check.  This
value is compared to the actual location of the spot check of the car, $\mathit{l_{sc}}$,
and, if the location communicated by the car is erroneous, then
the authority imposes a fine, $\mathit{fine}$.

By applying the rules of our type system we may show that $\Gamma;\es \proves {\mathit System} \hastype \Theta$,
where:
\begin{eqnarray*}
	\Gamma	&=&	\set{
				%\mathit{spotcheck} :T_{sc},
				%\mathit{topa}:T_{etp}^l,
				%\mathit{r}:T_c^l,
				lc: \gtype{\ETP}{\lambda},
				%s: T_{etp}^l,
				%r_i: T_{pa}^l,
				f:T_{pa}^f,
				%\\
			%&&	
				l: \lambda,
				l_{sc}:\gtype{\spotcheck}{\lambda},
				\mathit{fee}:\phi,
				\mathit{fine}:\phi}
				%\mathit{newl}:\lambda}
	\\
	\Theta	& = &	\Fee	:	\hiert{\ETP}{\hiert{\PA}{\set{\str, \up}}}, %];\;\; \Loc\gg \ETP[\PA[\{acc,\rd\}]];
	\\
			&&		\Loc	:	\hiert{\ETP}{\hiert{\Car}{\set{\str}}},
	\\
			&&		\Loc	:	\hiert{\ETP}{\hiert{\Car}{\hiert{\OBE}{\set{\diss{\ETP}{\infinite}}}}},\\
			&&		\Loc	:	\hiert{\ETP}{\hiert{\Car}{\hiert{\GPS}{\set{\up}}}},\\
			&&		\Loc	:	\hiert{\ETP}{\hiert{\PA}{\set{\rf,\str,\rd,\idp,\use{\spotcheck},\aggr}}}
\end{eqnarray*}
A possible privacy policy for the system might be one that states
that locations may be freely forwarded by the \OBE. We may
define this by $\Policy = \Loc \gg H$ where:
\begin{eqnarray*}
	H	&=&		\role{\ETP}{\ndiss{\sens}} {\\
		&&		\quad	\role{\Car}{\str}{\lrole{\OBE}{\rd,\idp, \diss{\ETP}{\infinite}}{}},\\
		&&		\quad	\lrole{\GPS}{\rd,\idp,\up}{},\\
		&&		\quad	\lrole{\PA}{\rf,\str,\rd,\idp,\use{\spotcheck},\aggr}\\
		&&		}
\end{eqnarray*}
%
%{\begin{tabbing}
%xxxxx\=\kill
%\> $ H = \ETP:\{\ndiss{\sens}\}\;$\\
%\>\hspace{0.9in}\=$[\Car:\{\str\}[$\= $\OBE :\{\rd,\idp, \diss{\ETP}{*}\},$\\
%\>\>\>$        \GPS :\{\rd,\idp,\up\}],$\\
%%\>\>\>$]\;$,\\
%\>\>$ \;\PA:\{\rf,\str,\rd,\idp,\use{\spotcheck},\aggr\}]$
%%\>\>$]$
%\end{tabbing}}
%%
%
%We have that
% ${P} \Vdash \Loc\gg \ETP[\PA[\{acc,\rd\}]]$, since the permissions
% assigned to groups $\ETP$ and $\PA$ by the policy are equal to
% $\{acc,\rd\}\preceq \{acc,\rd\} = \permis{\ETP[\PA[\{acc,\rd\}]]}$.
% Similarly, \sloppy{ $\Pol \Vdash \Loc\gg \ETP[\Car[\OBE[\{acc,ds{\ETP}{*}\}]]]$}
% and $\Pol \Vdash \Loc\gg \ETP[\Car[\GPS[\{ds{\Car}{*}\}]]]$.
\noindent
By comparing environment $\Theta$ with policy
$\Policy$, we may conclude that ${\mathit System}$ satisfies $\Pol$.

This architecture is simple but also very weak in protecting the
privacy of individuals: the fact that the \PA gets detailed travel information about
a vehicle constitutes a privacy and security threat.
%central
%database where this information is stored will be an attractive target for
%individuals or organisations with unfriendly intentions, like terrorists or blackmailers.
%
In our system this privacy threat can be pinpointed to private data type $\Loc$ and the fact
that references to locations may be communicated to the PA for an unlimited
number of times. % via the bound names $\mathit{loc}$ and $\mathit{l_s}$.
An alternative implementation that limits
the transmission of locations is presented in the second implementation proposal presented
below.

\subsection{The decentralised approach}
To avoid the disclosure of the complete travel logs of a car, the decentralised solution
employs a third trusted entity (e.g. smart card) to make computations of the fee locally
on the car and send its value to the authority which in turn may make spot checks
to obtain evidence on the correctness of the calculation.

The policy here would require that locations can be communicated for at most
a small fixed amount of times and that the OBE may read the fee computed
by the smart card but not change its value. The new privacy
policy might be
$\Policy = \Loc \gg H, \Fee\gg F$ where:
\begin{longtable}{rcl}
	$H$	&=&		\role{\ETP}{\ndiss{\sens}} {\\
		&&		\quad	\role{\Car}{\str}{\lrole{\OBE}{\diss{\ETP}{2}}{}},\\
		&&		\quad	\lrole{\GPS}{\up},\\
		&&		\quad	\lrole{\SC}{\rd, \idp,\use{\computefee}},\\
		&&		\quad	\lrole{\PA}{\rf,\rd,\idp,\use{\spotcheck}}\\
		&&		}
		\\
	$F$	&=&		\role{\ETP}{\ndiss{\sens}} {\\
		&&		\quad	\role{\Car}{}{\lrole{\OBE}{\diss{\ETP}{1}}{}},\\
		&&		\quad	\lrole{\GPS}{},\\
		&&		\quad	\lrole{\SC}{\str,\up, \diss{\Car}{\infinite}},\\
		&&		\quad	\lrole{\PA}{\rf,\rd,\idp}\\
		&&		}
\end{longtable}
%
%{\begin{tabbing}
%xxxxx\=\kill
%\> $ H = \ETP:\{\ndiss{\sens}\}\;$\\
%\>\hspace{0.9in}\=$[\Car:\{\str\}[$\= $\OBE :\{\diss{\ETP}{2}\},$\\
%\>\>\>$        \GPS :\{\up\}],$\\
%\>\>\>$        \SC: \{\rd,\idp,\use{\computefee}\}$\\
%%\>\>\>$]\;$,\\
%\>\>$ \;\PA:\{\rf,\rd,\idp,\use{\spotcheck}\}]$
%%\>\>$]$
%\end{tabbing}}
%{\begin{tabbing}
%xxxxx\=\kill
%\> $ F = \ETP:\{\ndiss{\sens}\}\;$\\
%\>\hspace{0.9in}\=$[\Car[$\= $\OBE :\{\diss{\ETP}{1}\},$\\
%\>\>\>$        \GPS :\{\}],$\\
%\>\>\>$        \SC: \{\str,\up, \diss{\Car}{*}\}$\\
%%\>\>\>$]\;$,\\
%\>\>$ \;\PA:\{\rf,\rd,\idp\}]$\\
%%\>\>$]$
%\end{tabbing}}
%
To model the new system as described above
we assume a new group $\SC$ and a new component $S$ which defines the code of
a smart card, belonging to group \SC:
\begin{eqnarray*}
	{\mathit System'}	& = &	\group{\ETP}
					{\;
						\newn{\mathit{spotcheck}}\newn{\mathit{topa}}
						(\;
							\group{\Car} C
							\;\Spar \;
							\group{\PA}A
						\;)
					\;}
	\\
	C		& = &	{\;\newn{r}
					\;(\;
						\store{r}{\id}{l}
						\;\;\Spar\;\;
						\group{\GPS}{L}
						\;\;\Spar \;\;
						\group{\OBE}{O}
						\;\;\Spar\;\;
						\group{\SC}{S}
					\;)
					\;}
	\\
	L		& = &	\repl\; \inp{lc}{newl}\out{r}{\pdata{\id}{newl}} \inact
	\\
	O		& = &	\inp{\mathit{spotcheck}}{z} \out{z}{r} \inact %\inp{\mathit{spotcheck}}{s}\out{s}{r}\inact\\
			\\
			&\Par&	\inp{\mathit{send}}{f} \out{\mathit{sendpa}}{f} \inact
	\\
	S		& = &	\newnp{f}{
						\store{f}\id{v} \Par (\inp{r}{\pdata{x}{y}})^m %\ifelse{y}{l}{\out{f}{\pdata{\id}{\mathit{fee}}}}
						%{\out{f}{\pdata{\id}{\mathit{fee'}}}}
						\out{f}{\pdata{\id}{\mathit{fee}}}\out{\mathit{send}}{\mathit{f}} \inact
					}
	\\
	A		& = &	\newnp{s}{
						 \out{spotcheck}{s}{\inp{s}{r}{\inp{r}{\pdata{x}{y}}\ifelse{y}{l_{sc}}{\inact}{\out{f}{\pdata{\id}{\mathit{fine}}}}}}\inact
					}
					\\
			&\Par&	\inp{\mathit{sendpa}}{\mathit{x}}\inp{x}{\pdata{x}{\mathit{y}}}\inact
\end{eqnarray*}

\noindent In this system, we point out that process $S$ has a local store for saving the
fee as this is computed after reading a number of location values from the
car store. We write $(\inp{r}{\pdata{x}{y}})^m P$ as a shorthand
for prefixing the $P$ process by $m$ occurrences of the input label ${r?}(\pdata{x}{y})$.
As in the centralised approach, this is implemented as an abstraction of the actual behaviour
of the system where, in fact, the complete log is communicated to the smart card for a
specific time period. This could be captured more precisely in a timed extension of our
framework. Moving on to the behaviour of $S$, once computing the fee, process $S$ is responsible for announcing
this fee to the \OBE process which is then responsible for disseminating the
fee to the Pricing Authority. When the \PA receives a reference to the
data it reads the fee assigned to the car.
As a final point, we observe that the \OBE process engages in exactly
one spot check.

We may verify that $\Gamma'; \es \proves {\mathit System'} \hastype \Theta'$,
where $\Gamma' = \Gamma\cup\{\mathit{send}:\gtype{\Car}{T_{pa}^f},
\mathit{sendpa}:\gtype{\ETP}{T_{pa}^f}\}$
and interface $\Theta'$ satisfies the enunciated policy.

\subsection{Privacy-preserving speed-limit enforcement}

\label{sec:speed_control}

Speed-limit enforcement is the action taken by
appropriately empowered authorities to check that road vehicles
are complying with the speed limit in force on roads and highways.
A variety of methods are employed to this effect including automated
roadside \emph{speed-camera} systems based on Doppler-radar based measurements,
radio-based transponder techniques, which require additional hardware
on vehicles, or section control systems~\cite{SR14speedlimit}.

In most of these techniques, the process of detecting whether
a vehicle has exceeded the speed limit is done centrally: data recorded
is stored at a central server and
processing of the data is implemented in order to detect
violations. However, this practice clearly
endangers the privacy of individuals and in fact goes against
the law since the processing of personal data is prohibited unless
there is evidence of a speed limit violation.

In our study below, we model an abstraction of the speed-camera scheme and an
associated privacy policy and we show by typing that the model
of the system satisfies the policy.
Beginning with the requirements for privacy-preservation,
the system should ensure the following:
\begin{enumerate}
	\item	Any data collected by the speed cameras must only be used for detecting
			speed violations and any further processing is prohibited.

	\item	Evidence data collected by speed cameras must not be stored permanently
			and must be destroyed immediately if no speed limit violation has been discovered.
			Storage beyond this point is only permitted in the case that a speed limit
			violation has been detected.

	\item	Evidence data relating to the driver's identity must not be revealed
			unless a speed-limit violation is detected.
\end{enumerate}
Specifically, in our model
we consider a system (group $\SpeedControl$)
composed of car entities (group $\Car$)  and the speed-control authority (group $\SCSystem$),
itself
composed of speed-camera entities (group $\trcamera$), the authority responsible
for processing the data (group $\Auth$) and the database of the system (group $\DBase$)
where the personal data of all vehicle owners are appropriately stored.

As far as types are concerned, we assume the existence of two ground types: $\Speed$ referring
to vehicle speed and $\RegNum$ referring to vehicle registration numbers and
three private data types $\CarReg$, $\CarSpeed$ and $\DriverReg$.
Precisely, we consider $\CarReg$ to be the type of private
data pertaining to a car's registration numbers as they exist on vehicles, $\CarSpeed$ to be the
type of private data of the speed of vehicles and $\DriverReg$ to
be the type of private data associating a driver's identity with the
registration numbers of their car as they might exist on a system's
database. Finally, we assume the constant
type $\limit$: constants of type $\limit$ may be compared against other data for
the purpose of checking a speed limit violation.

The system should conform to the following policy:
\begin{longtable}{rcl}
	\CarReg & \haspolicy &\role{\SpeedControl}{\ndiss{\sens}} {\\
		&&	\quad	\lrole{\Car}{\str, \aggr, \diss{\SpeedControl}{\infinite}},\\
		&&	\quad	\role{\SCSystem}{}{\\
		&&	\qquad	\lrole{\trcamera}{\rf, \diss{\SCSystem}{\infinite}},\\
		&&	\qquad	\lrole{\Auth}{\rf, \rd, \aggr,\idf{\DriverReg}},\\
		&&	\qquad	\lrole{\DBase}{}\\ %\diss{\System}{\infinite}, \str, \aggr}{}\\
		&&	\quad }\\
		&&	}
	\\
	\CarSpeed & \haspolicy &\role{\SpeedControl}{\ndiss{\sens}} {\\
		&&	\quad	\lrole{\Car}{\up, \str, \aggr, \diss{\SpeedControl}{\infinite}},\\
		&&	\quad	\role{\SCSystem}{}{\\
		&&	\qquad	\lrole{\trcamera}{\rf, \diss{\SCSystem}{\infinite}},\\
		&&	\qquad	\lrole{\Auth}{\rf, \rd, \aggr, \use{\limit}, \str},\\
		&&	\qquad	\lrole{\DBase}{}\\ %\diss{\System}{\infinite}, \str, \aggr}{}\\
		&&	\quad}\\
		&&	}
	\\
	\DriverReg & \haspolicy &\role{\SpeedControl}{\ndiss{\sens}}{\\
		&&	\quad	\lrole{\Car}{}\\%{\str, \aggr, \diss{}{\infinite}}{},\\
		&&	\quad	\role{\System}{}{\\
		&&	\qquad	\lrole{\trcamera}{},\\
		&&	\qquad	\lrole{\Auth}{\rf, \rd, \idp},\\
		&&	\qquad	\lrole{\DBase}{\rf, \diss{\System}{\infinite}, \str}\\
		&&	\quad}\\
		&&	}
\end{longtable}
%
%\[
%\begin{array}{rcl}
%	\CarReg &\gg &\role{\SpeedControl}{\ndiss{\sens}} {\\
%		&&	\quad \role{\Car}{\str, \aggr, \diss{\SpeedControl}{\infinite}}{},\\
%		&&	\quad \role{\SCSystem}{}{\\
%		&&	\qquad \role{\trcamera}{\rf, \diss{\SCSystem}{\infinite}}{},\\
%		&&	\qquad \role{\Auth}{\rf, \rd, \aggr,\idf{\DriverReg}}{},\\
%		&&	\qquad \role{\DBase}{}{}\\ %\diss{\System}{\infinite}, \str, \aggr}{}\\
%		&&	\quad }\\
%		&&	}
%	\\
%	\CarSpeed &\gg &\role{\SpeedControl}{\ndiss{\sens}} {\\
%		&&	\quad \role{\Car}{\up, \str, \aggr, \diss{\SpeedControl}{\infinite}}{},\\
%		&&	\quad \role{\SCSystem}{}{\\
%		&&	\qquad \role{\trcamera}{\rf, \diss{\SCSystem}{\infinite}}{},\\
%		&&	\qquad \role{\Auth}{\rf, \rd, \aggr, \use{\limit}, \str}{},\\
%		&&	\qquad \role{\DBase}{}{}\\ %\diss{\System}{\infinite}, \str, \aggr}{}\\
%		&&	\quad}\\
%		&&	}
%	\\
%	\DriverReg & \gg &\role{\SpeedControl}{\ndiss{\sens}}{\\
%		&&	\quad	\role{\Car}{}{}\\%{\str, \aggr, \diss{}{\infinite}}{},\\
%		&&	\quad\role{\System}{}{\\
%		&&	\qquad \role{\trcamera}{}{},\\
%		&&	\qquad \role{\Auth}{\rf, \rd, \idp}{},\\
%		&&	\qquad \role{\DBase}{\rf, \diss{\System}{\infinite}, \str}{}\\
%		&&	\quad}\\
%		&&	}
%\end{array}
%\]
%
According to the policy, data of type \CarReg, corresponding to
the registration number plate of a car, exist (are stored)  in a car in public sight
and thus can be disseminated within the speed-control system by the car
for an unlimited number of times. A traffic camera, \trcamera, may thus
gain access to a reference of this data (by taking a photo) and disseminate such
a \CarReg reference
inside the speed control \SCSystem.
The authority \Auth may then receive this reference/photo and by various
means read the actual $\CarReg$ data, which however remains anonymised
unless the authority performs a check for the identification of the owner
against other $\CarReg$ data.
A database \DBase does not have any permissions on data of type \CarReg.

Data of type \CarSpeed correspond to the speed
of a \Car and can be stored in a \Car and updated during computation. Via a
speed radar, \CarSpeed can be disseminated.
A \trcamera may thus read the speed of a car and disseminate it (along
with any other accompanying information such as a photograph of the car) inside
the speed control \SCSystem.
An authority \Auth may read the anonymised
\Car speed and use it for the purpose of
checking speed violation.
A database \DBase does not have any permissions on
data of type \CarSpeed.

Finally, a \DriverReg corresponds to the association
of a registration number and a driver inside a
database. \Car and \trcamera have no permissions
on such data. An authority \Auth may read the data
\DriverReg with their identity. A database
system \DBase stores these data and disseminates
them inside the speed control system.

We model an abstraction of the speed-camera scheme as follows.
 Note that for the purposes of the example we have enriched
our calculus with operator $<$. This extension can
easily be incorporated into our framework.
\begin{eqnarray*}
	{\mathit System} & = &
							\group{\SpeedControl}{
								\;\group{\Car}{C}
								\;\Spar \;
								\group{\SCSystem}{
									\group{\trcamera}{SC} \;\Spar\; \group{\Auth}{A}
								\;\Spar \;
								\group{\DBase}{D}\;}
							}
	\\
     C				& = &
							\newn{r} \newnp{s} {
								\store{r}{\id}{\mathit{reg}} \Par \store{s}{\id}{\mathit{speed}}
								\\
					&\Par&		\repl \inp{cs}{y}\out{s}{\pdata{\id}{y}} \inact\\
					&\Par&		\repl \out{p}{r, s} \inact
							}
	\\
	SC				& = &	\repl \inp{p}{x, y} \out{a}{x, y} \inact
	\\
	A				& = &	\repl \inp{a}{k_1, k_2} \inp{k_2}{\pdata{\hid}{z}} \IfElse{(z > \slimit)}{V}{\inact}
	\\
	V				& = &	\inp{k_1}{\pdata{\hid}{y}}
					\\
					&&			(\inp{r_1}{\pdata{x}{w}} \ifelse{w}{y}{\newnp{r}{\store{r}{x}{z}}}{\inact}
					\\
					&\Par&		\dots
					\\
					&\Par&		\inp{r_n}{\pdata{x}{w}} \ifelse{w}{y}{\newnp{r}{\store{r}{x}{z}}}{\inact})
	\\
%
%	I				& = &	\inp{k_1}{\pdata{\hid}{y}}{\out{b}{y}\inp{c}{x}\newnp{r}{\store{r}{x}{z}}}
%							\\
%
	D				& = &	\store{r_1}{\id_1}{\reg_1} \Par \ldots \Par \store{r_n}{\id_n}{\reg_n}
%	\\
%		& &  \Par \inp{b}{y}(\inp{r_1}{\pdata{x}{z}}{\ifelse{y}{z}{\out{c}{x}\inact}{\inact}}\\
%		& & \hspace{0.5in}\Par\ldots\Par \inp{r_n}{\pdata{x}{z}}{\ifelse{y}{z}{\out{c}{x}\inact}{\inact}})
\end{eqnarray*}

\noindent In system $\mathit{System}$ we have a car process $C$, an
authority process $A$, a speed camera $SC$, and a database $D$
nested within the $\SpeedControl$ group and sharing:
\begin{itemize}
	\item	name $p$
			between the car and the speed camera via which a photo
			of the car (registration number) and its speed are collected
			by the speed camera,

	\item	name $a$ via which this information
			is communicated from the speed camera to the speed-enforcement
			authority, and

	\item	references $r_1, \ldots, r_n$, via which the authority
			queries the database to extract the identity of the driver
            corresponding to the registration number of the car that
			has exceeded the speed limit.

%	\item	name $b$ via which the authority requests
%			the identity of the driver of car with registration number
%			$z$ in case the car has exceeded the speed limit and

%	\item
%			name $c$ via which the database responds to the authority.
\end{itemize}
According to the definition, a car $C$ possesses two stores containing
its registration number and its speed, with the speed of the
car changing dynamically, as modelled by input on channel $cs$.
These two pieces of information can be obtained by the speed
camera via channel $p$. On receipt of this information, the speed
camera forwards the information to the authority $A$.
In turn, authority $A$ may receive such information from the speed
camera and process it as follows: It begins by accessing the
speed of the car. Note that the identity of the car driver is
hidden. It then proceeds to check whether this speed is above
the speed limit. This is achieved by comparing the vehicle
speed with value $\slimit$ which captures unacceptable speed
values. In case a violation is detected, the authority proceeds
according to process $V$, and the authority communicates with
the database in order to find out the identity
of the driver associated with the speed limit violation.
The identification is performed by matching the anonymised registration
number of the vehicle against the records that are received by the database.
%by forwarding the associated registration numbers
%and once it receives this value, it store the identity along
%with the proof obtained of the speed violation in order for
%further action to be taken.
Finally, the database, stores all information about the registration number of
all the drivers.
% identifications, and when requested
%to extract the driver's identity relating to a specific car
%registration number, it searches the database to produce
%the required results.
%are

%Assume that:
%\begin{eqnarray*}
%	\rec{X}{P} &=& \newnp{a}{\repl \inp{a}{x} P \subst{\out{a}{\mathsf{1}} \inact}{\varp{X}} \Par \out{a}{\mathsf{1}} \inact} \qquad x \notin \fn{P}\\
%\end{eqnarray*}

Let us write $T_{cr} = \privatet{\CarReg}{\RegNum}$,
$T_{sp} = \privatet{\CarSpeed}{\Speed}$, $T_{dr} = \privatet{\DriverReg}{\RegNum}$,
$T_{cr}^{S} =\gtype{\SpeedControl}{T_{cr}}$, $T_{sp}^{S} =\gtype{\SpeedControl}{T_{sp}}$,
$T_{dr}^{S} = \gtype{\SCSystem}{T_{dr}}$.
By applying the rules of our type system we may show that $\Gamma;\es \proves \System \hastype \Theta$,
where:
\small
\begin{eqnarray*}
	\Gamma	&=&	\set{
				\slimit	:	\purposet{\limit}{\Speed},
%				r:	T_{cr}^{S},
%	            s:	T_{sp}^{S}, cs:\gtype{\Car}{\Speed},
			\\
             && p,a:\gtype{\SpeedControl}{T_{cr},T_{sp}},
				r_1,\ldots,r_n:T_{dr}^{S}}
	\\
	\Theta	& = &	\CarReg	: \hiert{\SpeedControl}{\hiert{\Car}{\set{\str,\aggr,\diss{\SpeedControl}{\infinite}}}}, \\
			&&		\CarReg	: \hiert{\SpeedControl}{\hiert{\SCSystem}{\hiert{\trcamera}{\set{\rf,\diss{\SpeedControl}{\infinite}}}}},\\
			&&		\CarReg	: \hiert{\SpeedControl}{\hiert{\SCSystem}{\hiert{\Auth}{\set{\rf, \rd, \aggr,\idf{\DriverReg}}}}},\\
			&&		\CarReg	: \hiert{\SpeedControl}{\hiert{\SCSystem}{\hiert{\DBase}{\set{}}}}\\
            &&  	\Speed	: \hiert{\SpeedControl}{\hiert{\Car}{\set{\up,\str,\aggr,\diss{\SpeedControl}{\infinite}}}}, \\
			&&		\Speed	: \hiert{\SpeedControl}{\hiert{\SCSystem}{\hiert{\trcamera}{\set{\rf,\diss{\SpeedControl}{\infinite}}}}},\\
			&&		\Speed	: \hiert{\SpeedControl}{\hiert{\SCSystem}{\hiert{\Auth}{\set{\rf, \rd, \aggr,\use{\limit}}}}},\\
			&&		\Speed	: \hiert{\SpeedControl}{\hiert{\SCSystem}{\hiert{\DBase}{\set{}}}}\\
            &&  	\DriverReg	: \hiert{\SpeedControl}{\hiert{\Car}{}}, \\
			&&		\DriverReg	: \hiert{\SpeedControl}{\hiert{\SCSystem}{\hiert{\trcamera}{}}},\\
			&&		\DriverReg	: \hiert{\SpeedControl}{\hiert{\SCSystem}{\hiert{\Auth}{\set{\rf, \rd, \idp}}}},\\
			&&		\DriverReg	: \hiert{\SpeedControl}{\hiert{\SCSystem}{\hiert{\DBase}{\set{\str, \diss{\SCSystem}{\infinite}}}}}
\end{eqnarray*}
\normalsize
We may prove that $\Theta$ is compatible with the enunciated policy; therefore, the policy is satisfied by the system.
\begin{comment}
\qquad \role{\trcamera}{\rf, \diss{\SCSystem}{\infinite}}{},\\
		&&	\qquad \role{\Auth}{\rf, \rd, \aggr,\idf{\DriverReg}}{},\\
		&&	\qquad \role{\DBase}{}{}\\

Types
\begin{eqnarray*}
	\slimit	&:&	\purposet{\limit}{\Speed}\\
	\pdata{\id}{\reg}, \pdata{\hid}{y} &:& \privatet{\CarReg}{\RegNum}\\
	\pdata{\id}{\speed}, \pdata{\hid}{z} &:& \privatet{\CarSpeed}{\Speed}\\
%
	r, k_1		&:&	\gtype{\SpeedControl}{\privatet{\CarReg}{\RegNum}}\\
	s, k_2		&:&	\gtype{\SpeedControl}{\privatet{\CarSpeed}{\Speed}}\\
%
	a, p		&:&	\gtype{\SpeedControl}{\gtype{\SpeedControl}{\privatet{\CarReg}{\RegNum}, \privatet{\CarSpeed}{\Speed}}}\\
%
	\pdata{\id_i}{\reg_i}, \pdata{x}{w} &:& \privatet{\DriverReg}{\RegNum}\\
%	\pdata{\id_i}{\speed_i}, \pdata{x}{z} &:& \privatet{\DriverSpeed}{\Speed}\\
%
	r_i, l		&:&	\gtype{\SCSystem}{\privatet{\DriverReg}{\RegNum}}\\
%	r_i^2, l^2	&:&	\gtype{\System}{\privatet{\DriverSpeed}{\Speed}}\\
%
%	b			&:&	\gtype{\System}{\gtype{\System}{\privatet{\DriverReg}{\RegNum}, \privatet{\DriverSpeed}{\Speed}}}
	\pdata{x}{w}	&:&	\privatet{\DriverReg}{\RegNum}\\
	\pdata{x}{z}	&:&	\privatet{\CarReg}{\Speed}
\end{eqnarray*}
\end{comment}

\section{Related work} \label{sec:related}

There exists a large body of literature concerned with reasoning about privacy.
To begin with, a number of languages have been proposed to express privacy
policies~\cite{P3P,EPAL,MayGL06,GJD11,LiuMX07,NiBL08,ChowdhuryGNRBDJW13}. Some of these languages
are associated with formal semantics and can be used to verify the consistency of
policies or to check whether a system complies with a certain policy. These
verifications may be performed \emph{a priori} via static techniques such as model
checking~\cite{LiuMX07,KoleiniRR13}, on-the-fly using monitoring,
e.g.~\cite{BasinKM10,SokolskySLK06}, or \emph{a posteriori}, e.g. through audit
procedures~\cite{DattaBCDGJKS11,BarthDMN06,DGJKD10}.

Among these studies
we mention work on the notion of Contextual Integrity~\cite{BarthDMN06},
which constitutes a philosophical account of privacy in terms of the transfer of
personal information. Aspects of this framework have been formalised in a logical framework
and were used for specifying privacy regulations
like HIPAA while notions of compliance of policies by systems were considered.
Close to our work is also
the work of Ni \emph{et al}.~\cite{NiBLBKKT10} where a family of models named P-RBAC
(Privacy-aware Role Based Access Control) is presented that extends the
traditional role-based access control to support specification of
complex privacy policies. This model is based on the concepts of roles,
permissions and data types, similar to ours, but it may additionally
specify conditions, purposes and obligations. The
methodology is mostly geared towards expressing policies and checking for conflicts
within policies as opposed to assessing the satisfaction of policies by systems,
which is the goal of our work. Finally, we point out that the notion of
policy hierarchies is closely related to the concept of hierarchical P-RBAC
of~\cite{NiBLBKKT10}. Hierarchical P-RBAC introduces, amongst others, the
notion of role hierarchies often present in extensions to
RBAC. Role hierarchies are a natural means
for structuring roles to reflect an organisation's lines of authority
and responsibility. Nonetheless, all these works focus on privacy-aware
\emph{access control}. Our work extends these approaches by
considering \emph{privacy} as a general notion and addressing a wider
set of  privacy violations such
as identification and aggregation.

Also related to our work is the research line on type-based security in
process calculi. Among these works, numerous studies have focused on access
control which
is closely related to privacy. For instance the work on the D$\pi$ calculus
has introduced sophisticated type systems for controlling the access to resources
advertised at different locations~\cite{HennessyRY05,HennessyR02,HennessyBook}. Furthermore,
discretionary access control has been considered in~\cite{BugliesiCCM09} which
similarly to our work employs the $\pi$-calculus with groups, while role-based
access control (RBAC) has been considered
in~\cite{BraghinGS06,Dezani-CiancagliniGJP10,CompagnoniGB08}. In addition,
authorization policies and their analysis via type checking has been considered
in a number of papers including~\cite{FournetGM07,BackesHM08,BengtsonBFGM11}.
Our work
is similar in spirit to these works: all approaches consider some type of
policy/schema and a type system that ensures that systems comply to their
policy. Futhermore, we mention that a type system for checking differential
privacy for security protocols was developed in~\cite{EignerM13} for enforcing
quantitative privacy properties.

These works, however, differ from the work presented in this paper. To begin
with, our approach departs from these works in that our study
introduces the concepts of an attribute, hierarchies
of disclosure zones and by basing the methodology around the problem of
policy satisfaction. Furthermore, our goal is to provide
foundations for a more general treatment of privacy which departs from
access-control requirements. Inspired by Solove's taxonomy,  we propose
a framework for reasoning about identification, data aggregation and secondary use.
To the best of our knowledge, our work is the first formal study of these notions.
% because it lacks the notion of an
%attribute. While in RBAC it is possible to express that a doctor may read
%patient's data and send emails, it is not possible to detect the privacy violation
%breach executed when the doctor sends an email with the sensitive patient
%data. To control such information dissemination it is necessary to recognize
%the kind of attributes being communicated by actions in a system and to include a means
%of reasoning about these attributes in privacy policies, which has been an
%objective of our work. Also, a novelty of our approach is the concept of
%hierarchies within policies which allow to arrange the system into a
%hierarchical arrangement of disclosure zones while allowing the inheritance
%of permissions between groups within the hierarchy.

%Another line of work related to our work approach in terms of the methodology used
%is~\cite{martins:phd-thesis}, where the author proposes a typing discipline to
%control the migration of code in a distributed $\pi$-calculus.
%However, to the best of our knowledge, there
%exists no prior work on typing systems for checking compliance against privacy
%policies as implemented in our framework.

The Privacy calculus together with the proposed type system
can be used as a basis for future  tool and type checker
implementation. In this arena we can find a number of type checkers
inspired by process calculi. In~\cite{Cremet2006} a featherweight
Scala type-checking model is proposed. The model is able to capture
basic Scala constructs such as nested classes, abstract types, mixin
composition, and path dependent types. In correspondence with the
Privacy calculus we can benefit from nested classes to represent
a Policy/System hierarchy.
A second work that implements a distributed cryptographic asynchronous
$\pi$-calculus can be found in~\cite{so63179} which is designed to
enforce notions of secrecy. The implementation uses an intermediate
language that implements a typed distributed cryptographic $\pi$-calculus.
In our setting the notion of secrecy is related with the existence of groups.
Recently, there has been an effort by the behavioural types community
to implement $\pi$-calculus inspired type-checking
algorithms~\cite{DBLP:conf/fase/HuY16,DBLP:conf/ppdp/KouzapasDPG16,Demangeon2015,SILL:Tutorial}
as a variety of tools operating in different platforms. This line of work
provides evidence that $\pi$-calculus type systems are volatile enough
to provide a basis for mainstream programming.

To conclude, we mention our previous work of~\cite{KouzapasP13,KP15,KokkinoftaP15}. In these works
we again employed the $\pi$-calculus with
groups~\cite{Cardelli:secrecy} accompanied by a type system for capturing privacy-related notions.
The type system of~\cite{KouzapasP13} was based on i/o and linear types for reasoning about
information processing and dissemination and a two-level type structure for distinguishing
between the permissions associated with a name upon receipt and upon delivery by a process.
In~\cite{KP15}, the type system was reconstructed and simplified:  the notion
of groups was employed to distinguish between the different entities of a system and we employ the
type system in order to type-check a system against the standard types of the
$\pi$-calculus with groups while performing type inference to associate permissions
with the different components of a system.
Furthermore, a
safety criterion was proved for the framework thus providing the necessary machinery for proving
privacy preservation by typing. Finally, in~\cite{KokkinoftaP15} we
extended the results of \cite{KP15} to encompass the notion of a purpose.
In the present paper, we extend these works  providing
a more thorough treatment of privacy violations, including data identification, aggregation
and secondary use. To achieve this, it was necessary to extend our policy
language to include a wider range of privacy requirements and also to enrich both the
underlying calculus as well as the associated type system
in order to capture privacy-related concepts in a more satisfactory manner. Furthermore,
the current paper contains the complete exposition of the methodology, including
the proofs of all results.

 \section{Conclusions}
In this paper we have presented a formal
framework based on the $\pi$-calculus with groups for studying privacy.
Our framework is accompanied by
a type system for capturing privacy-related notions and a privacy
language for expressing privacy policies. We have
proved a subject reduction
and a safety theorem for our framework where the latter states that if a system
$Sys$ type checks
against a typing $\Gamma$ and produces a permission interface $\Theta$ which satisfies a policy $\Policy$, then $Sys$
complies to  $\Policy$. Consequently, whenever a system type checks against a typing that is
compatible with a privacy policy we may conclude that the system satisfies the policy.

%In future work, we would like to extend our framework in order
%to encompass reasoning about additional privacy violation instances
%(e.g. identification and aggregation) as well as the possibility of implementing
%the theory developed into a tool. We are also interested in exploring
%a more dynamic setting where the roles held by an agent may evolve over time.

The approach we have proposed is to a large extent an orthogonal
extension of the selected framework, the $\pi$-calculus
with groups. Modelling a system and constructing its related types
is developed by using standard process-calculus techniques without
the need of considering privacy matters, other than the scoping of
groups. Then, the actual treatment of privacy within our framework
takes place at the level of privacy policies against which a system
should comply. The policy language we have proposed is a simple language
that constructs a hierarchical structure of the entities composing
a system and assigning permissions for accessing sensitive data to
each of the entities while allowing to reason about possible
privacy violations.  These permissions are certainly not
intended to capture every possible privacy issue, but
rather to demonstrate a method of how one might formalise
privacy rights. Identifying an appropriate and complete
set of permissions for providing foundations for the notion
of privacy in the general context should be the result of intensive
and probably interdisciplinary research that justifies
each choice.

To this effect,  Solove's taxonomy of privacy
violations
forms a promising context in which these efforts can be based and
it provides various directions for future work.
One possible extension would be to add semantics both at the
level of our metatheory as well as our policy language to
capture \emph{information-processing}-related privacy violations such
as \emph{distortion} and \emph{insecurity} violations.
Distortion allows for relating false information to
a data subject and insecurity violations take place
when some adversary steals an identity and poses
as the data subject. %A possible way to implement
%identification is to assign data subjects to private
%data inside a policy and then use the type
%system to control how this assignment is handled within
%processes.
%There are also harder questions to be addressed such
%as how we can formally define aggregation of
%data and to distinguish when an adversary has achieved data aggregation
%over a data subject.

As a long-term goal, a possible application of such work
would be to implement type checkers for statically ensuring that
programs do not violate privacy policies.
For such an effort to have merit various aspects need to been
taken into account. To begin with, the machinery employed needs
to be carefully designed, in co-operation with legal
scholars and consultants, in order to guarantee the appropriateness
and completeness of the methodology with respect to
actual violations in the real world. Furthermore, one should
address the question whether privacy-policy type checking
would result in programs that would not create any legal
implications both for the owners and the users
of the program.
Last but not least,
the fact that privacy itself is subjective introduces the notion
of  consensus among communities from different disciplines
on definition(s) of privacy policies and
policy compliance.

Other possible directions for extending  our work
can be inspired by existing work on privacy within e.g. contextual
integrity and privacy-aware RBAC as well as $k$-anonymity. For instance,
we are currently extending our
work in order to reason about more complex privacy policies that
include \emph{conditional} permissions and the concept of an
\emph{obligation}. Finally, it would be interesting to explore
more dynamic settings where the roles held by an agent may evolve over
time.

	\bibliographystyle{plain}
	\bibliography{bibliography}

%	\pagebreak
	\appendix

\section{Encoding}
\label{app:enc}

For our calculus we propose the syntax and the semantics of a
class of higher-level processes, notably the store process
\store{r}{\id}{c}.
Nevertheless, the syntax and the interaction of the store
process can be fully encoded in terms of the standard
$\pi$-calculus extended with the \emph{selection
and branch} syntax and semantics.

\begin{defi}[Encoding to the $\pi$-calculus]
	\figref{enc} defines the encoding
	from the \Pcalc
	to the $\pi$-calculus.
	The rest of the operators are defined homomorphically.
\end{defi}

%The encoding now becomes:
\begin{figure}[h]%[t]
	\begin{eqnarray*}
		\map{\store{r}{\id}{c}} &\eqdef&
		\newnp{a}{\out{a}{ \pdata{\id}{c} } \inact \Par \inp{a}{ \pdata{x}{y} } \inp{r}{l} \\
		&& \qquad \qquad l \triangleright \left\{
				\begin{array}{rcl}
					\rdl &:& \out{l}{\pdata{x}{y}} \out{a}{\pdata{x}{y}} \inact,\\
					\wrl &:& \inp{l}{\pdata{w}{z}} \\
					&&\If\ \match{w}{\id}\ \Then\ \sel{l}{\ok} \out{a}{\pdata{w}{z}} \inact \\
					&&\Else\ \sel{l}{\fail} \out{a}{\pdata{x}{y}} \inact
				\end{array}
			\right\}
		}
		\\
		\\
		\map{\inp{u}{k} P} &\eqdef& \newnp{a}{{\out{u}{a} \sel{a}{\rdl} \inp{a}{k} \map{P}}} \qquad \quad k \not=x
		\\
		\\
		\map{\out{u}{\pdata{\ii}{\con}}} P &\eqdef& \newnp{a}{\newnp{b}{\out{b}{\pdata{\ii}{\con}} \inact \Par \repl \inp{b}{k} \newnp{e}{\out{u}{e} \sel{e}{\wrl} \out{e}{k}\\
		&& \bra{e}{\ok: \eout{a} \inact, \fail: \out{b}{k} \inact} }} \Par \einp{a} \map{P}}
	\end{eqnarray*}
	\caption{Encoding of the store process from the \Pcalc into the standard $\pi$-calculus}
	\label{fig:enc}
\end{figure}

\figref{enc} presents the encoding of the store process syntax and semantics.
A store process is represented as a recursive process that receives a name
$y$ and subsequently offers the choices of read (\rdl) and write (\wrl) on $y$.
A client of the store process will make a selection between the choices for
read and write.
In the case of $\rdl$, the store simply sends the private data to the client and
continues by recursion to its original state.
In the case of $\wrl$, the store will receive data from the
client and store them, given that the data have the correct \id. If the data
do not have the correct identity the interaction does not take place and the store
continues by recursion to its previous state.

An input on a reference channel $r$ is encoded with the creation of a new name $a$
that is subsequently being sent via
$r$ to the store process. It then sends the $\rdl$ label on channel $a$
and receives from the store the private data value.

An output on a reference channel $r$ is also encoded with the creation of a new name
$a$ that is subsequently sent via channel $r$ to the store process. In contrast
to the input encoding it will send the $\wrl$ label on channel $a$ and then send
the private data to the store. The store will then either reply with a success via label
\ok whereby the process will continue with the continuation \map{P},
or it will reply with \fail whereby the process will use recursion to
continue to its starting state.

The next theorem shows that the encoding enjoys sound and complete operational correspondence.

\begin{thm}[Operational Correspondence]
	Let $P$ be a process constructed on the process terms of the $\pi$-calculus without the store and
	extended with the selection/branch construct.

	\begin{enumerate}[label=(\roman*)]
		\item	If $P \trans{} P'$ then $\map{P} \trans{} \map{P'}$.
		\item	If $\map{P} \trans{} Q$ then either $Q \Trans{} P$, or there exists $P'$ such that $P \trans{} P'$ and $Q \trans{} \map{P'}$.
	\end{enumerate}
\end{thm}

\begin{proof}
	The proof for part (i) is done by induction on the structure of transition
	\trans{}. There are two interesting base cases:

	\begin{itemize}
		\item	Case: $\inp{r}{\pdata{x}{y}} P \Par \store{r}{\id}{c} \trans{} P \subst{\pdata{\id}{c}}{\pdata{x}{y}} \Par \store{r}{\id}{c}$
		\item	Case: $\out{r}{\pdata{\id}{c'}} P \Par \store{r}{\id}{c} \trans{} P \Par \store{r}{\id}{c'}$
	\end{itemize}

	The requirements of part (i) can be easily verified following simple transitions.

	The proof for part (ii) is done by induction on the cases where \map{P} \trans{} Q.
	The interesting cases are the base cases:

	\begin{itemize}
		\item	$\map{\inp{r}{\pdata{x}{y}} P \Par \store{r}{\id}{c}} \trans{} Q$

				We can verify that $Q \Trans{} \map{P \subst{\pdata{\id}{c}}{\pdata{x}{y}} \Par \store{r}{\id}{c}}$
				with simple transitions.

		\item	$\map{\out{r}{\pdata{\id}{c'}} P \Par \store{r}{\id}{c}} \trans{} Q$

				We can verify that $Q \Trans{} \map{\store{r}{\id}{c} \trans{} P \Par \store{r}{\id}{c'}}$
				with simple transitions.

		\item	$\map{\out{r}{\pdata{\id'}{c'}} P \Par \store{r}{\id}{c}} \trans{} Q$ with $\id' \not= \id$.

				We can verify that $Q \Trans{} \map{\out{r}{\pdata{\id'}{c'}} P \Par \store{r}{\id}{c}}$
				with simple transitions.\qedhere
	\end{itemize}
\end{proof}

\end{document}